\documentclass[sigconf]{acmart}
\AtBeginDocument{%
  \providecommand\BibTeX{{%
    \normalfont B\kern-0.5em{\scshape i\kern-0.25em b}\kern-0.8em\TeX}}}

\usepackage{multirow}
\usepackage{caption}
\usepackage[textfont=normalfont]{subcaption}
\usepackage{enumitem}
\usepackage[export]{adjustbox}
\setlist[itemize]{leftmargin=*}
\usepackage{color}
\definecolor{bc}{RGB}{0,0,255}
\definecolor{rc}{RGB}{255, 0, 0}
\usepackage[ruled, vlined, linesnumbered]{algorithm2e} 
\usepackage{algpseudocode}
\usepackage{amsmath}

\usepackage{amssymb}
\usepackage{amsthm}
\usepackage{changepage}
\usepackage{tabularx}
\usepackage{url}

\newtheorem{observation}{Observation}

\newcommand{\hide}[1]{}

\newcommand{\SG}{\mathcal{G}}
\newcommand{\SV}{\mathcal{V}}

\newcommand{\SH}{\mathcal{H}}
\newcommand{\SHP}{\mathcal{H'}}
\newcommand{\SN}{\mathcal{N}}
\newcommand{\SE}{\mathcal{E}}

\newcommand{\tv}{t_{v,d_\theta}}

\newcommand{\SGTv}{\mathcal{G}^{(t_{v,d_\theta})}}
\newcommand{\SGTe}{\mathcal{G}^{(t_{e,d_\theta})}}
\newcommand{\SGT}{\mathcal{G}^{(t)}}
\newcommand{\SVT}{\mathcal{V}^{(t)}}

\newcommand{\SET}{\mathcal{E}^{(t)}}

\newcommand{\SGR}{\mathcal{\tilde{G}}}
\newcommand{\SGRT}{\mathcal{\tilde{G}}^{(t)}}
\newcommand{\SERT}{\mathcal{\tilde{E}}^{(t)}}
\newcommand{\SVRT}{\mathcal{\tilde{V}}^{(t)}}

\newcommand{\SNTi}{m^{(t)}_i}
\newcommand{\SNTone}{m^{(t)}_1}
\newcommand{\SNTtwo}{m^{(t)}_2}
\newcommand{\SNTthirty}{m^{(t)}_{30}}
\newcommand{\SNTj}{m^{(t)}_j}
\newcommand{\STR}{\mathcal{\tilde{T}}}

\newcommand{\ST}{\mathcal{T}}

\newcommand{\bit}{\begin{itemize}}
\newcommand{\eit}{\end{itemize}}
\newcommand{\ben}{\begin{enumerate*}}
\newcommand{\een}{\end{enumerate*}}

\newtheorem{problem}{Problem}

\newcommand{\ul}{\underline}

\newcommand{\f}[1]{\textbf{#1}}
\newcommand{\s}{\underline}

\newtheorem{definition}{Definition}
\newtheorem{theorem}{Theorem}
\newtheorem{lemma}{Lemma}
\newcommand{\DTV}{d^{(t)}(v)}

\newcommand{\DTVi}{d^{(t_i-1)}(v)}
\newcommand{\DTUi}{d^{(t_i-1)}(u)}

\begin{document}
\newcommand{\smallsection}[1]{{\vspace{0.05in} \noindent {\bf{\underline{\smash{#1}}}}}}

\title{\huge Graphlets over Time: A New Lens for Temporal Network Analysis}


\settopmatter{authorsperrow=4}
\author{Deukryeol Yoon}
\affiliation{%
  \institution{KAIST AI}
  \city{Seoul}
  \country{South Korea}
}
\email{deukryeol.yoon@kaist.ac.kr}

\author{Dongjin Lee}
\affiliation{%
  \institution{KAIST EE}
  \city{Daejeon}
  \country{South Korea}
}
\email{dongjin.lee@kaist.ac.kr}

\author{Minyoung Choe}
\affiliation{%
  \institution{KAIST AI}
  \city{Seoul}
  \country{South Korea}
}
\email{minyoung.choe@kaist.ac.kr}

\author{Kijung Shin}
\affiliation{%
  \institution{KAIST AI \& EE}
  \city{Seoul}
  \country{South Korea}
}
\email{kijungs@kaist.ac.kr}





\setlength{\textfloatsep}{0.12cm}
\setlength{\dbltextfloatsep}{0.12cm}
\setlength{\abovecaptionskip}{0.12cm}
\setlength{\skip\footins}{0.12cm}

\settopmatter{printacmref=false} 
\settopmatter{printfolios=false}
\fancyhead{}
    
\begin{abstract}
    Graphs are widely used for modeling various types of interactions, such as email communications and online discussions. Many of such real-world graphs are temporal, and specifically, they grow over time with new nodes and edges.

Counting the instances of each graphlet (i.e., an induced subgraph isomorphism class) has been successful in characterizing local structures of graphs, with many applications. While graphlets have been extended for temporal graphs, the extensions are designed for examining temporally-local subgraphs composed of edges with close arrival times, instead of long-term changes in local structures.

In this paper, as a new lens for temporal graph analysis, we study the evolution of distributions of graphlet instances over time in real-world graphs at three different levels (graphs, nodes, and edges). At the graph level, we first discover that the evolution patterns are significantly different from those in random graphs. Then, we suggest a graphlet transition graph for measuring the similarity of the evolution patterns of graphs, and we find out a surprising similarity between the graphs from the same domain. At the node and edge levels, we demonstrate that the local structures around nodes and edges in their early stage provide a strong signal regarding their future importance. In particular, we significantly improve the predictability of the future importance of nodes and edges using the counts of the roles (a.k.a., orbits) that they take within graphlets.
\end{abstract}

\maketitle

\begin{figure}[t]
     \captionsetup[subfigure]{justification=centering}
    \begin{subfigure}{0.23\textwidth}
         \includegraphics[width=\textwidth]{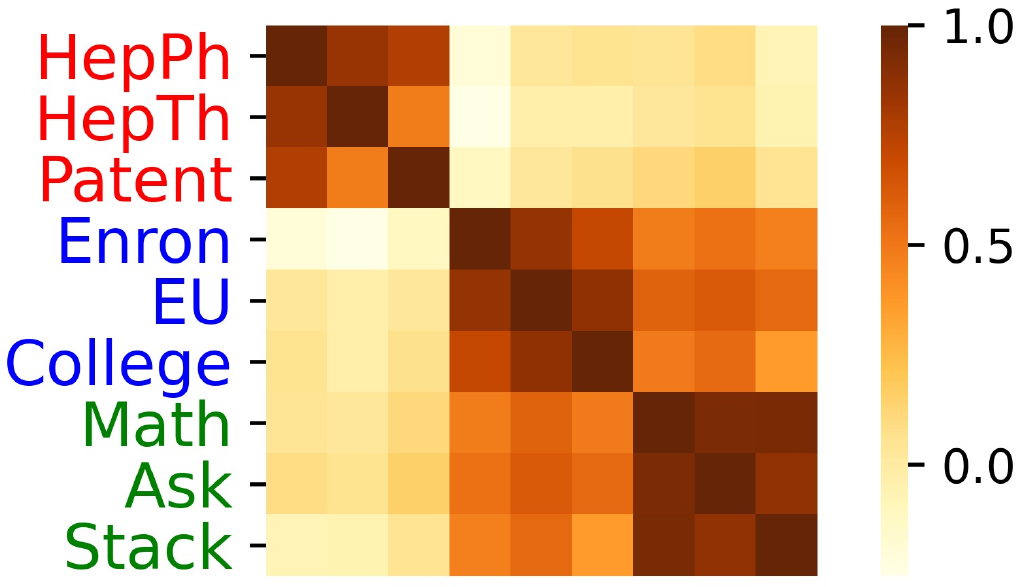}
         \caption{Similarity between graphs w.r.t. \textbf{graphlet transitions} 
         \newline(classification accuracy = \textbf{97.2\%})}
     \end{subfigure}
     \begin{subfigure}{0.23\textwidth}
        \includegraphics[width=\textwidth]{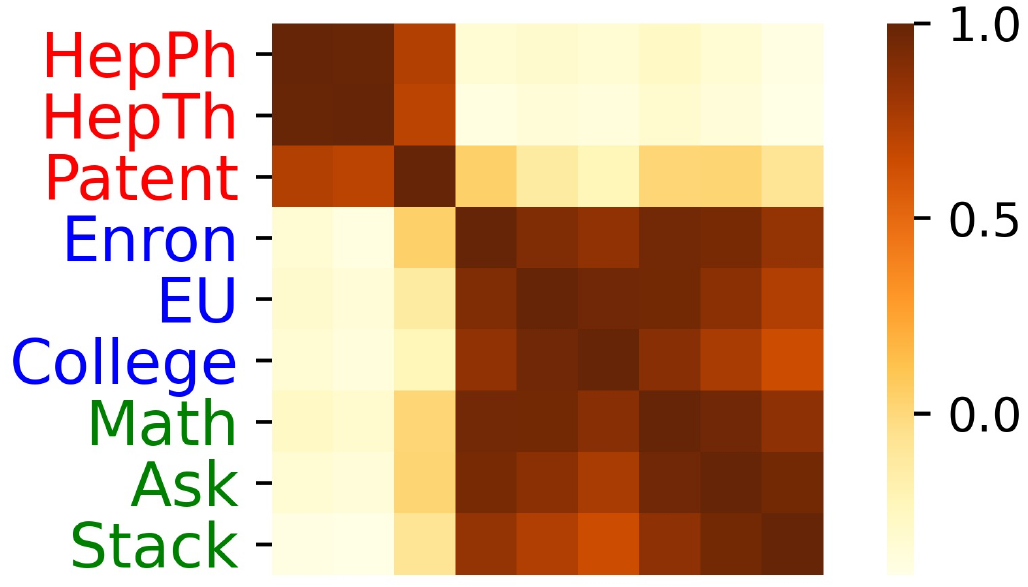}
        \caption{Similarity between graphs w.r.t. \textbf{graphlet occurrences}
        (classification accuracy = \textbf{83.3}\%)}
     \end{subfigure}
     
      \caption{\label{fig:domain_correlation} 
      Real-world temporal graphs from the same domain share similar evolution patterns captured by transitions between graphlets. The figures show the pairwise similarity between 9 graphs from 3 domains (distinguished by text colors) with respect to the transitions between graphlets (see (a)) and the occurrences of graphlets (see (b)). The domains of graphs can be classified more accurately in (a) than in (b). Specifically, with the best thresholds of similarity, the classification accuracy is $97.2\%$ in (a) and $83.3\%$ in (b). See Section~\ref{section:graph:transition} for details about the similarity measures.
      }
\end{figure}

\section{Introduction}
\label{sec:intro}
Graphs are a simple yet powerful tool, and thus they have been used for representing various types of 
interactions: email communications, online Q/As, research collaborations, to name a few.
Due to newly formed interactions, such real-world graphs are \textit{temporal}, i.e., they evolve over time with new nodes and edges.
Many studies have examined the dynamics of real-world temporal graphs and revealed interesting patterns, including densification \cite{leskovec2005graphs}, shrinking diameter \cite{leskovec2005graphs}, and temporal locality in triangle formation \cite{lee2020temporal}.


Graphlets have been widely employed for analyzing local structures of graphs.
\textit{Graphlets} \cite{prvzulj2007biological} are defined as the sets of isomorphic small subgraphs with a predefined number of nodes.
Specifically, the relative counts of the instances of different graphlets effectively characterize the local structures of graphs, with successful applications in graph classification \cite{milo2002network,milo2004superfamilies}, community detection \cite{arenas2008motif,benson2016higher,tsourakakis2017scalable}, anomaly detection \cite{juszczyszyn2011motif}, and node embedding \cite{liu2021motif,lee2019graph,yu2019rum}.


As temporal graphs are pervasive, the concept of graphlets has been generalized in a number of ways for temporal graph analysis. 
\textit{Temporal network motifs} \cite{paranjape2017motifs, kovanen2011temporal} are sets of temporal subgraphs that are (a) identical not just topologically but also temporally, 
(b) composed of a fixed number of nodes, and (c) temporally local, i.e., composed of edges whose arrival times are close enough (see Section~\ref{section:relwork} for details).
Due to the last condition, they are suitable for analyzing short-term changes of graphs but not for long-term changes in local structures, which are the focus of this paper.  

In this paper, we examine the long-term evolution of local structures captured by graphlets, as a new lens for temporal graph analysis, in nine real-world temporal graphs from three different domains. Our analysis is at three levels: graphs, nodes, and edges. 


At the graph level, we first investigate the changes in the distributions of graphlet instances over time. We find out that the evolution patterns are distinguished from those in randomized graphs that are obtained by randomly shuffling edges. Moreover, the evolution patterns in graphs from the same domain share some common characteristics.
In order to compare the evolution patterns in a systematic way, we introduce \textit{graphlet transition graphs}, which encode transitions between graphlets due to changes in graphs.
As shown in Figure~\ref{fig:domain_correlation}(a), graphs from the same domain share similar graphlet-transition patterns, which facilitates accurate graph classification, although the sizes of the graphs vary.



\begin{figure*}[t]
    \centering
    \begin{subfigure}{0.19\textwidth}\captionsetup{justification=centering} 
        \includegraphics[height=5cm]{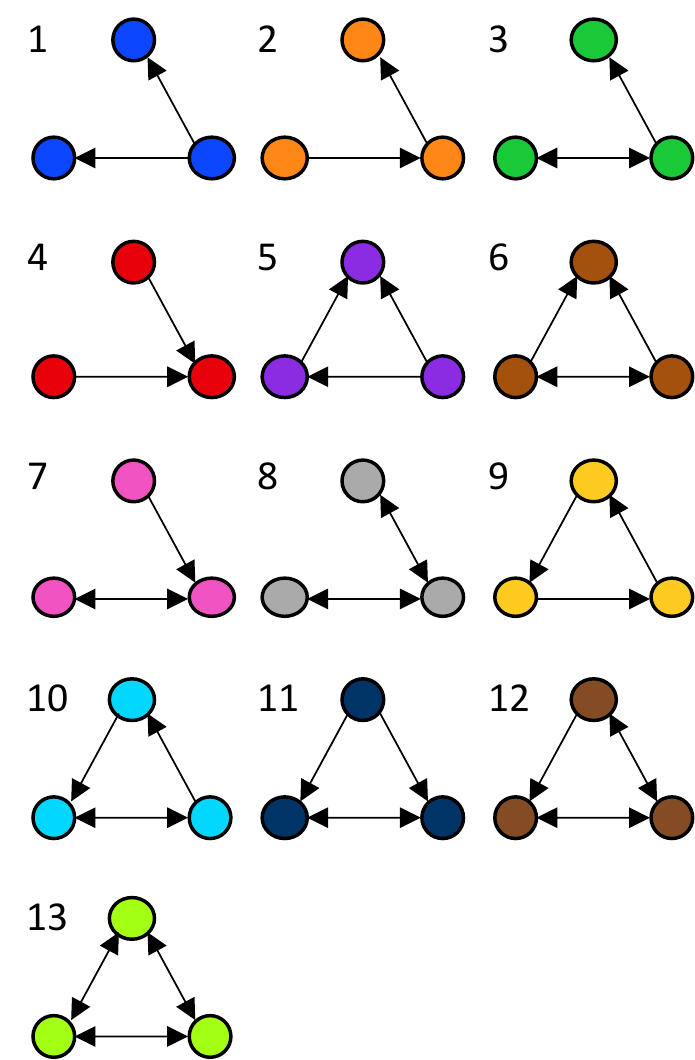}
        \caption{$13$ graphlets}
     \end{subfigure}
     \hspace{1mm} 
     \begin{subfigure}{0.39\textwidth}
        \includegraphics[height=5cm]{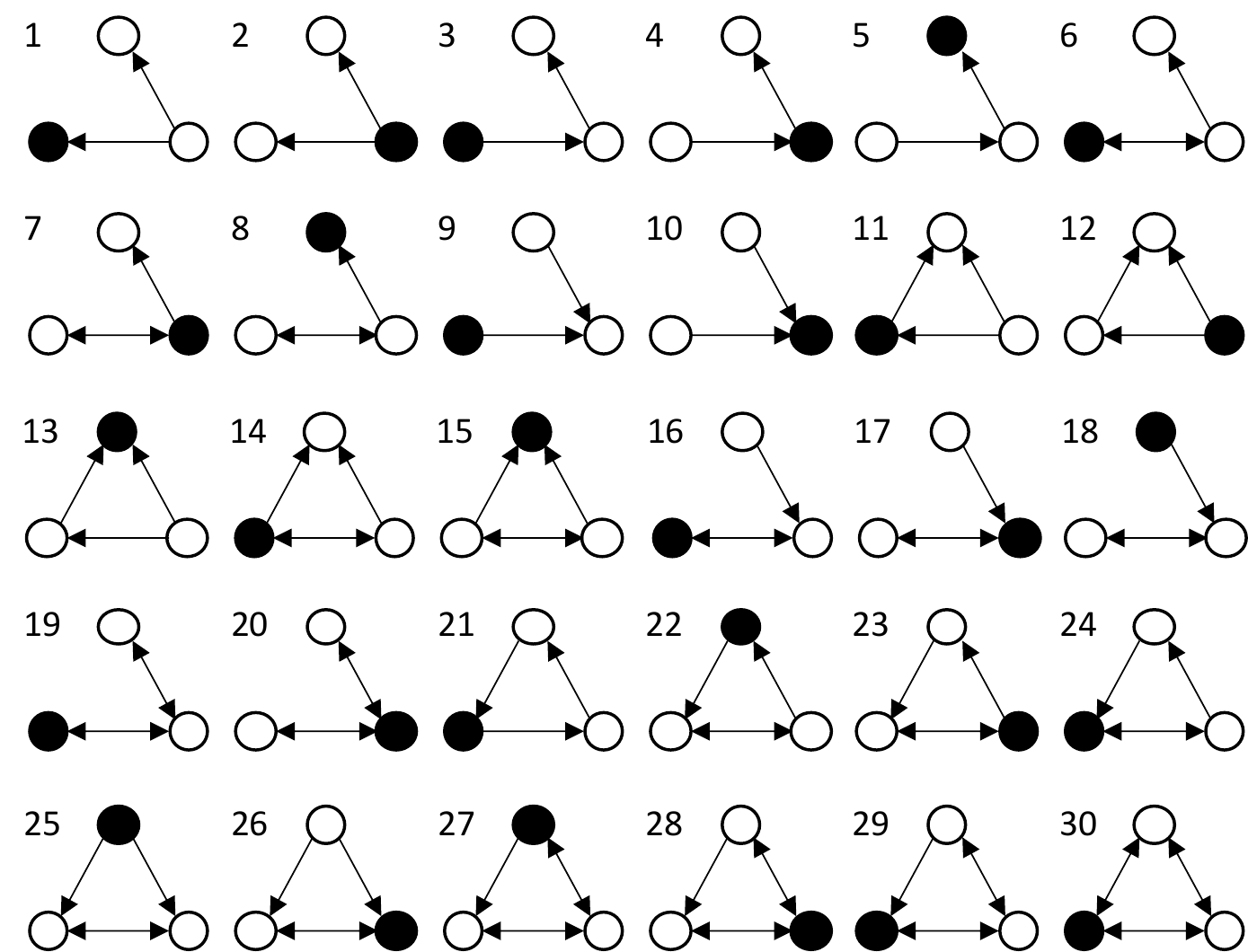}
        \caption{$30$ node roles (also known as, node orbits)}
     \end{subfigure}
     \begin{subfigure}{0.39\textwidth}
        \includegraphics[height=5cm]{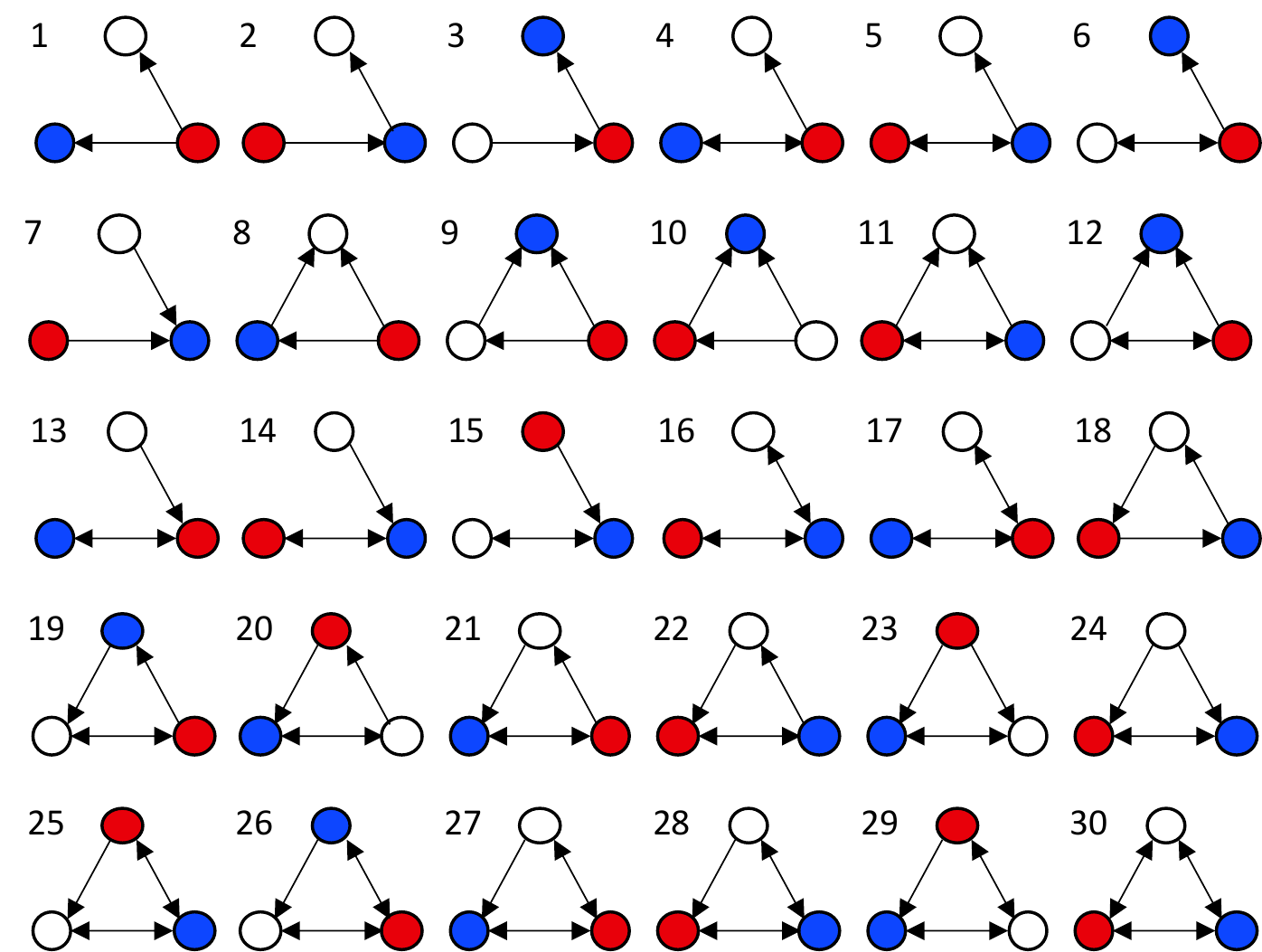}
        \caption{$30$ edge roles (also known as, edge orbits) }
     \end{subfigure}
      \caption{\label{fig:graphlet_and_role} (a) The 13 graphlets \cite{prvzulj2007biological} with three nodes. (b) The 30 node roles \cite{prvzulj2007biological} within the graphlets (see the positions of black nodes). (c) The 30 edge roles within the graphlets (see the positions of edges from a red node to a blue node).
      }
\end{figure*}

At the node and edge levels, we investigate how local structures around each node and edge in their early stage signal their future importance.
Specifically, as local structures, we consider \textit{node roles} (formally, node automorphism orbits \cite{prvzulj2007biological}) and \textit{edge roles} \cite{hovcevar2016computation}, which are roughly sets of symmetric positions of nodes and edges within graphlets.
We also demonstrate that the counts of the roles taken by each node and edge in their early stage are more informative than previously-used features \cite{yang2014predicting}, and they are complementary to simple global features (e.g., total counts of nodes and edges) for the task of predicting future centralities (specifically, in-degree, betweenness~\cite{freeman1977set}, closeness~\cite{bavelas1950communication}, and PageRank \cite{page1999PageRank}).

We summarize our contributions as follows: 
\bit
    \item \textbf{Patterns:} We make several interesting observations about the temporal evolution of graphlets: a surprising similarity in graphs from the same domain and local-structural signals regarding the future importance of nodes and edges.
    \item \textbf{Tool:} We introduce graphlet transition graphs, which is an effective tool for measuring the similarity of local dynamics in temporal graphs of different sizes.
    \item \textbf{Prediction:} We enhance the prediction accuracy of the future importance of nodes and edges by introducing role-based local features, which are complementary to global features.
\eit
\noindent\textbf{Reproducibility:} The code and the datasets are available at 
\url{https://github.com/deukryeol-yoon/graphlets-over-time}.

In Section~\ref{section:prelim}, we introduce basic concepts, notations, and datasets. In Section~\ref{section:graph}, we present our graph-level analysis. 
In Section~\ref{section:node} and Section~\ref{section:edge}, we present our node-level and edge-level analyses. 
In Section~\ref{section:relwork}, we present a brief survey of related works.
In Section~\ref{section:conclusion}, we conclude our work.
\section{Basic Concepts, Notations, and Data}\label{section:prelim}

In this section, we first introduce some basic concepts and notations. Then, we describe the nine datasets used in this paper.

\subsection{Basic Concepts and Notations}
\label{sec:prelim:concept}

\smallsection{Temporal Graph:} 
A \textit{temporal graph} $\SG = (\SV, \SE, \ST)$ consists of a set of nodes $\SV$, a set of directed edges $\SE:=\{e_1,\cdots,e_{|\SE|}\}$, and a multiset of edge arrival times $\ST:=[t_1,\cdots, t_{|\SE|}]$.
For each directed edge $e_i\in \SE$,  we use $t_i \in \ST$ to denote the arrival time of $e_i$.
We use $u \rightarrow v$ to denote a directed edge from a node $u$ to a node $v$, and the nodes $u$ and $v$ are \textit{adjacent} if $u \rightarrow v$ or $v \rightarrow u$ exists.
From now on, we will use the term \textit{edge} to indicate a directed edge when there is no ambiguity.

\begin{table}[t]
	\centering
	\caption{\label{tab:notations} Table of symbols.}
    \resizebox{\columnwidth}{!}{
	\begin{tabular}{ r | l }
		\toprule
		{\bf Notation} & {\bf Definition} \\
		\midrule
	    $\SG = (\SV, \SE, \ST)$ & temporal graph with nodes $\SV$, edges $\SE$, and times $\ST$  \\
	    $\SGT=(\SVT,\SET)$ & snapshot of $\SG$ at time $t$ \\
	    \midrule
	    $\SGR= (\SV, \SE, \STR)$ & a temporal graph randomized from $\SG$ \\
	    $\SGRT=(\SVRT,\SERT)$ & snapshot of $\SGR$ at time $t$ \\
	    \midrule
	    $\SNTi(v)$ & count of node role $i$ at a node $v$ in $\SGT$ \\
	    
		\bottomrule
	\end{tabular}}
\end{table}


\smallsection{Randomized Graph:} 
A \textit{randomized graph} $\SGR= (\SV, \SE, \STR)$ of $\SG=(\SV,\SE, \ST)$ is obtained by assigning arrival times in $\ST$ to edges in $\SE$ uniformly at random in a one-to-one manner. For each edge $e_i\in \SE$, we use $\tilde{t}_i\in \STR$ to denote the arrival time assigned to it. 

\smallsection{Snapshot:}
We define the \textit{snapshot at time $t$} of $\SG=(\SV, \SE, \ST)$ as $\SGT=(\SVT,\SET)$ where $\SET:=\{e_i\in \SE:t_i \leq t\}$ and $\SVT\subseteq \SV$ is the endpoints of any edge in $\SET$.
That is, $\SGT$ consists of the nodes and edges arriving at time $t$ or earlier. Similarly, the snapshot at time $t$ of $\SGR= (\SV, \SE, \STR)$ is 
$\SGRT=(\SVRT,\SERT)$ where $\SERT:=$ $\{e_i\in \SE:\tilde{t}_i \leq t\}$ and $\SVRT$ is the endpoints of any edge in $\SERT$. 
We define the \textit{neighbors} of a node $v\in\SVT$ in a snapshot $\SGT$ as the nodes adjacent to $v$ in $\SGT$.  
We define the \textit{degree} of a node $v\in\SVT$ in a snapshot $\SGT$, which is denoted by $\DTV$, as the number of directed edges whose endpoints include $v$ in $\SGT$.
We simply use $d(v)$ to denote the degree of the node $v$ in the last snapshot $\SG^{(t_{|\SE|})}$.

\smallsection{Induced Subgraphs:} 
A subgraph of a snapshot $\SGT=(\SVT,\SET)$ is \textit{induced} if and only if it consists of a subset of $\SVT$ and all of the edges connecting pairs of the nodes in the subset. 
Two subgraphs $\SH$ and $\SHP$ are \textit{isomorphic} if there exists a one-to-one mapping $f$ between the nodes of both graphs such that there exists an edge from a node $u$ to a node $v$ in $\SH$ if and only if there exists an edge from the node $f(u)$ to the node $f(v)$ in $\SHP$.

\smallsection{Graphlets:}
A \textit{graphlet} is the set of induced subgraphs that are isomorphic to each other.
In this paper, we limit our attention to the $13$ graphlets consisting of three connected nodes. 
An induced subgraph is called an \textit{instance} of graphlet $k$ if it is isomorphic to the $k$-th graph in Figure~\ref{fig:graphlet_and_role}(a).



\smallsection{Node Roles:} 
Consider an induced subgraph $\SH$ with a node set $\SV'$. 
An \textit{automorphism} of $\SH$ is an isomorphism between $\SH$ and itself. i.e., an automorphism of $\SH$ is a one-to-one mapping between nodes of $\SH$ such that there exists an edge from a node $u$ to a node $v$ in $\SH$ if and only if there exists an edge from the node corresponding to $u$ to the node corresponding to $v$ in $\SH$.  
If denoting the set of automorphisms of $\SH$ by $Aut(\SH)$,
the \textit{automorphism orbit} of a node $u\in \SV'$ is the set $\{y \in \SV' : \exists g \in Aut(\SH) \text{ s.t. } y = g(u)\}$ of nodes \cite{prvzulj2007biological}.
Formally, \textit{node roles} are node automorphism orbits, and roughly, they are sets of symmetric positions of nodes within graphlets.
Figure~\ref{fig:graphlet_and_role}(b) (see the positions of black nodes) shows all $30$ node roles in the $13$ graphlets that we consider.
We say a node $v$ ``takes'' node role $i$ in a graphlet instance if there exists an isomorphism of the graphlet instance and the $i$-th graph in Figure~\ref{fig:graphlet_and_role}(b) that maps $v$ to the black node in the graph.
We define the \textit{count of node role $i$ at a node $v$} as the number of graphlet instances where $v$ takes $i$, and $\SNTi(v)$ denotes the count at a snapshot $\SGT$.

\begin{table}[t]
	\centering
	\caption{\label{tab:data} Summary of nine real-world temporal graphs used throughout this paper.} 
     \resizebox{\columnwidth}{!}{
		\begin{tabular}{c | c | c | c | c }
			\toprule
			{\bf Domain }& {\bf Dataset} & {$|V|$} & {\bf$|E_T|$} & {\bf Period}\\
			\midrule
			\multirow{3}{*}{Citation}
			& \ul{\bf{\smash{HepPh}}}    & $34,565$      & $346,849$     & 9 years \\ 
			& \ul{\bf{\smash{HepTh}}}    & $18,477$      & $136,190$     & 10 years \\
			& \ul{\bf{Patent}}           & $3,774,362$   & $16,512,782$  & 25 years \\
			\midrule
			\multirow{3}{*}{Email/Message}
			& \ul{\bf{Enron}}               & $55,655$      & $209,203$     & 24 years \\
		    & \ul{\bf{EU}}                  & $986$         & $24,929$      & 1.5 years \\
		    & \ul{\bf{\smash{College}}}   & $1,899$       & $20,296$      & 0.5 years \\
		    \midrule
		    \multirow{3}{*}{Online Q/A}
		    & \ul{\bf{Ask}}ubuntu             & $159,316$     & $262,106$     & 6 years \\
		    & \ul{\bf{Math}}overflow          & $24,818$      & $90,489$      & 7 years \\
		    & \ul{\bf{Stack}}overflow         & $2,601,977$   & $16,266,395$  & 8 years \\
			\bottomrule
		\end{tabular}}
\end{table}

\smallsection{Edge Roles:} 
Consider an induced subgraph $\SH$ with an edge set $\SE'$.
Based on the concepts defined above, we define the \textit{edge role} of an edge $u\rightarrow v$ is the set $\{x\rightarrow y \in \SE' : \exists g \in Aut(\SH) \text{ s.t. } x = g(u) \wedge y = g(v)\}$ of edges.
Roughly, edge roles are the sets of symmetric positions of edges within graphlets.
Figure~\ref{fig:graphlet_and_role}(c) (see the positions of edges from a red node to a blue node) shows all $30$ edge roles in the $13$ considered graphlets.
We say an edge $u\rightarrow v$ ``takes'' edge role $j$ in a graphlet instance if there exists an isomorphism of the graphlet instance and the $j$-th graph in Figure~\ref{fig:graphlet_and_role}(c) that maps $u$ and $v$ to the red node and the blue node, respectively, in the graph.
We define the \textit{count of edge role $j$ at an edge $e$} as the number of graphlet instances where $e$ takes $j$.

\subsection{Datasets}\label{section:datasets}
Throughout this paper, we use the nine real-world temporal graphs from the three domains, which are summarized in Table~\ref{tab:data}. 

\smallsection{Citation Graphs:} 
Each node is a paper or a patent. Each directed edge from a node $u$ to a node $v$ means that $u$ cites $v$. 

\smallsection{Email/Message Graphs:}
Each node is a user. Each directed edge from a node $u$ to a node $v$ indicates that $u$ sends $v$ emails (messages). 

\smallsection{Online Q/A Graphs:}
Each node is a user. Each directed edge from a node $u$ to a node $v$ means that $u$ answers $v$'s questions.
\begin{table*}[t]
\begin{center}
\vspace{-2mm}
\caption{\label{tab:graphlet_evolution} 
Ratios of graphlets over time. The colors in the plots are matched with the colors of the graphlets in Figure~\ref{fig:graphlet_and_role}, and the evolution ratio means the fraction of edges added to graphs.
The evolution patterns in real-world graphs vary depending on domains (Observation~\ref{obs:graphlet_evolve}), and they are clearly distinguished from the evolution patterns in randomized graphs (Observation~\ref{obs:graphlet_evolve:random}). 
}
\scalebox{0.95}{
\begin{tabular}{c|ccc|ccc}
    \toprule
    & \multicolumn{3}{c|}{Temporal graph $\SG$} & \multicolumn{3}{c}{Randomized graph $\SGR$} \\
    \hline
    \parbox[t]{2mm}{\multirow{8}{*}{\rotatebox[origin=c]{90}{\ \ \ \ \ Citation}}} &  
        \raisebox{-.9\totalheight}{\includegraphics[width=0.1475\textwidth]{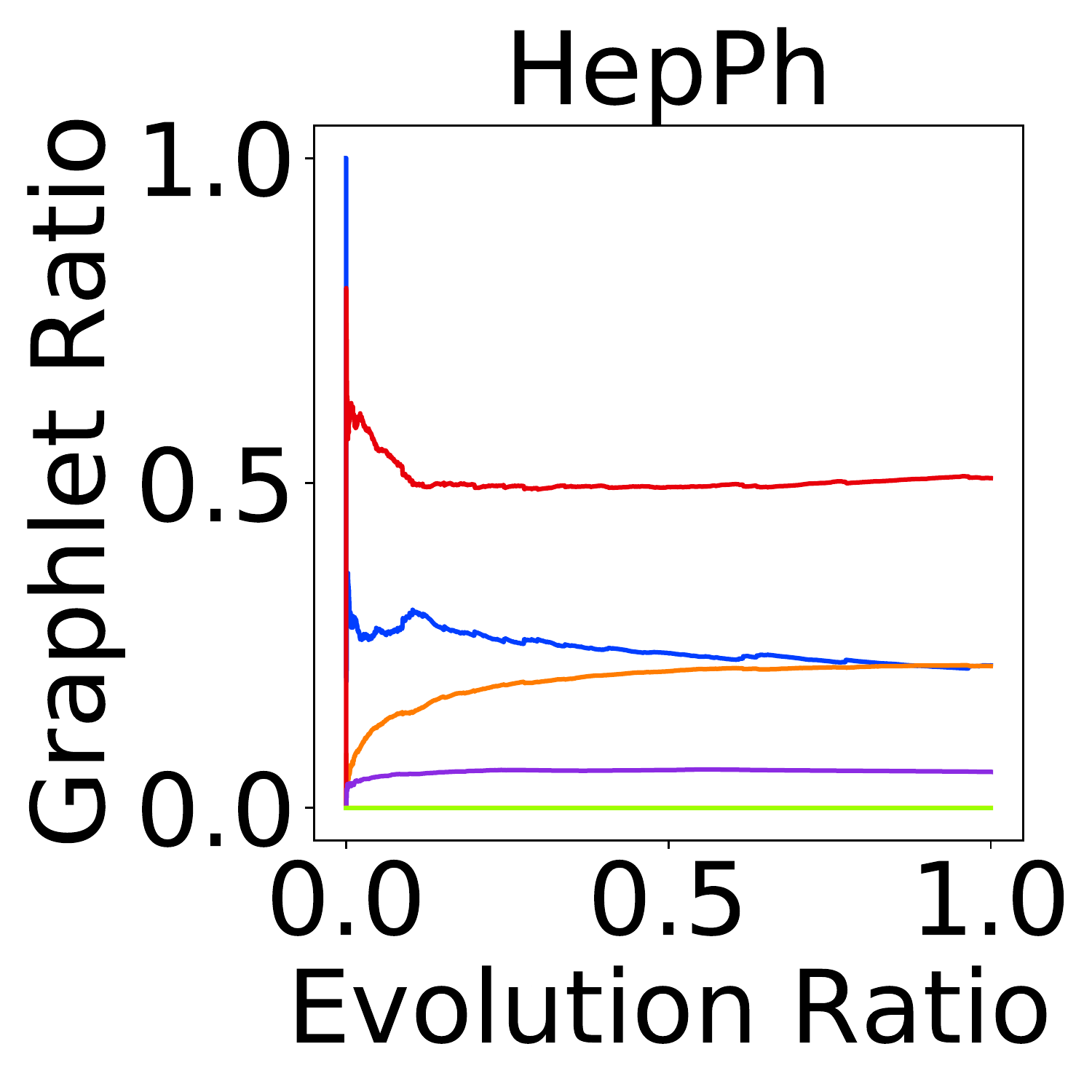}} &
        \raisebox{-.9\totalheight}{\includegraphics[width=0.1475\textwidth]{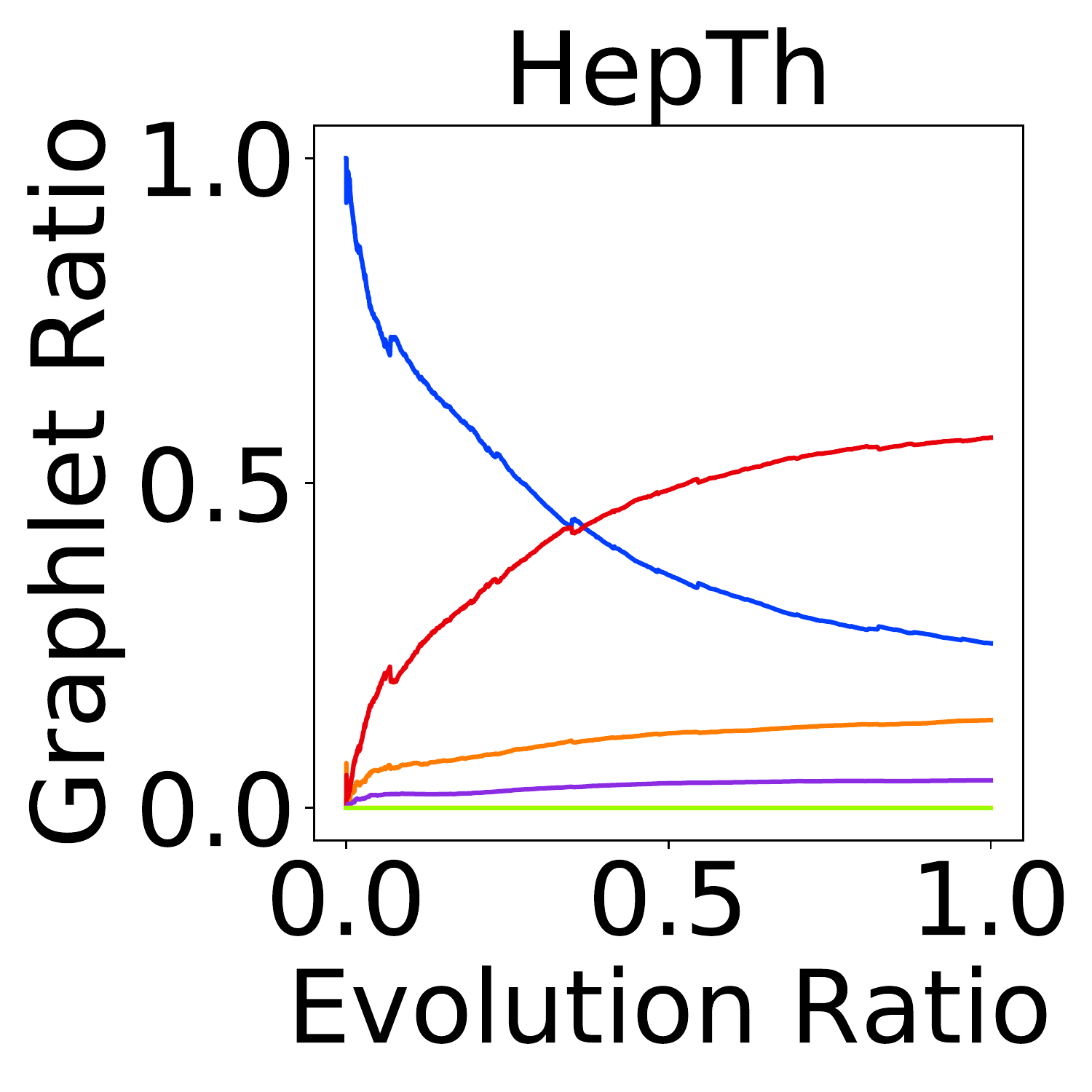}} & 
        \raisebox{-.9\totalheight}{\includegraphics[width=0.1475\textwidth]{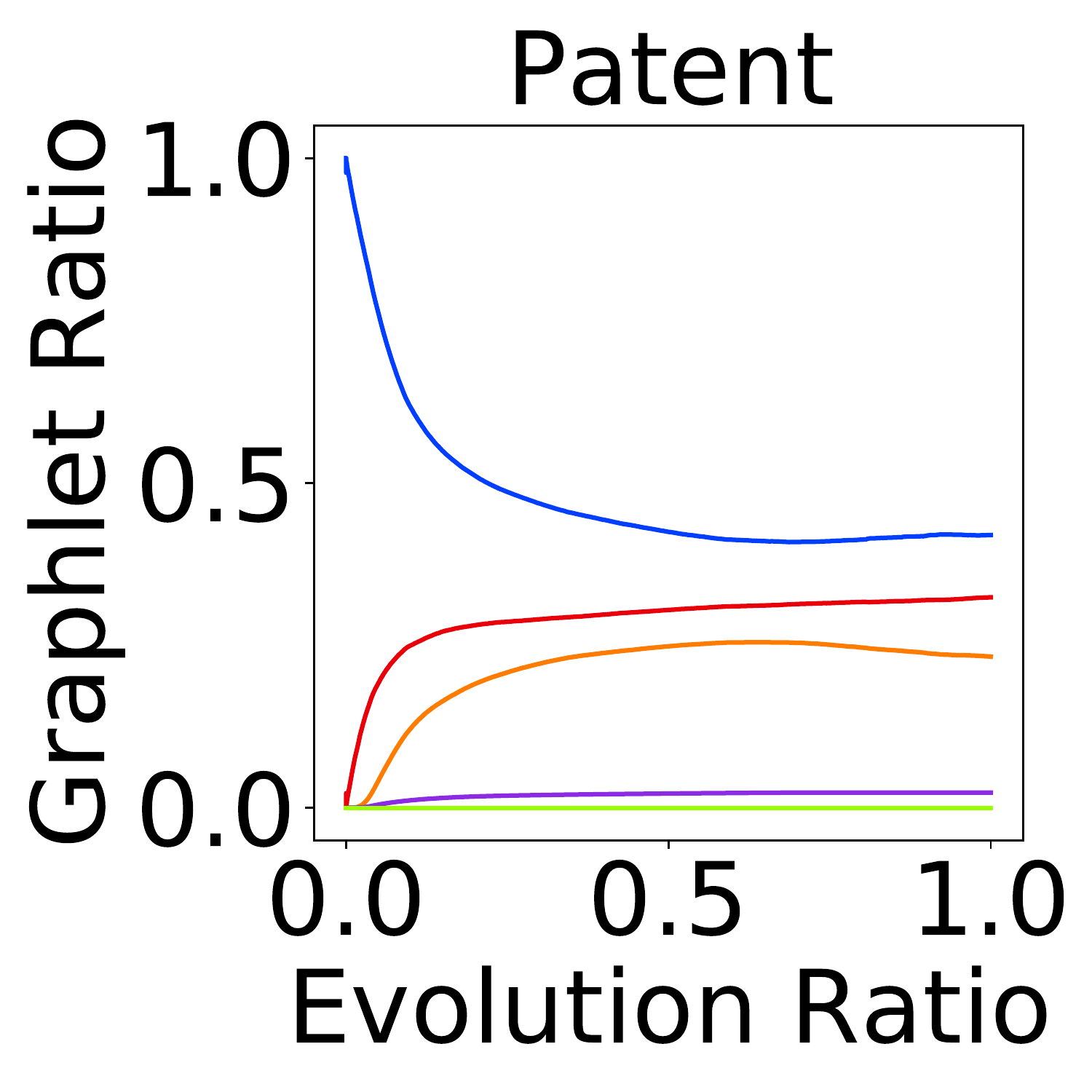}} & 
        \raisebox{-.9\totalheight}{\includegraphics[width=0.1475\textwidth]{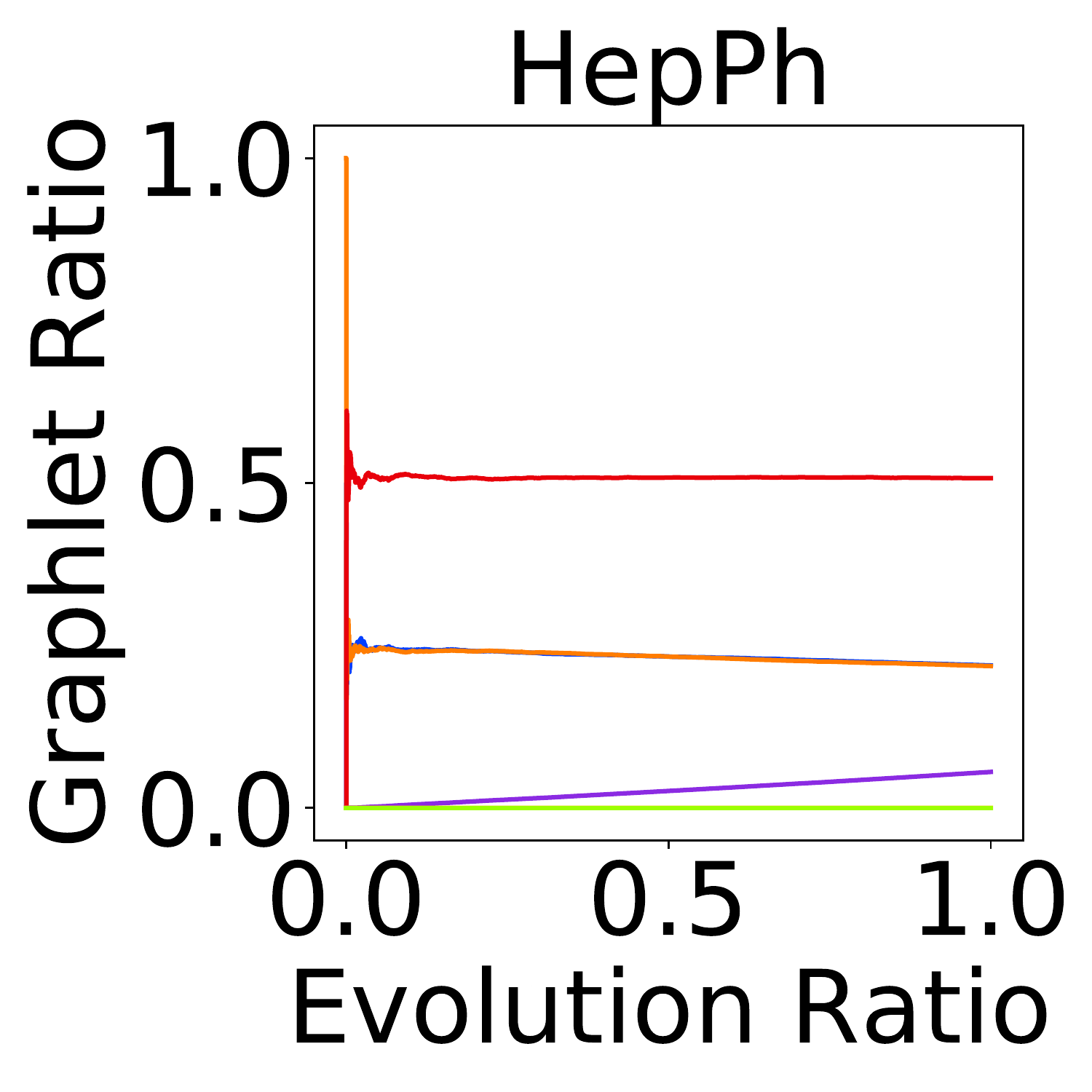}} & \raisebox{-.9\totalheight}{\includegraphics[width=0.1475\textwidth]{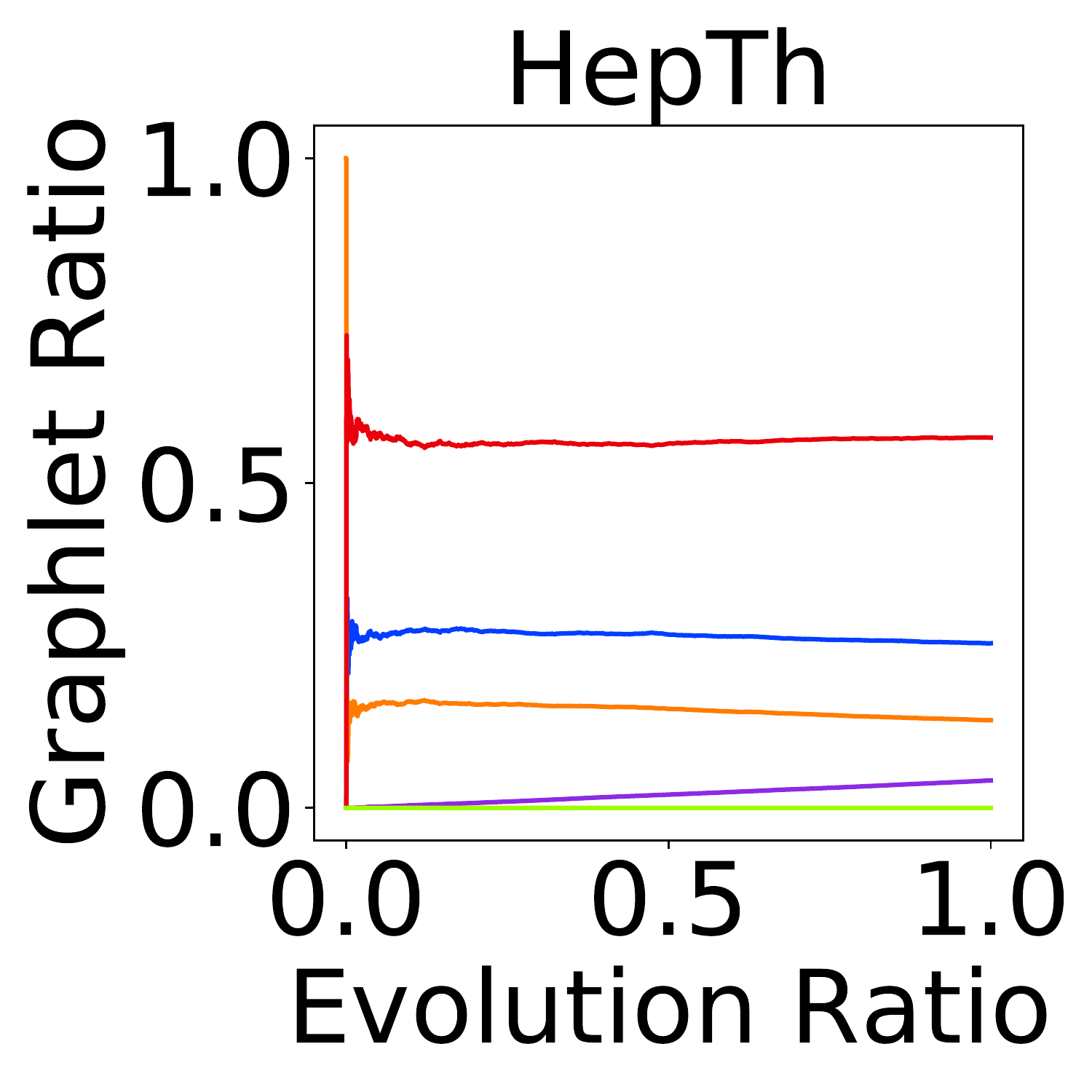}} &
        \raisebox{-.9\totalheight}{\includegraphics[width=0.1475\textwidth]{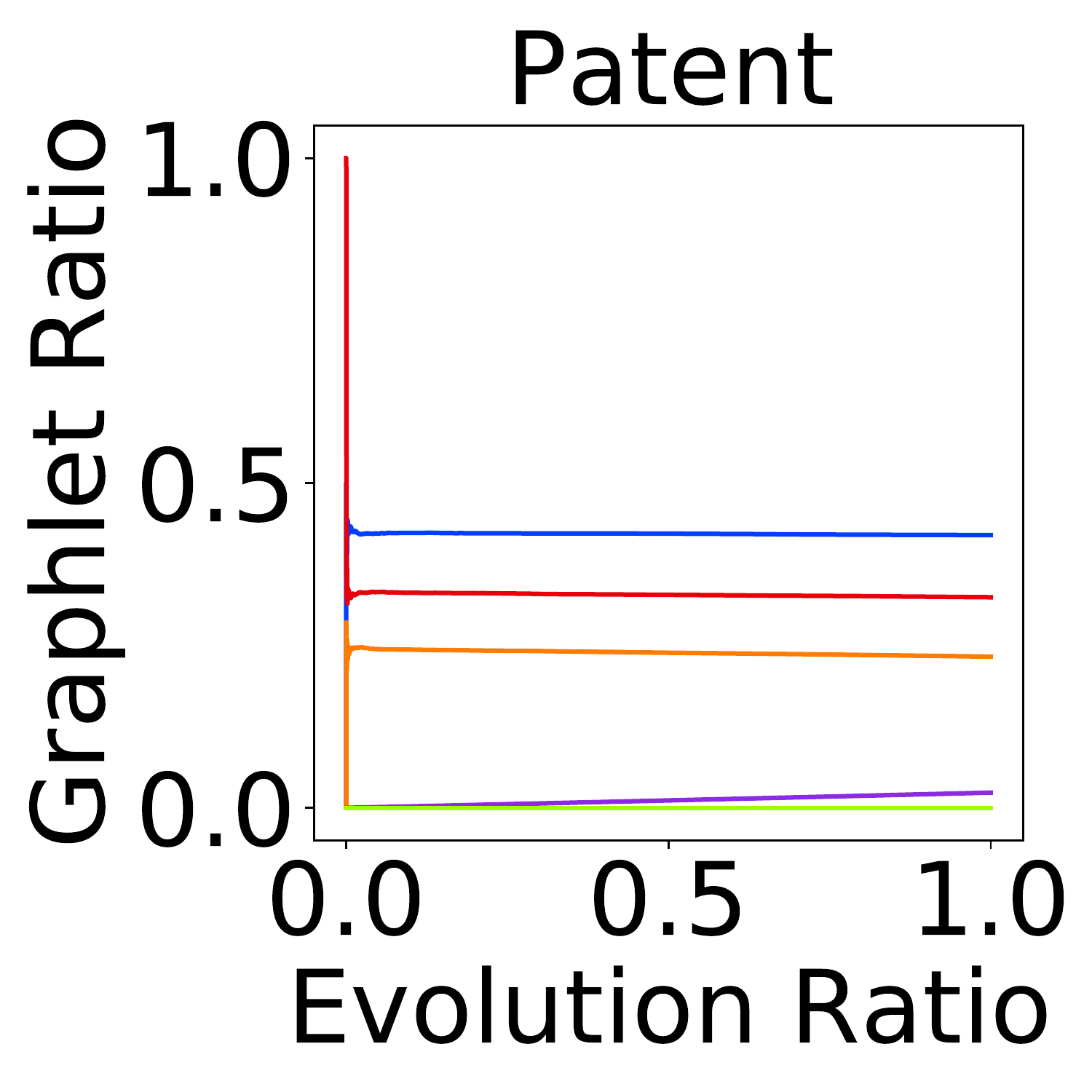}} \\
    \hline
    \parbox[t]{2mm}{\multirow{8}{*}{\rotatebox[origin=c]{90}{\ \ \ \ \ Email/Message}}} &  
        \raisebox{-.9\totalheight}{\includegraphics[width=0.1475\textwidth]{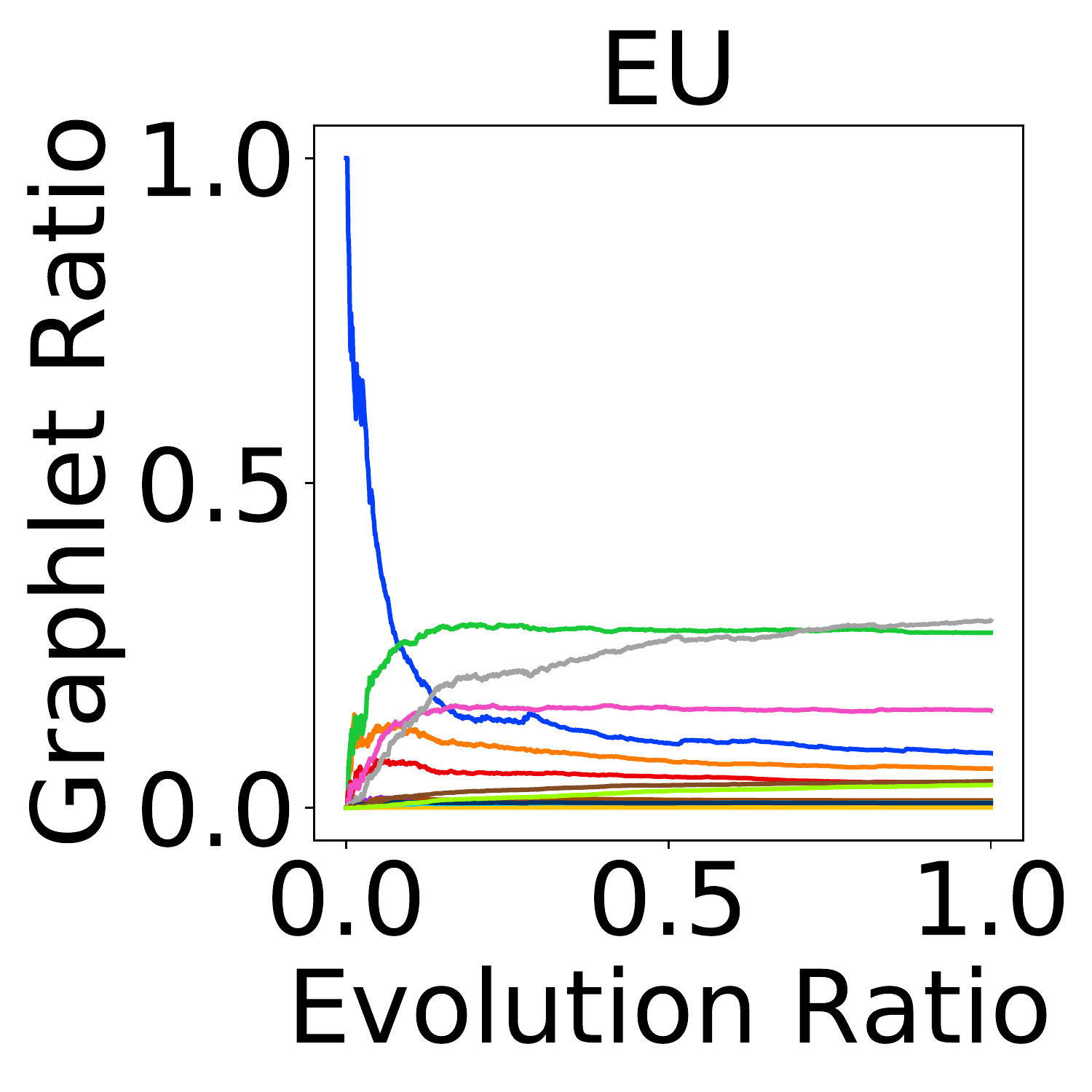}} & \raisebox{-.9\totalheight}{\includegraphics[width=0.1475\textwidth]{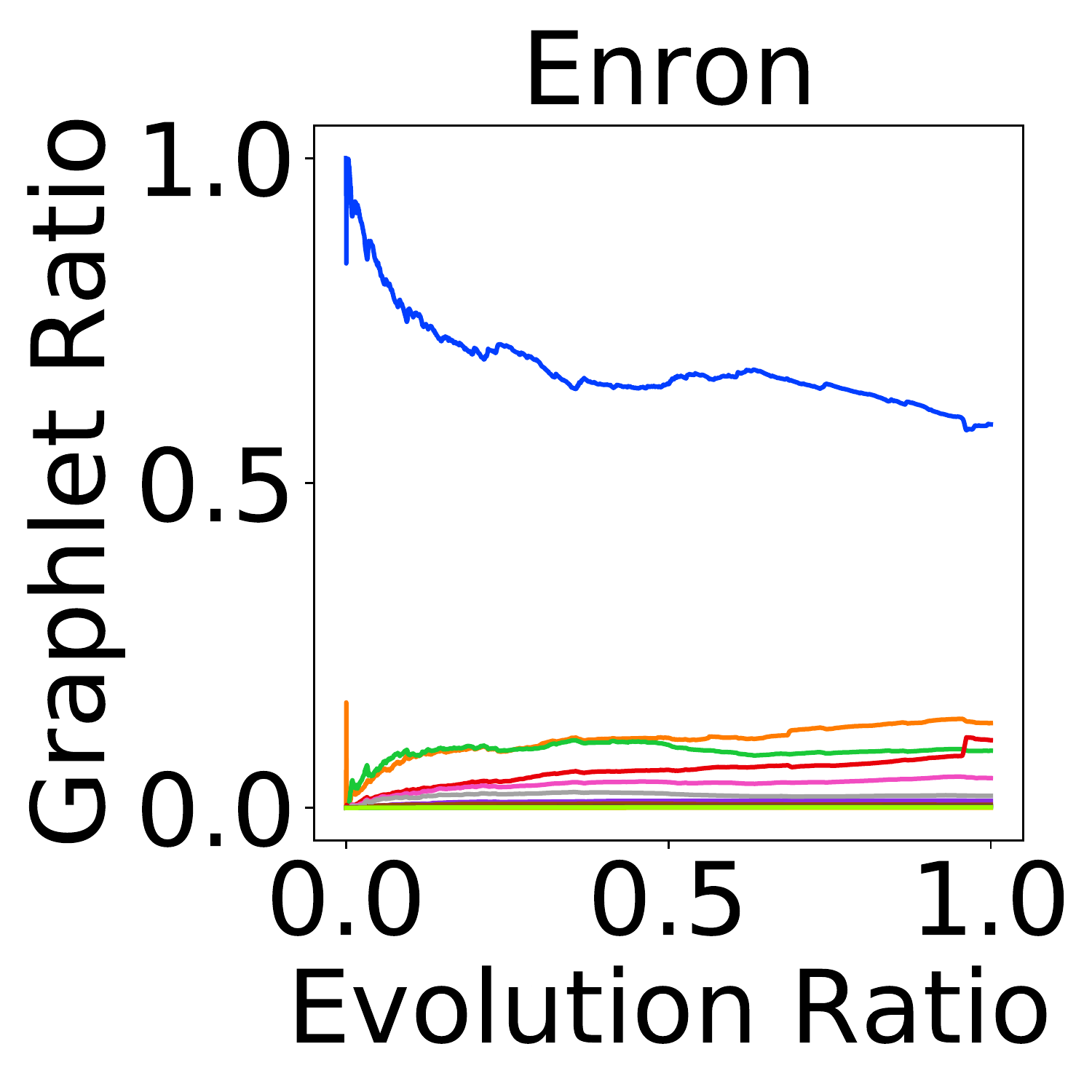}} & 
        \raisebox{-.9\totalheight}{\includegraphics[width=0.1475\textwidth]{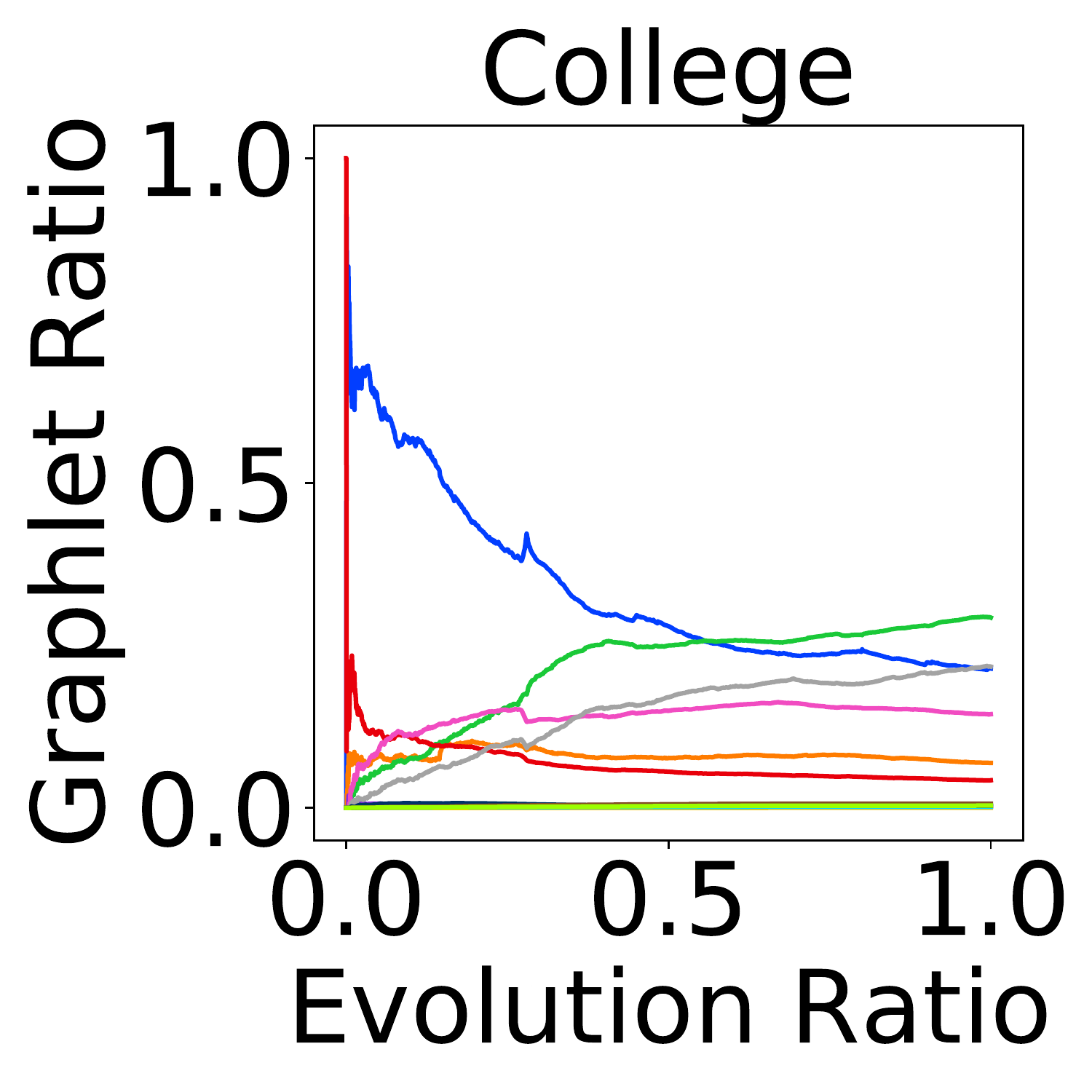}} & 
        \raisebox{-.9\totalheight}{\includegraphics[width=0.1475\textwidth]{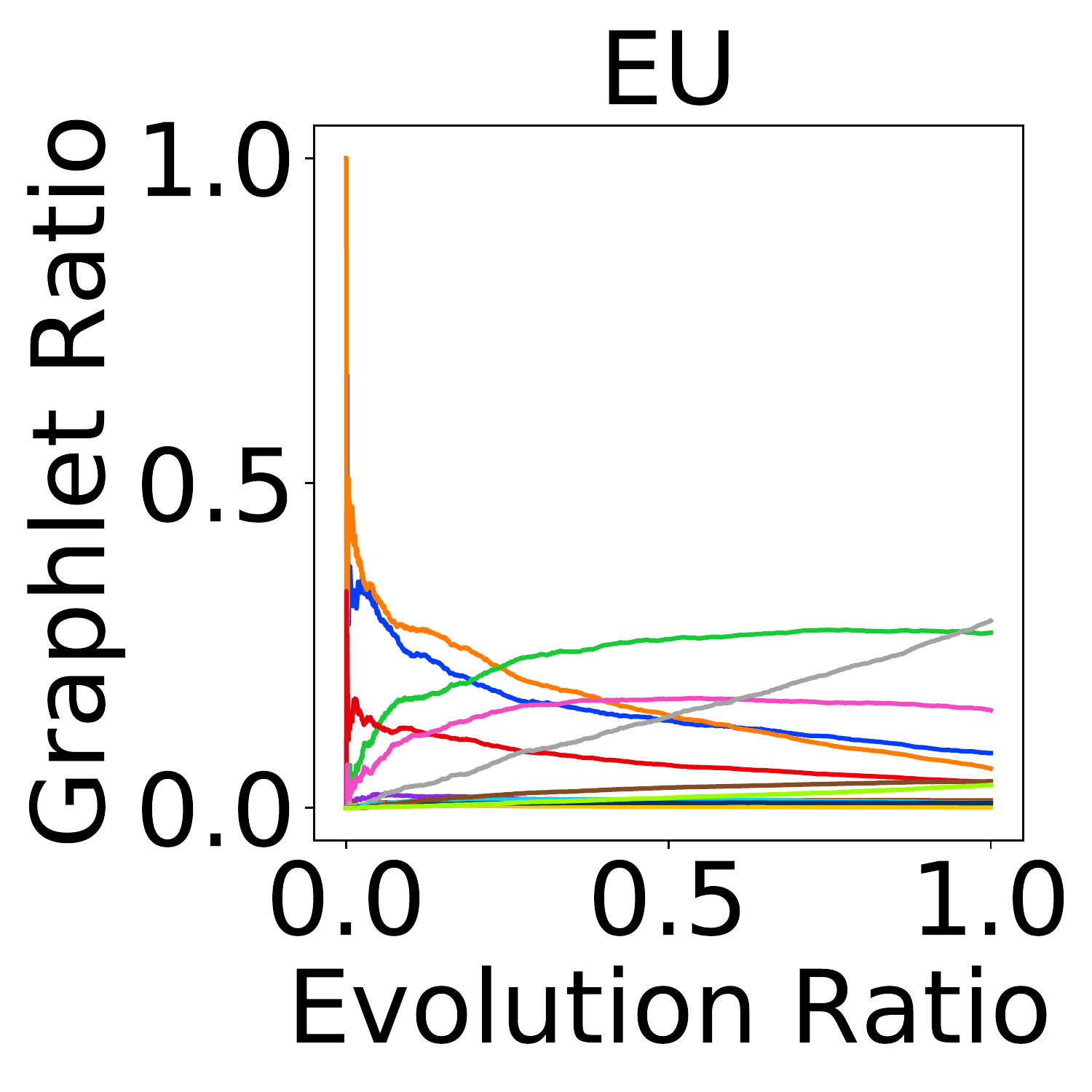}} & \raisebox{-.9\totalheight}{\includegraphics[width=0.1475\textwidth]{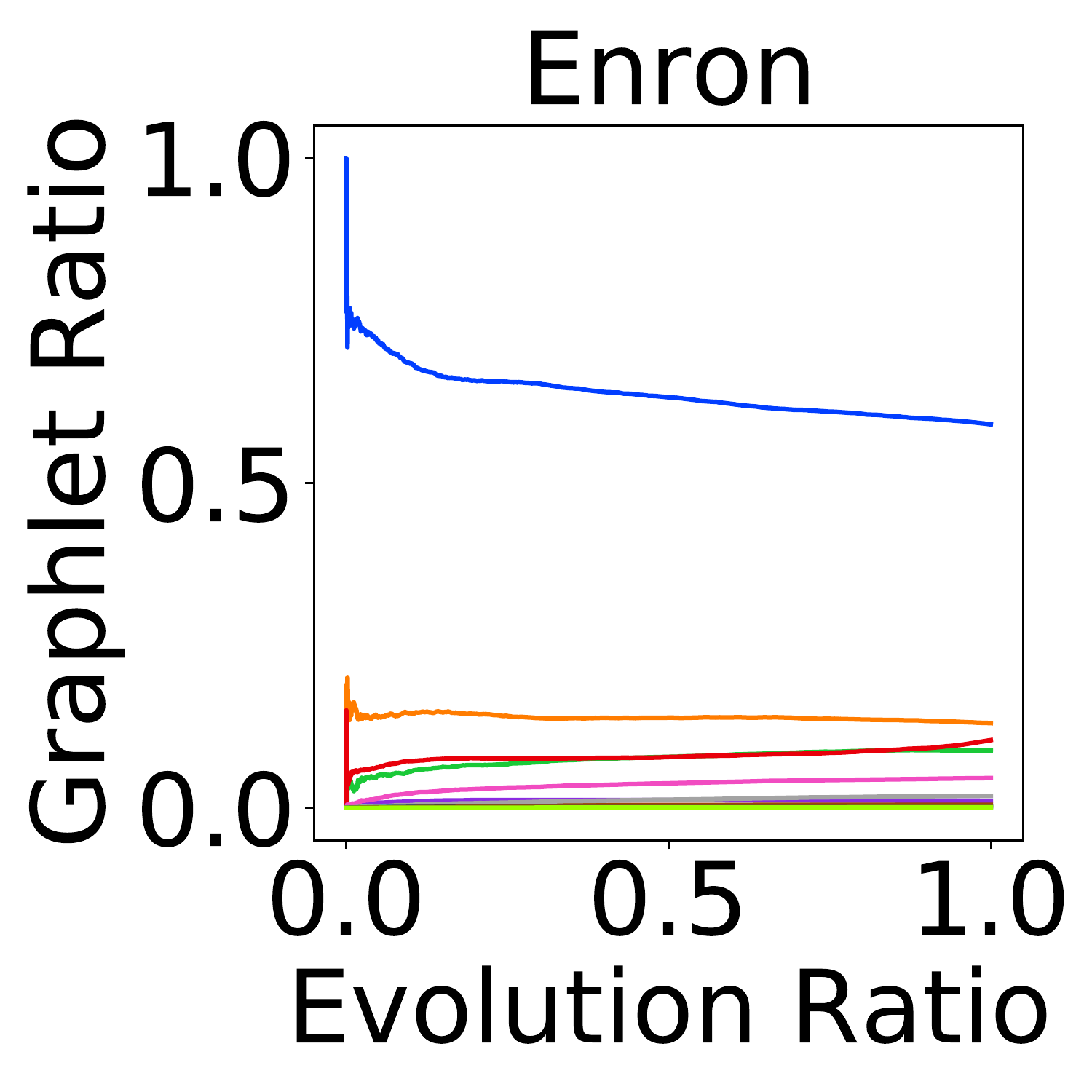}} &
        \raisebox{-.9\totalheight}{\includegraphics[width=0.1475\textwidth]{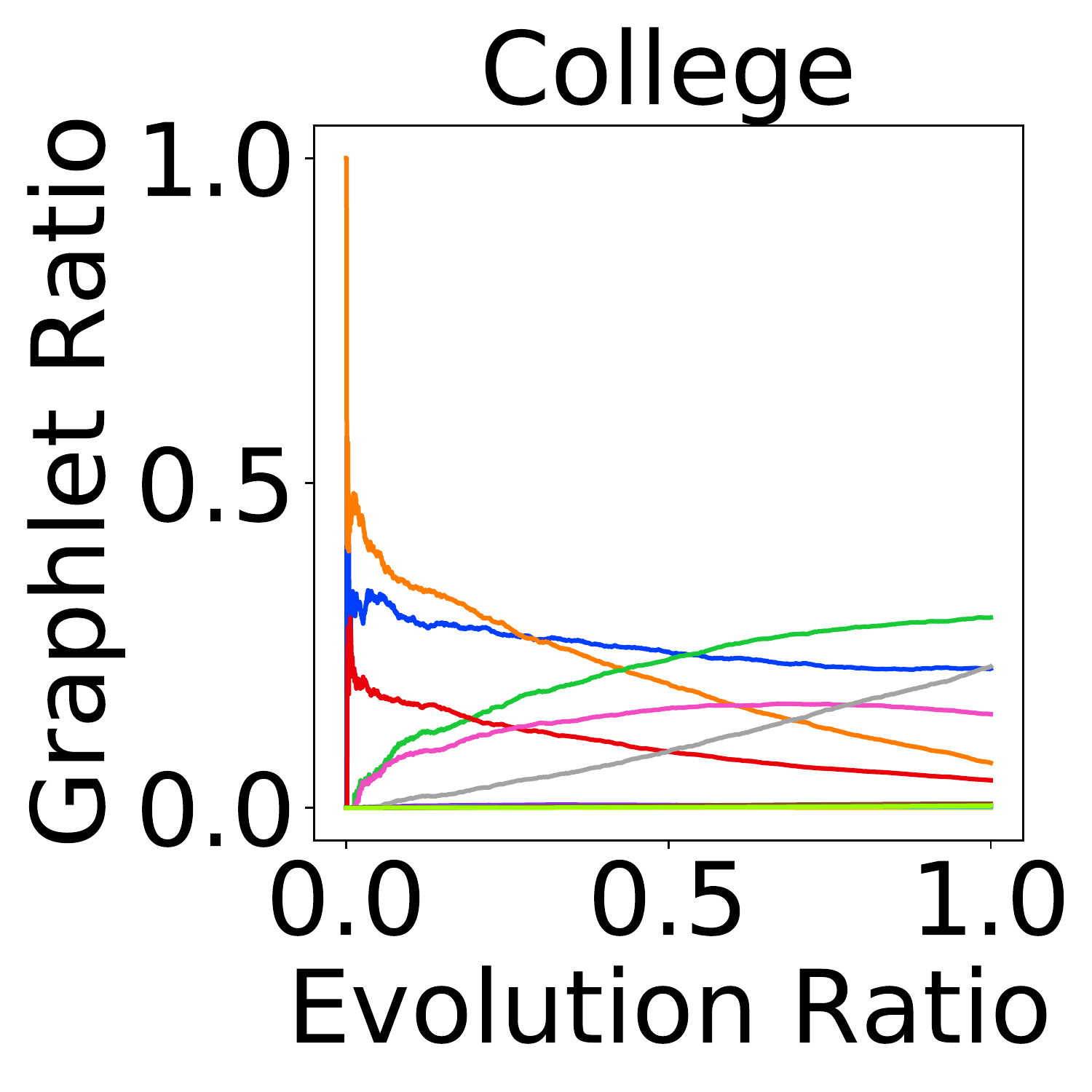}} \\
    \hline
    \parbox[t]{2mm}{\multirow{8}{*}{\rotatebox[origin=c]{90}{\ \ \ \ Online Q/A}}} &  
        \raisebox{-.9\totalheight}{\includegraphics[width=0.1475\textwidth]{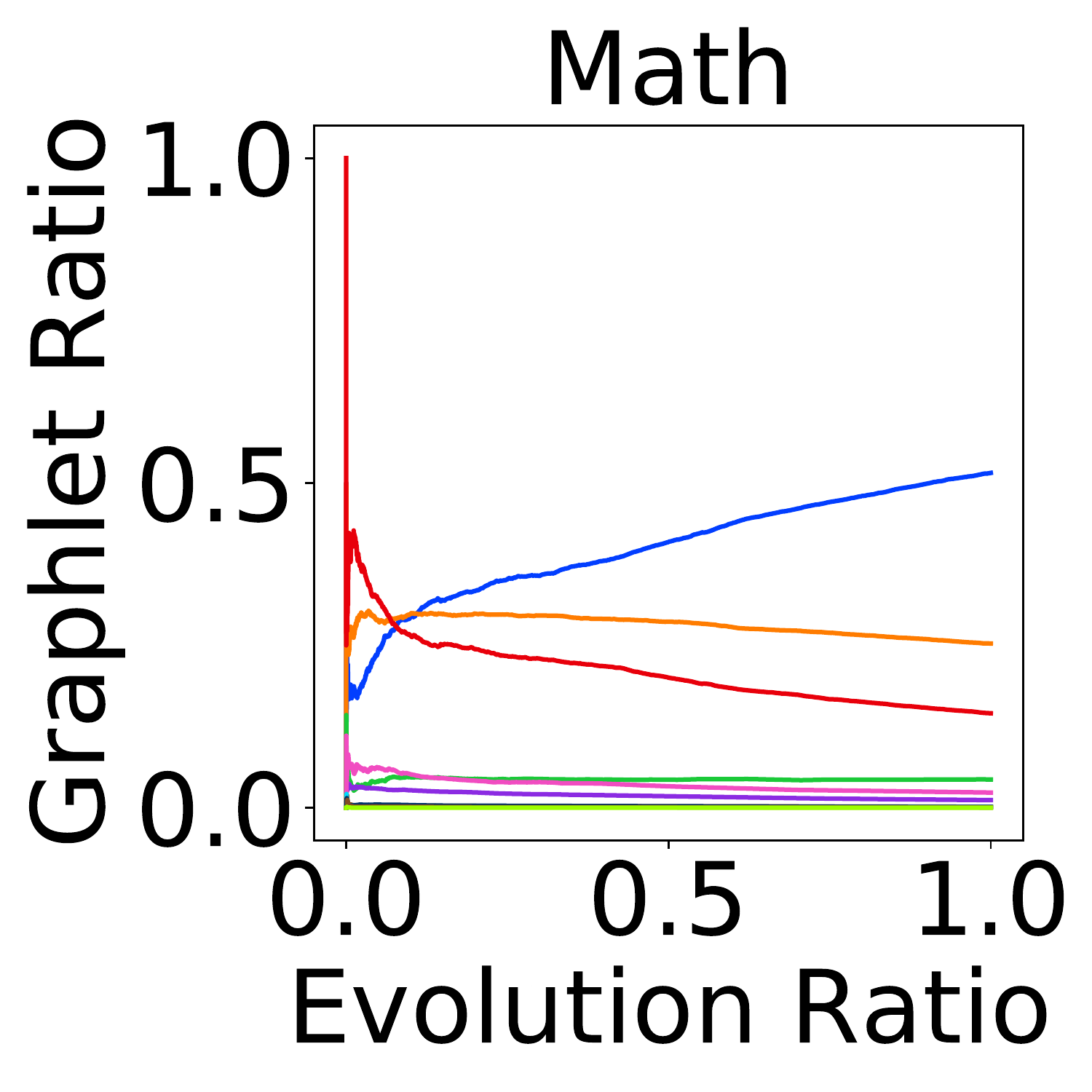}} & \raisebox{-.9\totalheight}{\includegraphics[width=0.1475\textwidth]{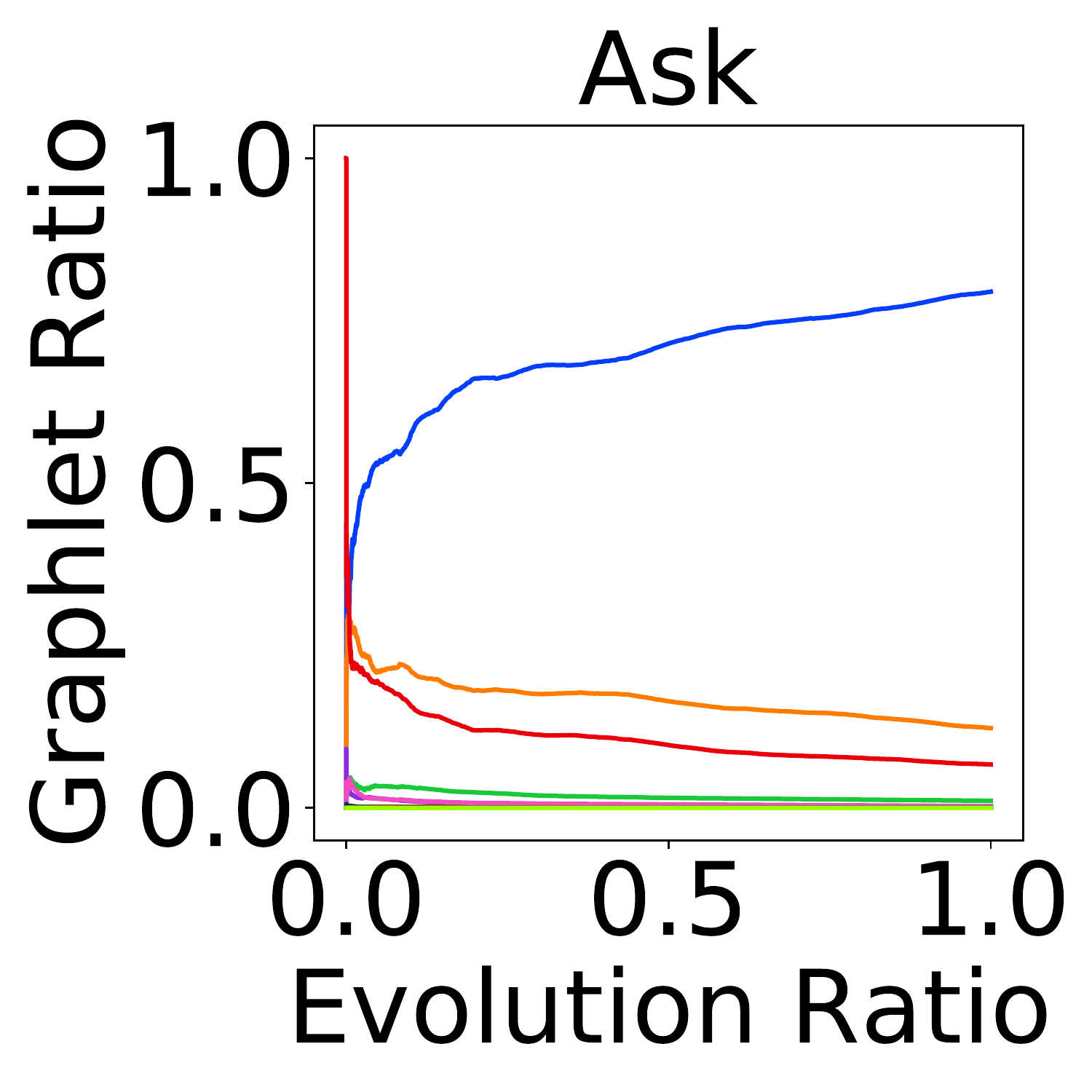}} & 
        \raisebox{-.9\totalheight}{\includegraphics[width=0.1475\textwidth]{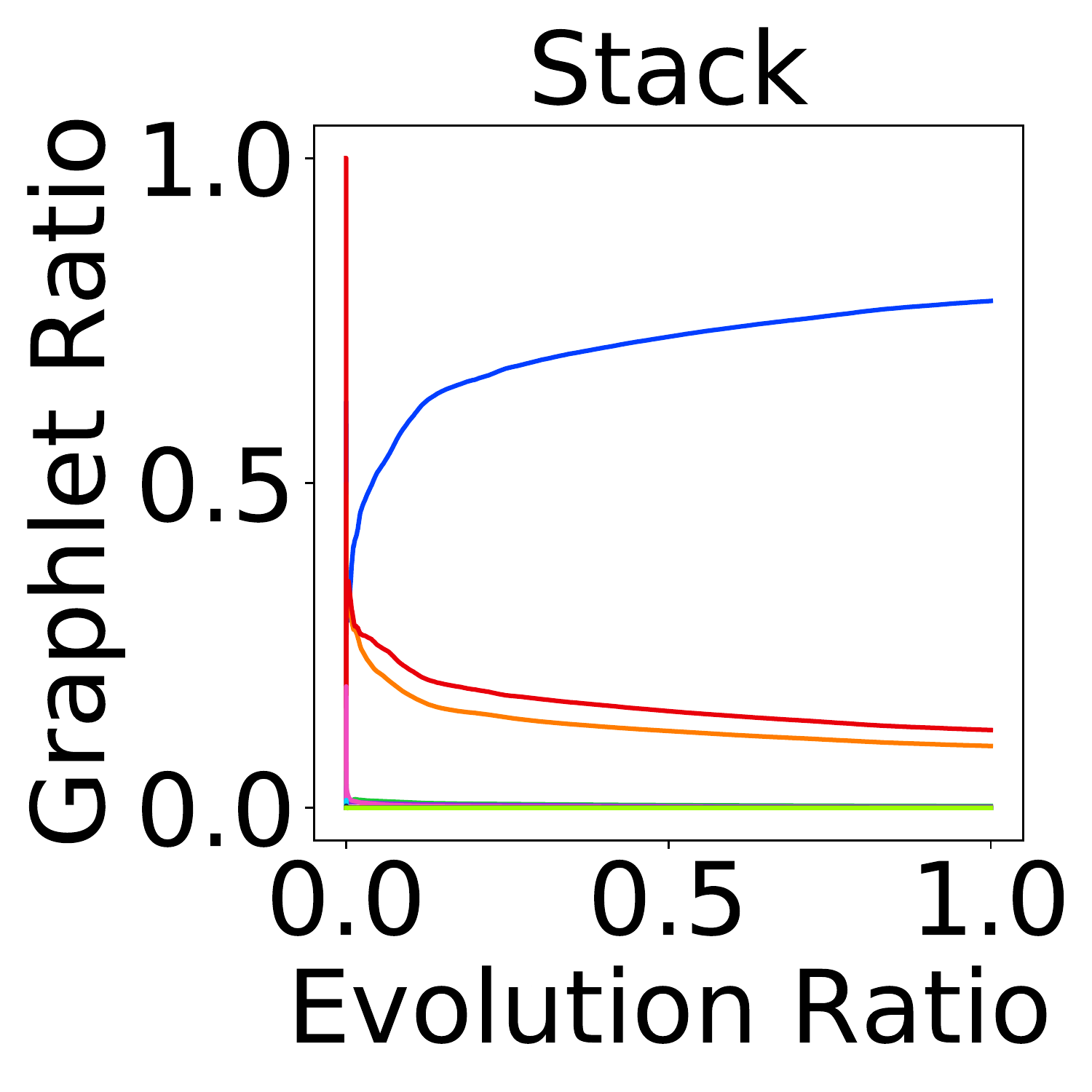}} & 
        \raisebox{-.9\totalheight}{\includegraphics[width=0.1475\textwidth]{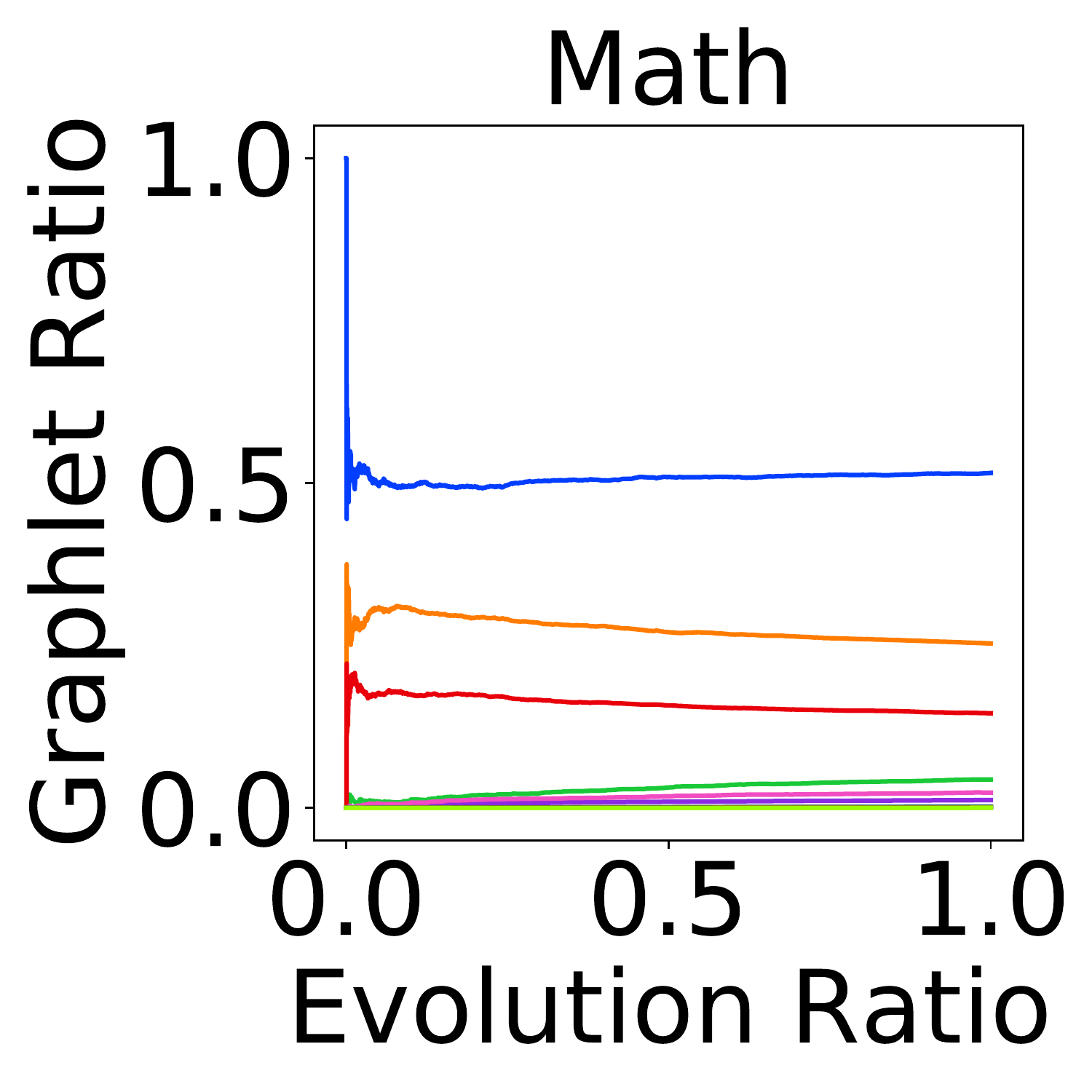}} & \raisebox{-.9\totalheight}{\includegraphics[width=0.1475\textwidth]{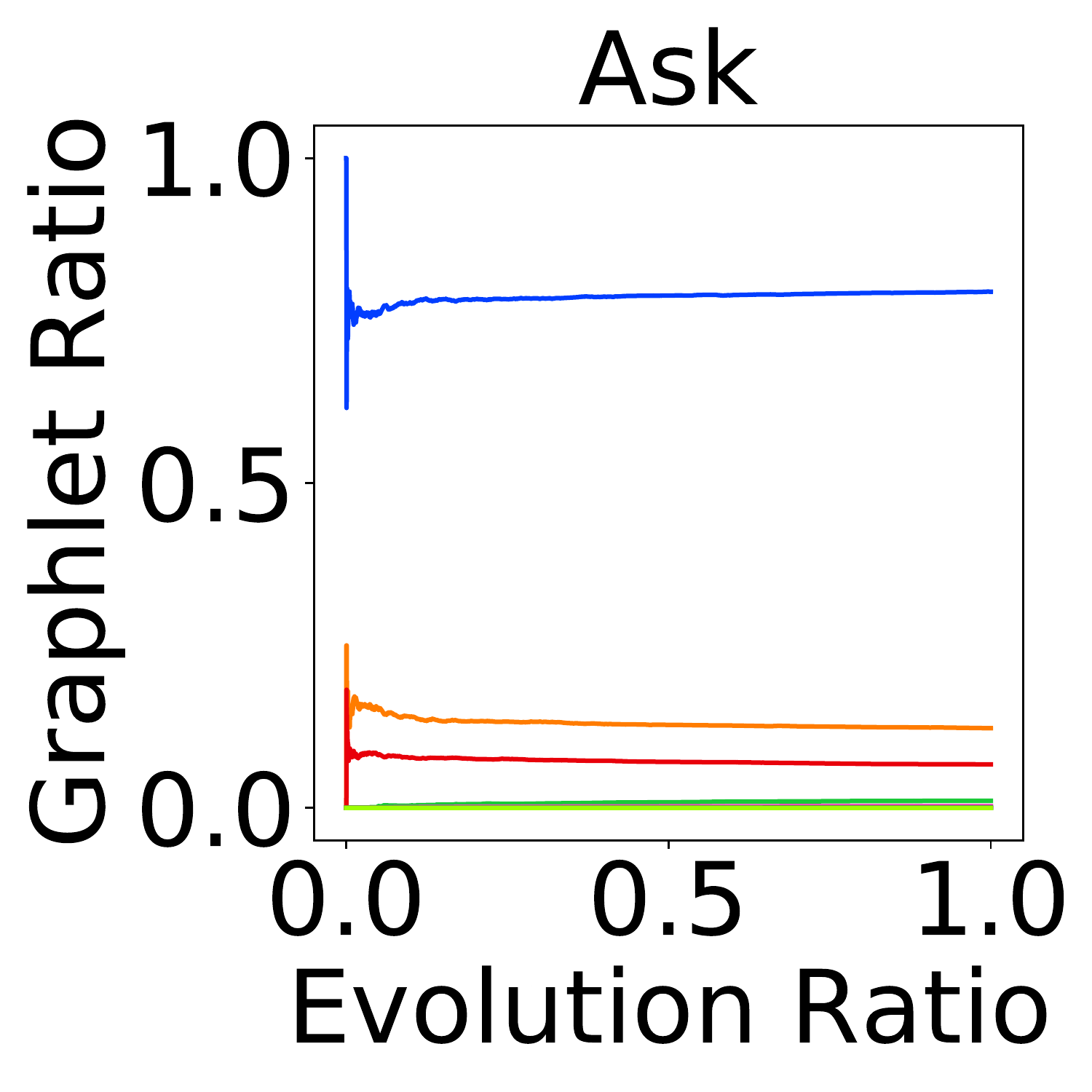}} &
        \raisebox{-.9\totalheight}{\includegraphics[width=0.1475\textwidth]{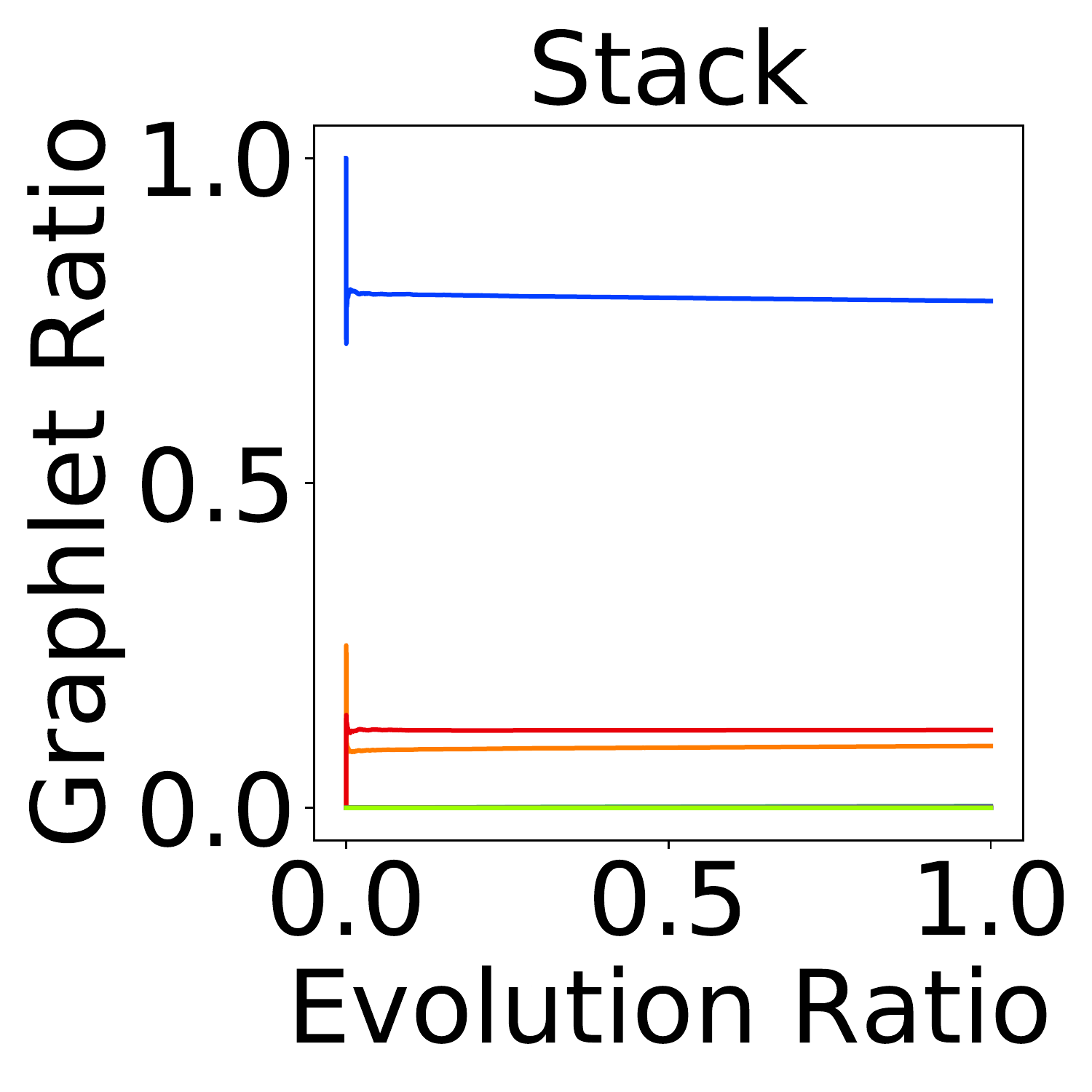}} \\
    \bottomrule
\end{tabular}}
\end{center}
\end{table*}

\begin{algorithm}[t]
\small
	\caption{Counting the Instances of Each Graphlet in a Temporal Graph \label{alg:track_graphlet}} 
	\SetKwInOut{Input}{Input}
	\SetKwInOut{Output}{Output}
	\SetKwProg{Procedure}{Procedure}{}{}
	\SetKwFunction{update}{UPDATE}
	
	\Input{Temporal Graph $\SG=(\SV, \SE, \ST)$}
	\Output{The count of the instances of each graphlet in $\SG$}
	\vspace{3pt}
	
	Initialize the count of the instances of each graphlet to zero \\
	Initialize $\SE$ to an empty set \\
    \vspace{3pt}
    
	\For{\upshape\textbf{each edge }$e_i=u \rightarrow v$ in arrival order}{
	    
	    $\SN \leftarrow $ union of the neighbors of $u$ and the neighbors of $v$ (except for $u$ and $v$) \label{alg:track_graphlet:line:neighbors}
	    
	    \For{\upshape\textbf{each} $w\in \SN$}{
	        \If{$u$, $v$ \textbf{and} $w$ form a graphlet instance}{ \label{alg:track_graphlet:line:if_instance}
	         decrement the count of the graphlet of the instance formed by $u$, $v$ and $w$ \label{alg:track_graphlet:line:decrease}
	         \\
            
            }
         }
         
         add $u\rightarrow v$ to $\SE$
         
         \For{\upshape\textbf{each} $w\in \SN$}{
              increment the count of the graphlet of the instance formed by $u$, $v$ and $w$ \label{alg:track_graphlet:line:increase}
	         \\
         }
	    \vspace{3pt}
	    
	}
	
	\textbf{return} count of the instances of each graphlet instances \\
\end{algorithm}

\section{Graph Level Analysis}\label{section:graph}

In this section, we study the evolution of local structures in real-world graphs.
We examine the dynamics in the distribution of graphlet instances and transitions between graphlets.

\subsection{Global Level 1. Graphlets Over Time}
\label{section:graph:time}
We track how the ratio of the instances of each graphlet changes as the considered real-world graphs evolve over time. Our tracking algorithm, which is described in Algorithm~\ref{alg:track_graphlet}, is adapted from StreaM \cite{schiller2015stream}, which maintains the counts of the instances of the $4$-node undirected graphlets in a fully dynamic graph stream, where edges are not just added but also deleted over time. The time complexity of Algorithm~\ref{alg:track_graphlet} is $\Theta(\Sigma_{v\in\SV}(d(v))^2)$, as proven in Theorem~\ref{thm:time:track}.
It should be noticed that, by Lemma~\ref{lem:time:optimality}, the time complexity is $\Theta($the number of instances of all graphlets in the last snapshot$)$, which is the optimal time complexity achievable by any algorithm that counts graphlet instances by enumerating them.

\begin{theorem} \label{thm:time:track} 
The time complexity of Algorithm \ref{alg:track_graphlet} is $\Theta(\Sigma_{v\in\SV}(d(v))^2)$.
\end{theorem}
\begin{proof}
Since the number of nodes forming each graphlet instance is a constant, finding the graphlet corresponding to a given instance and updating the corresponding count (lines \ref{alg:track_graphlet:line:if_instance}-\ref{alg:track_graphlet:line:decrease} and \ref{alg:track_graphlet:line:increase}) take $O(1)$ time. Thus, the time complexity of processing each incoming edge $e_i=u\rightarrow v$ is that of computing the union of the neighbors of $u$ and $v$ (line~\ref{alg:track_graphlet:line:neighbors}), which is $\Theta(\DTUi+\DTVi)$.
Hence, the total complexity is $\Theta(\sum_{e_i=u \rightarrow v \in E} (\DTUi+\DTVi) = \Theta(\sum_{v\in \SV}(d(v))^2)$.
\end{proof}
\begin{lemma} \label{lem:time:optimality} 
The number of instances of all graphlets in a snapshot $\SGT$ is $\Theta(\Sigma_{v\in\SVT}(\DTV)^2)$.
\end{lemma}
\begin{proof}
Given a snapshot $\SGT=(\SVT, \SET)$, for each node $v\in \SVT$, if we count the instances of all graphlets that consist of $v$ and its two neighbors, then the count of such instances is $\Theta((\DTV)^2)$ for each node $v$, and since $\DTV\geq 1$ for every node $v$, the total count $C$ is $\Theta(\Sigma_{v\in\SVT}(\DTV)^2)$. 

\smallsection{Lower Bound:} Since each graphlet instance, which consists of three nodes, is counted at most three times, $C$ is at most three times the number of instances of all graphlets in $\SGT$.
In other words, the number of instances of all graphlets is at least $1/3$ of $C$, and thus it is $\Omega(\Sigma_{v\in\SVT}(\DTV)^2)$.

\smallsection{Upper Bound:} In each graphlet instance, there exists at least one center node, who composes the graphlet together with its neighbors. Thus, each instance is counted at least once, and thus $C$ is at least the number of instances of all graphlets in $\SGT$.
In other words, the number of instances of all graphlets is at most $C$, and thus it is $O(\Sigma_{v\in\SVT}(\DTV)^2)$.
\end{proof}

As seen in Table~\ref{tab:graphlet_evolution}, the dynamics of the ratios depend on the domains of the graphs, as summarized in Observation~\ref{obs:graphlet_evolve}.

\noindent\fbox{%
        \parbox{\columnwidth}{%
        \vspace{-2mm}
        \begin{observation} \label{obs:graphlet_evolve}
            The dynamics in the distributions of graphlet instances in graphs from the same domain share some commonalities.
            \begin{itemize}
                \item Instances of graphlet 4 are more dominant in the citation graphs than other graphs. 
                \item Graphlets with many edges (e.g., graphlets 8, 12, and 13) account for a larger fraction in email/message networks than in other networks. 
                \item The fraction of graphlet 1 increases over time only in the online Q/A graphs.
            \end{itemize}
        \end{observation}
        \vspace{-2mm}
        }%
    }

\noindent However, the dynamics are not exactly the same within domains. For example, while graphlets 1, 2, and 4 are dominant compared to other graphlets in all citation graphs, the ratios among them vary greatly in different graphs.

We also notice a consistent difference between the dynamics in real-world graphs and those in randomized graphs (see Section~\ref{sec:prelim:concept}), as summarized in Observation~\ref{obs:graphlet_evolve:random}.

\vspace{0.5mm}
\noindent\fbox{%
        \parbox{0.98\columnwidth}{%
        \vspace{-2mm}
        \begin{observation} \label{obs:graphlet_evolve:random}
            The ratios of graphlet instances change more linearly in randomized graphs than in real-world graphs.
        \end{observation}
        \vspace{-2mm}
        }%
    }
\vspace{0.5mm}

\begin{table}[t]
\caption{\label{tab:linearity} 
The non-linearity of the ratios of graphlet instances over time in real-world graphs and randomized graphs. 
We describe in Section~\ref{section:graph:time} how the non-linearity is measured.
The lower the non-linearity is, the more linear the change of the ratio of the corresponding graphlet instances is.
Note that
the ratios of graphlet instances change more linearly in randomized graphs than in real-world graphs.
}
\resizebox{\columnwidth}{!}{
\begin{tabular}{|c|ccc|ccc|ccc|}
    \hline
    	  Dataset & HepPh & HepTh & Patent & EU & Enron & College & Math & Ask & Stack \\
    	  \hline
    	  real      &  0.0027   & 0.0080    & 0.0093 & 0.0107 & 0.0042 & 0.0095 & 0.0028 & 0.0038 & 0.0047 \\
    	  random    &  0.0003   & 0.0011    & 0.0000 & 0.0081 & 0.0017 & 0.0058 & 0.0007 & 0.0005 & 0.0001 \\
    \hline
\end{tabular}}
\end{table}
\noindent In order to numerically support this observation, we measure the non-linearity \cite{kroll1993theoretical, hsieh2008statistical} of the ratios of graphlet instances over time.
Specifically, we fit a linear regression model and a non-linear polynomial regression model to each time series in Table~\ref{tab:graphlet_evolution}, and then we measure the average absolute difference between the predicted values of the two models as the non-linearity of the time series.\footnote{We use the linearity test implemented in Analyse-it (Ver. 5.65) and select a cubic model as the non-linearity polynomial model, as suggested in the program. For computational efficiency, we measure the absolute difference at 1,000 evolution ratios sampled uniformly at equal intervals.}
Lastly, we average the non-linearity of all time-series from each graph and report the results in Table~\ref{tab:linearity}.
Note that non-linearity is significantly higher in real-world graphs than in corresponding randomized graphs. 
That is, the ratios of graphlet instances change more linearly in randomized graphs than in real-world graphs.

\begin{table*}[ht]
\vspace{-2mm}
\caption{\label{tab:gtg} Using graphlet transition graphs (GTGs) and characteristic profiles (CPs) from GTGs, we can accurately characterize the dynamics of local structures in real-world graphs.
The colors of edges in GTGs indicate their normalized weights.
Note that GTGs and CPs are particularly similar in real-world graphs from the same domains (Observation~\ref{obs:transition:domain}).}
\resizebox{\textwidth}{!}{
\begin{tabular}{c|ccc|c}
    \toprule
    	      & \multicolumn{3}{c|}{Graphlet transition graphs (GTGs)} & Characteristic profiles (CPs) \\
    \hline
    \parbox[t]{2mm}{\multirow{6}{*}{\rotatebox[origin=c]{90}{Citation}}} &  
        \raisebox{-.9\totalheight}{\includegraphics[width=0.2\textwidth]{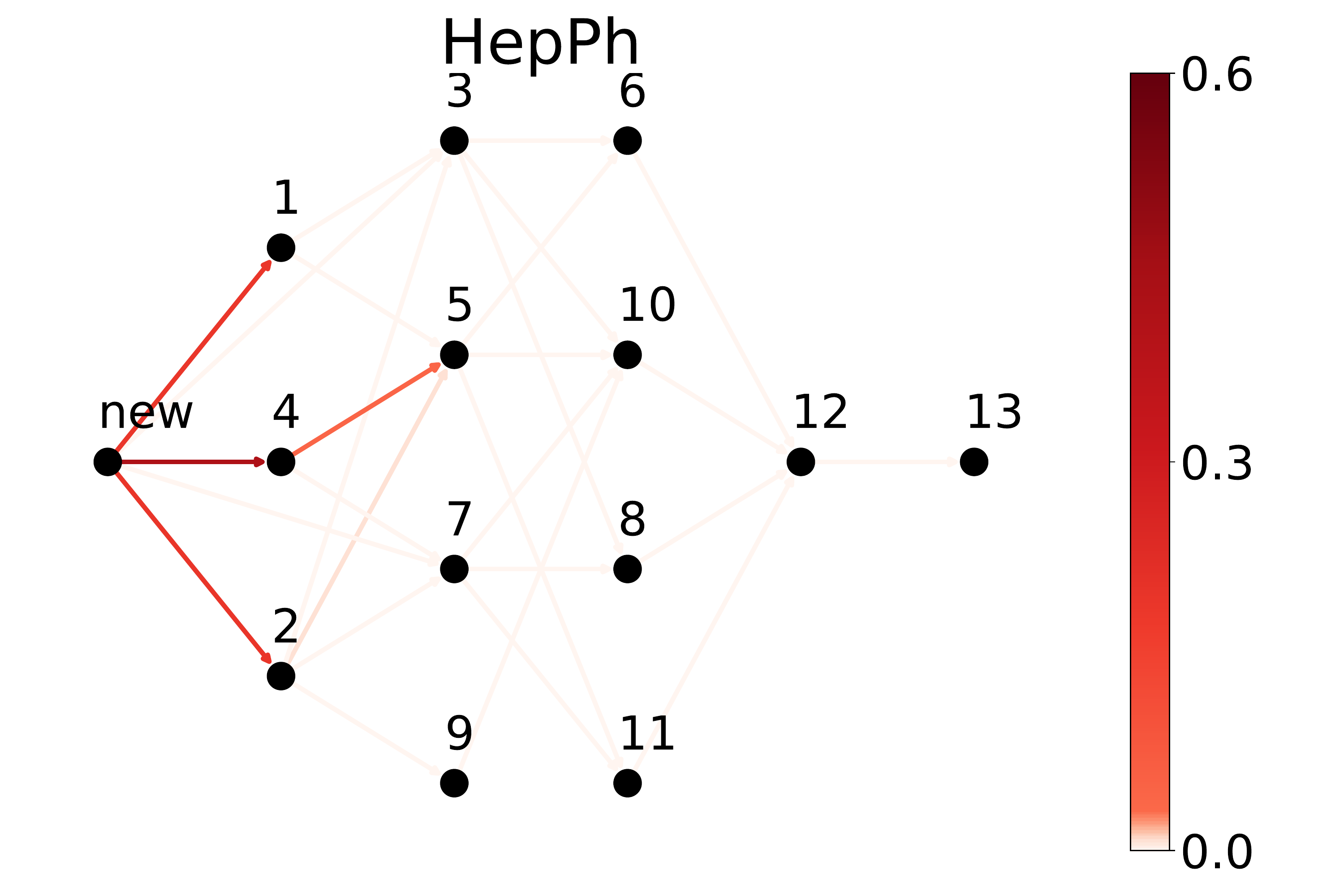}} &
        \raisebox{-.9\totalheight}{\includegraphics[width=0.2\textwidth]{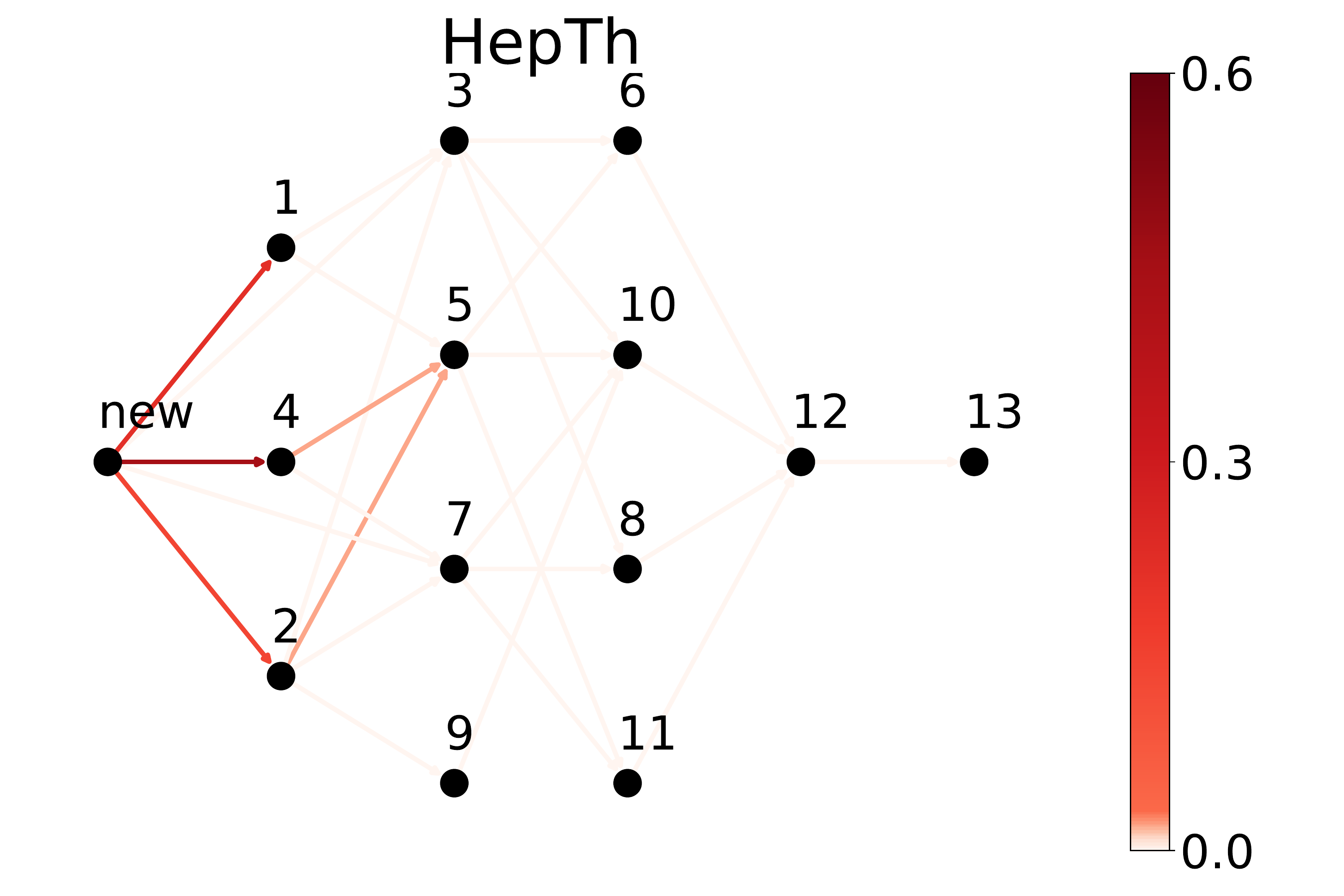}} & 
        \raisebox{-.9\totalheight}{\includegraphics[width=0.2\textwidth]{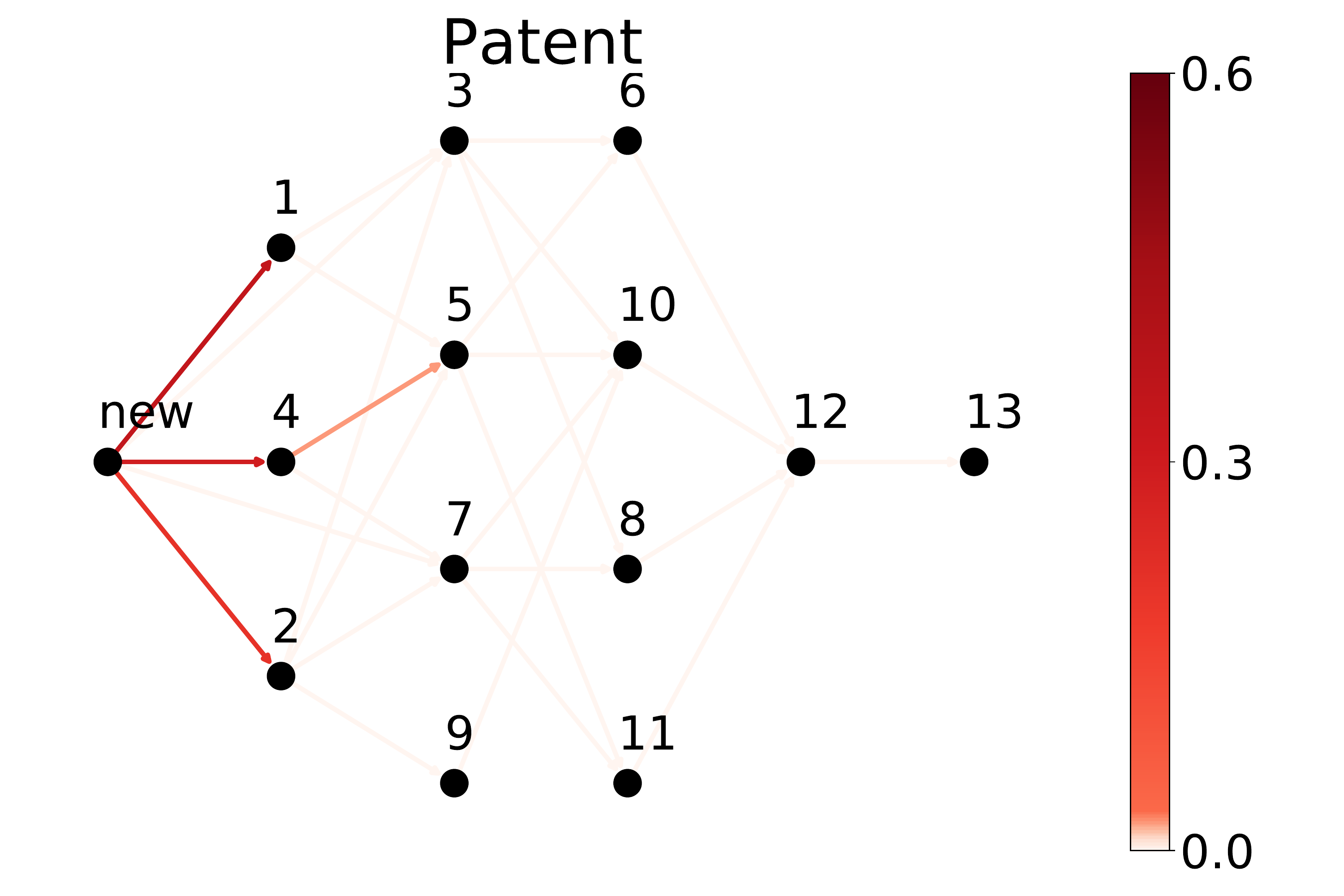}} & 
        \raisebox{-.9\totalheight}{\includegraphics[width=0.3\textwidth]{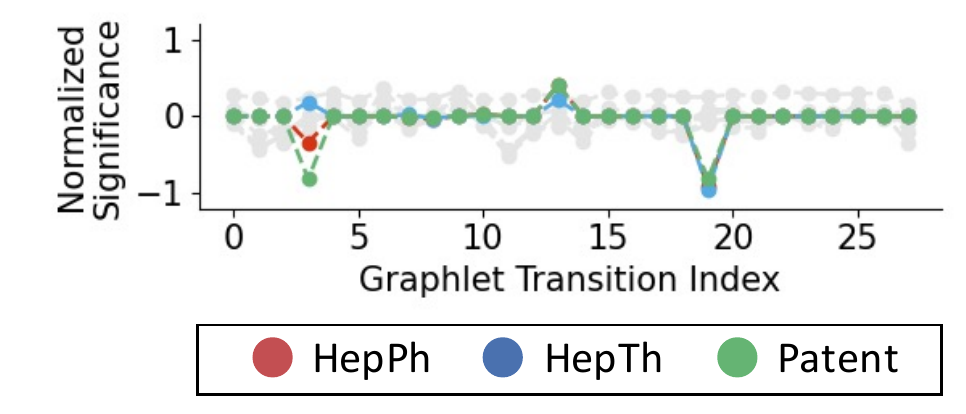}} \\
    \hline
    \parbox[t]{2mm}{\multirow{6}{*}{\rotatebox[origin=c]{90}{Email/Message}}} &  
        \raisebox{-.9\totalheight}{\includegraphics[width=0.2\textwidth]{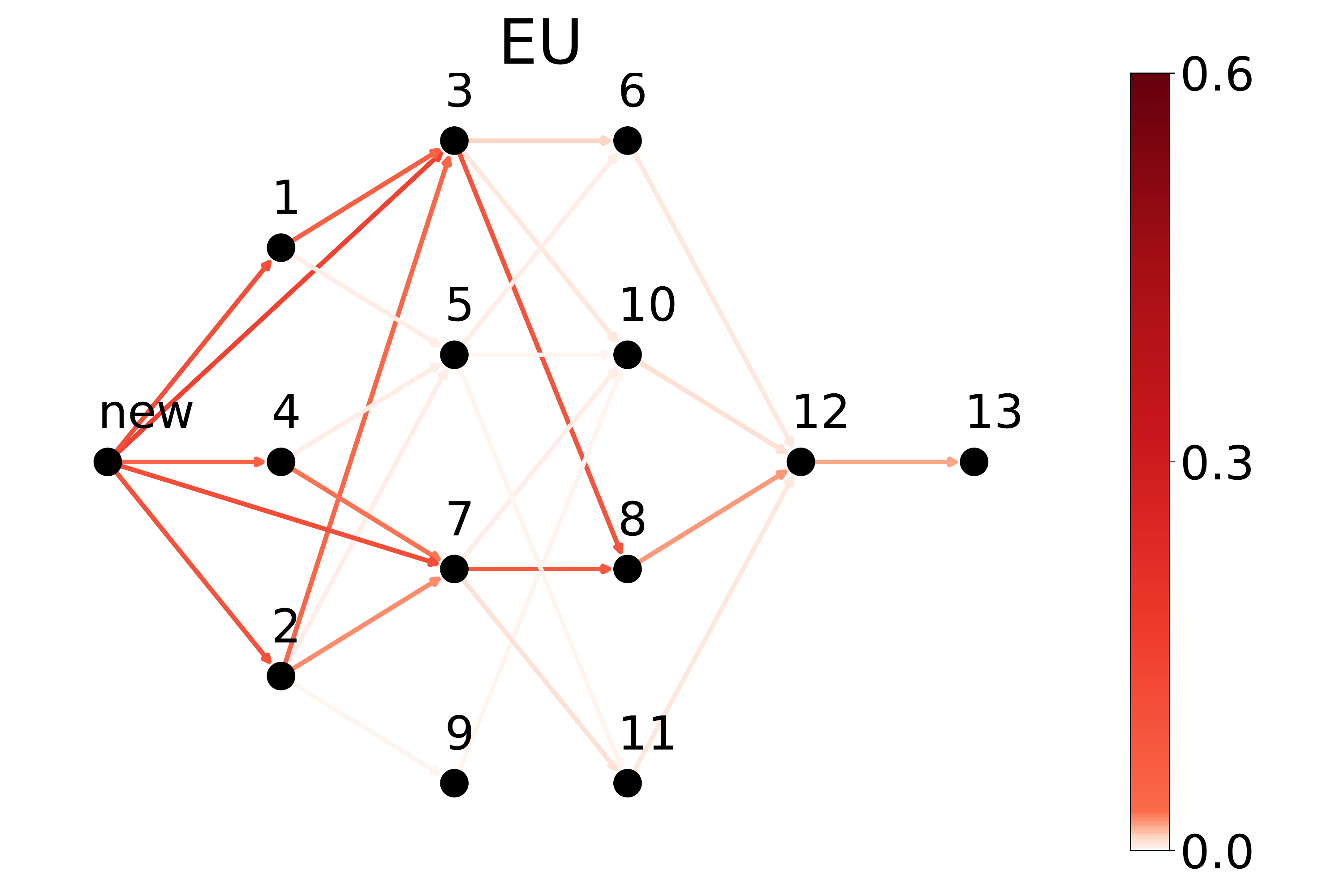}} &
        \raisebox{-.9\totalheight}{\includegraphics[width=0.2\textwidth]{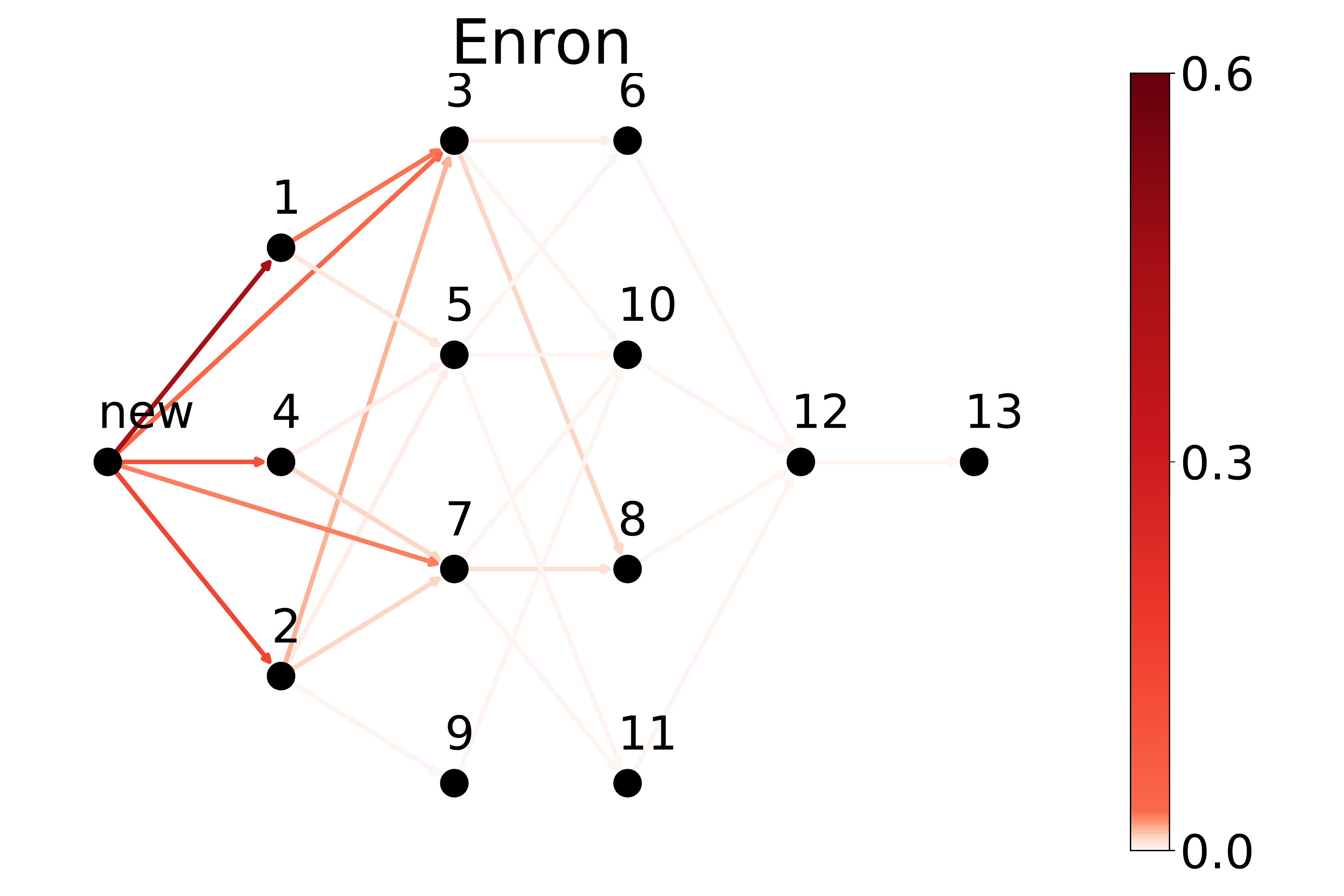}} & 
        \raisebox{-.9\totalheight}{\includegraphics[width=0.2\textwidth]{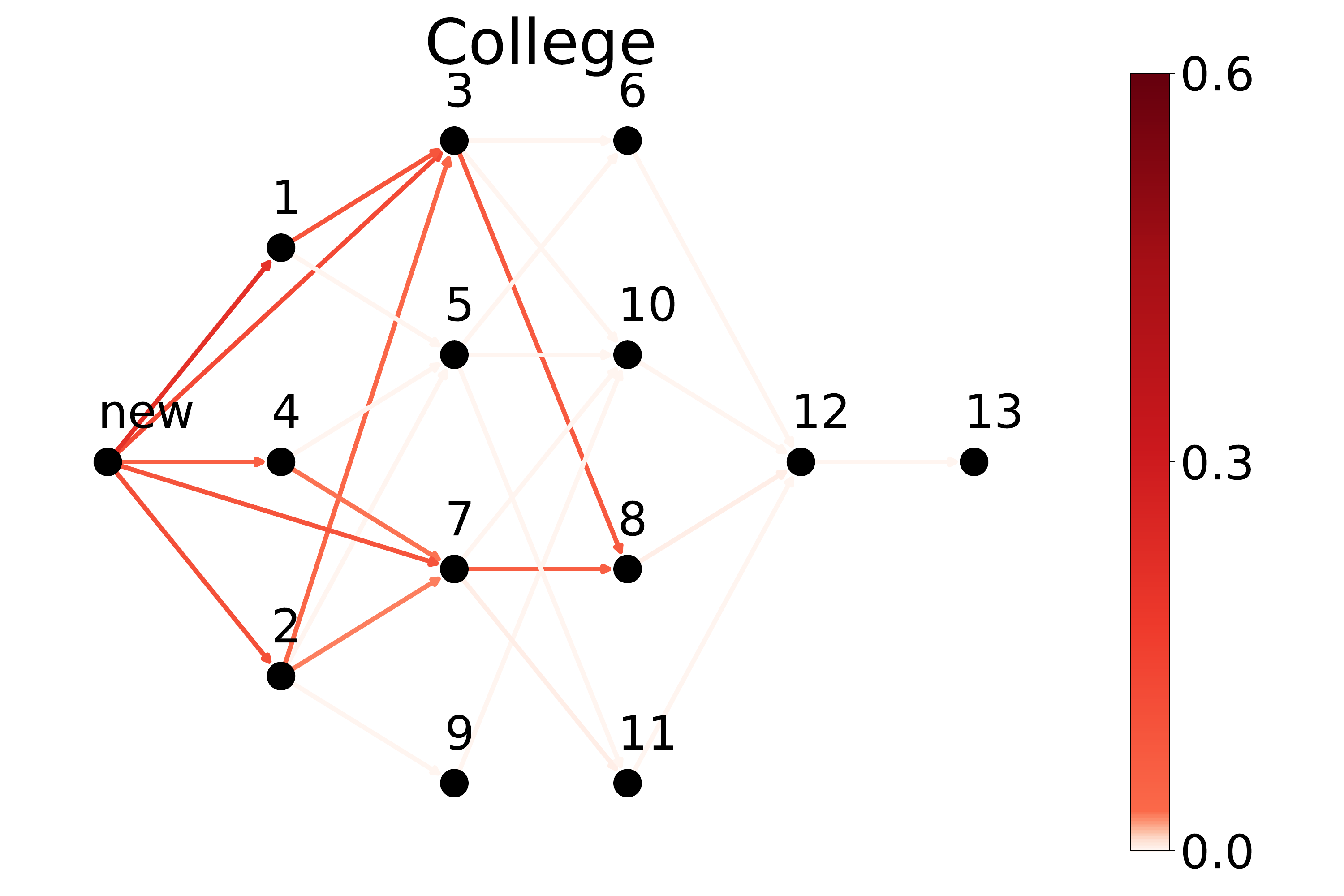}} & 
        \raisebox{-.9\totalheight}{\includegraphics[width=0.3\textwidth]{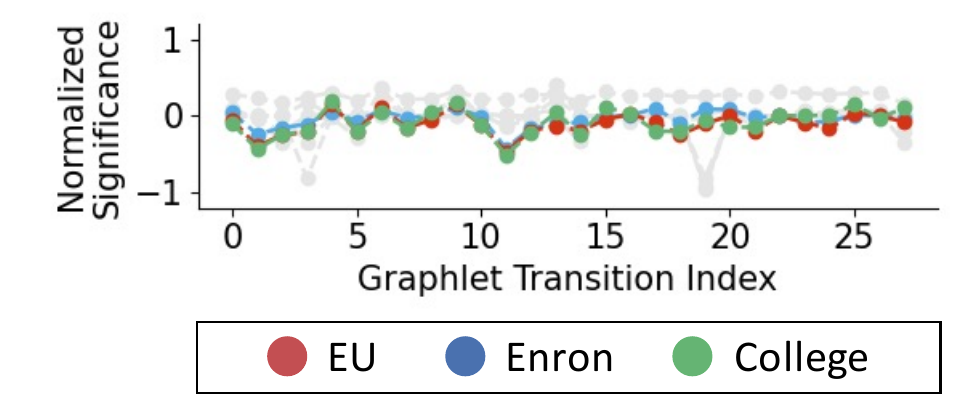}} \\
    \hline
    \parbox[t]{2mm}{\multirow{6}{*}{\rotatebox[origin=c]{90}{Online Q/A}}} &  
        \raisebox{-.9\totalheight}{\includegraphics[width=0.2\textwidth]{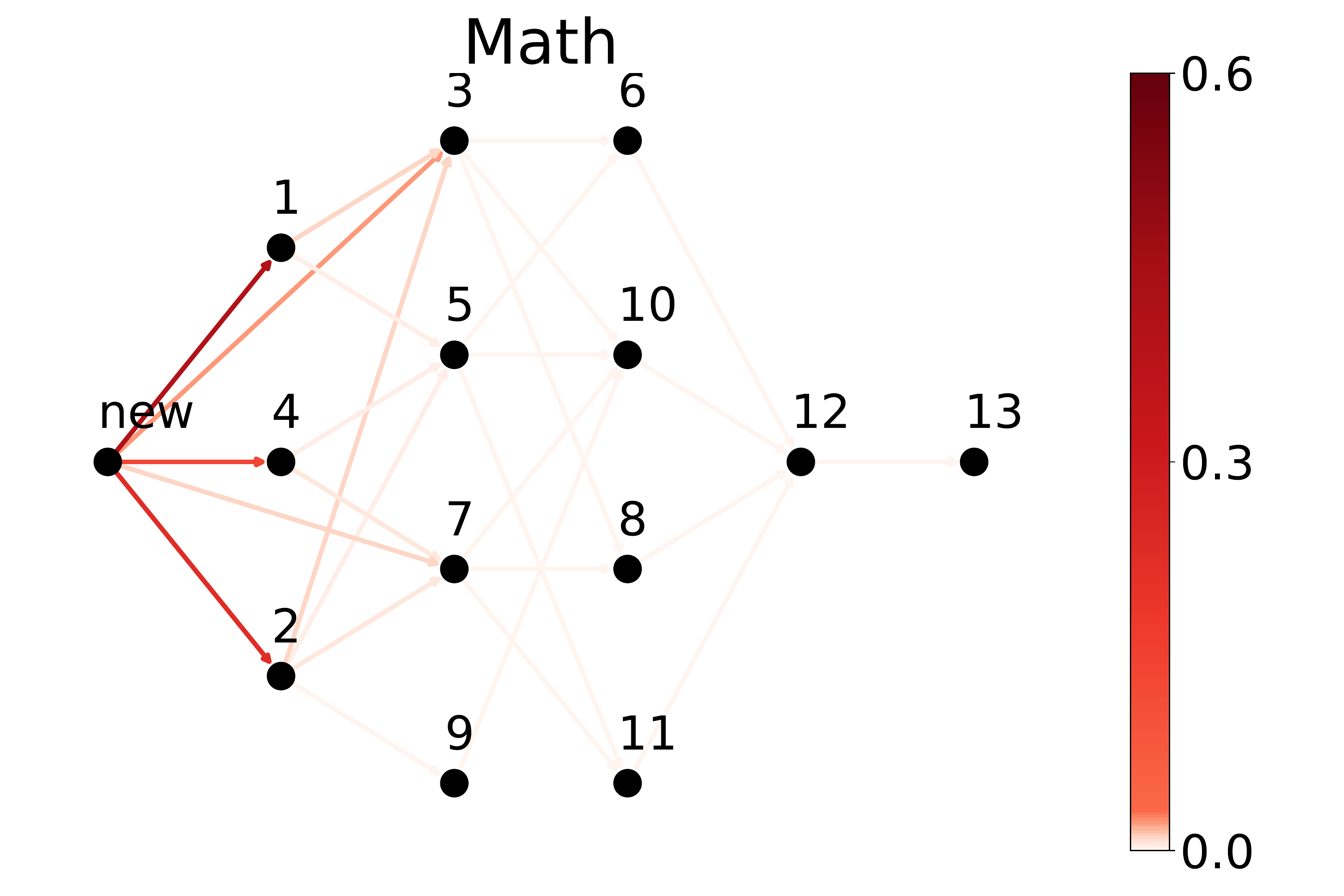}} &
        \raisebox{-.9\totalheight}{\includegraphics[width=0.2\textwidth]{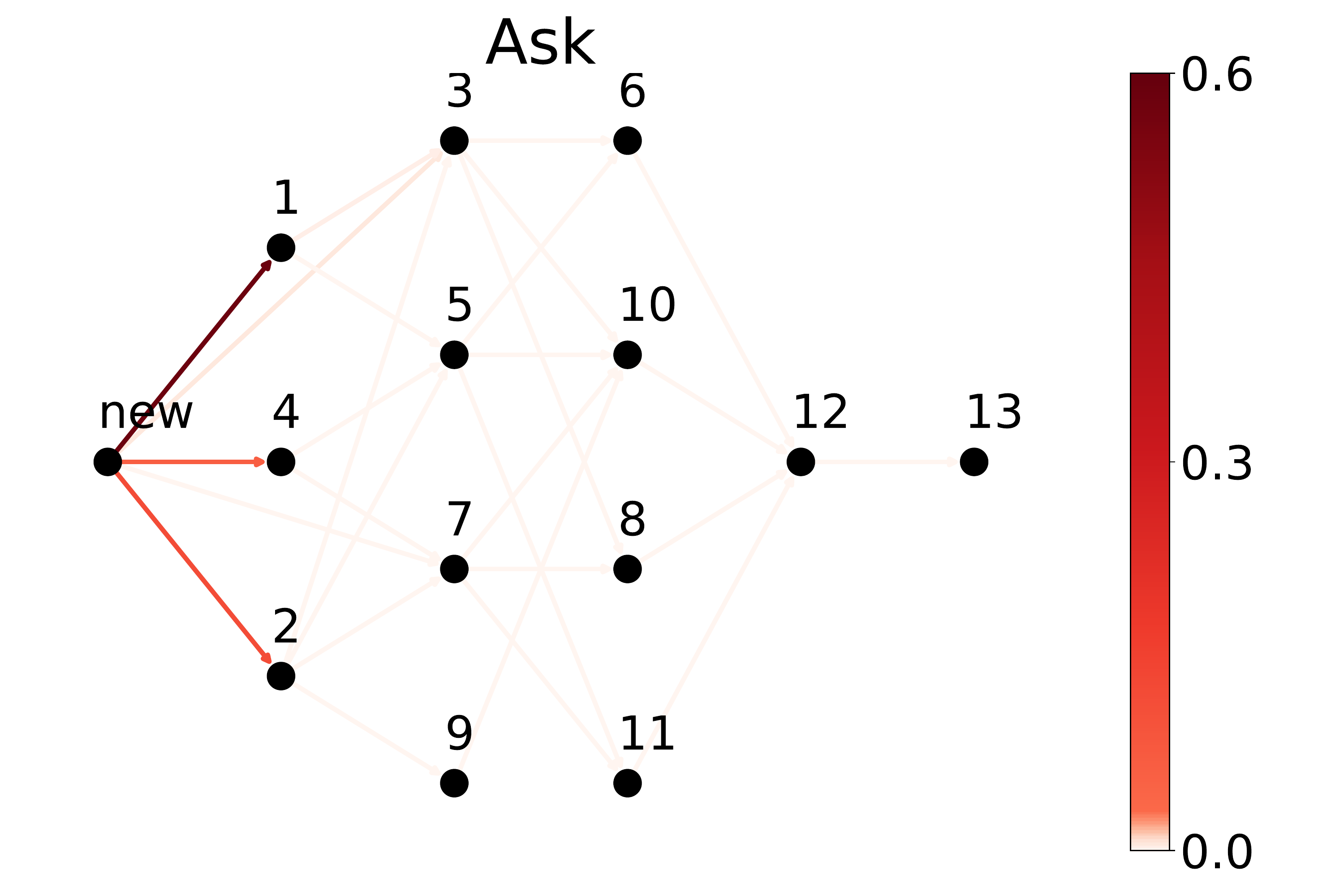}} & 
        \raisebox{-.9\totalheight}{\includegraphics[width=0.2\textwidth]{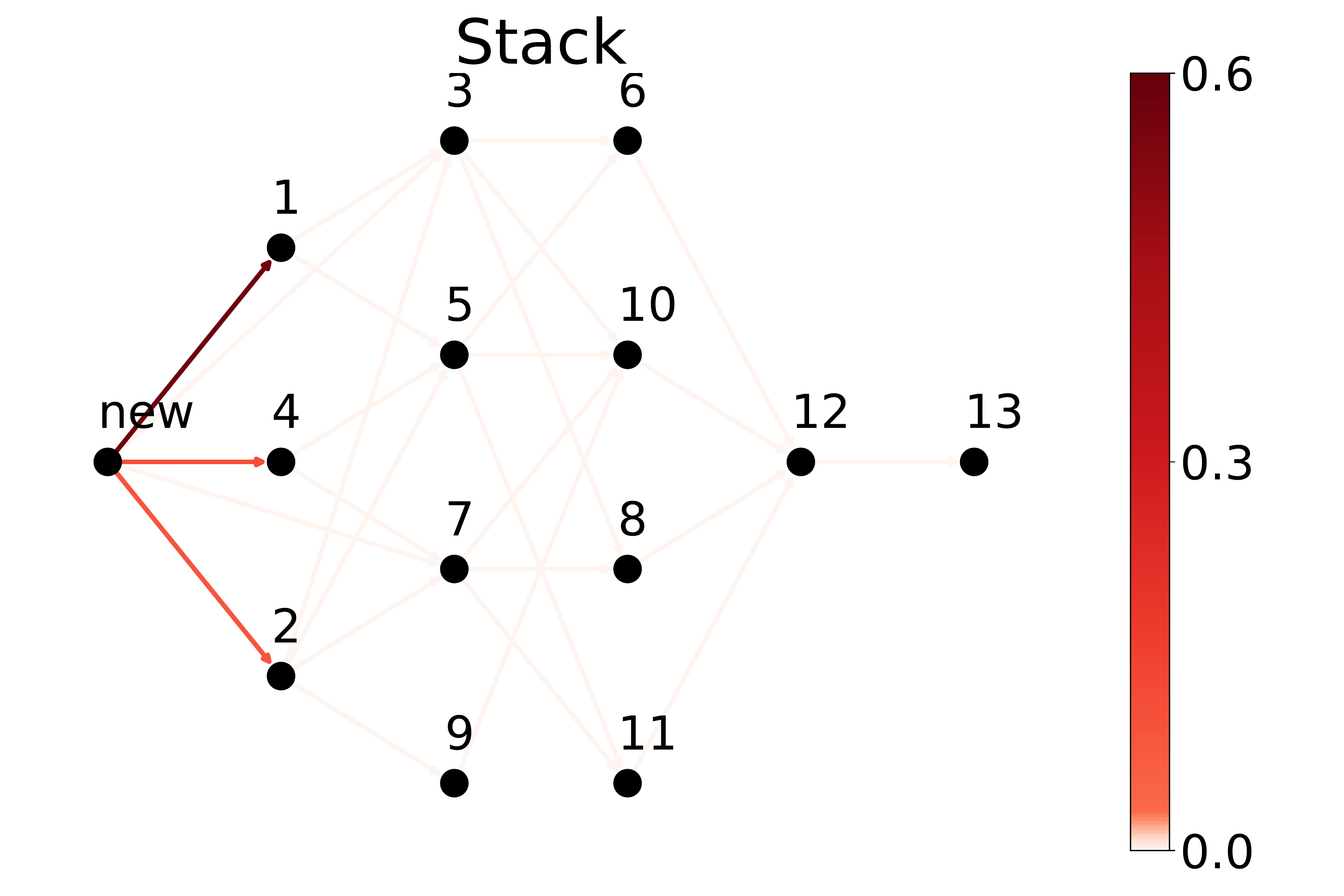}} & 
        \raisebox{-.9\totalheight}{\includegraphics[width=0.3\textwidth]{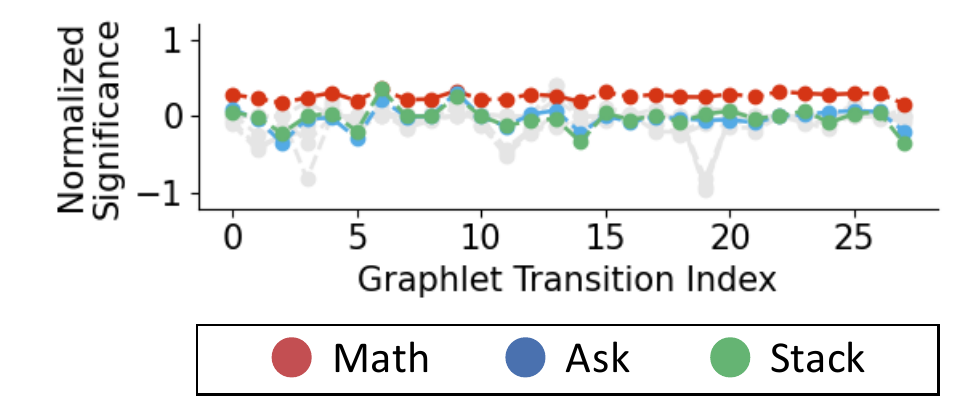}} \\
    \bottomrule
\end{tabular}}
\end{table*}

\subsection{Global Level 2. Graphlet Transitions}
\label{section:graph:transition}

In a temporal graph, an instance of a graphlet may transition to an instance of another graphlet due to new edges added to it.
In this subsection, we examine the counts of such transitions between graphlets to characterize the local dynamics in temporal graphs and also to make comparisons between them.

\smallsection{Graphlet Transition Graph:}
We define \textit{graphlet transition graphs} (GTGs) to encode transitions between graphlets.

\noindent\fbox{%
    \parbox{0.98\columnwidth}{%
    \vspace{-2mm}
    \begin{definition}[Graphlet transition graph] \label{defn:transition}
       A \textit{graphlet transition graph} (GTG) $G=(V,E,W)$ of a temporal graph $\SG$ is a static directed weighted graph where the nodes are graphlets and each edge indicates that the source graphlet is transformed into the destination graphlet by an edge added to $\SG$.
        The weight of edges is the number of occurrences of the corresponding transitions. 
        We use $W=\{w_1, \cdots, w_{|E|}\}$ to denote the edge weights.
    \end{definition}
    \vspace{-2mm}
    }%
 }
\vspace{0.5mm}
\\
\noindent Since we focus on the 13 graphlets in Figure~\ref{fig:graphlet_and_role}(a), a GTG consists of the 28 types of transitions between these graphlets.
In Table~\ref{tab:gtg}, we visualize the GTGs from the real-world graphs.
Algorithm~\ref{alg:count_graphlet} describes the computation of the edge weights of 
 a GTG. In a nutshell, for each edge in arrival order, we count the transitions caused by it.
Its time complexity is formalized in Theorem~\ref{thm:time:transition}.

\begin{theorem} \label{thm:time:transition} 
The time complexity of Algorithm~\ref{alg:count_graphlet} is $\Theta(\Sigma_{v\in\SV}(d(v))^2)$ = $\Theta($the number of instances of all graphlets in the last snapshot$)$.
\end{theorem}
\begin{proof}
We can prove the complexity of $\Theta(\Sigma_{v\in\SV}(d(v))^2)$ similarly to Theorem~\ref{thm:time:track}, and by Lemma~\ref{lem:time:optimality}, it is $\Theta($the number of instances of all graphlets in the last snapshot$)$.
\end{proof}

\begin{algorithm}[t]
\small
	\caption{Computing the Edge Weights of Graphlet Transition Graphs \label{alg:count_graphlet}} 
	\SetKwInOut{Input}{Input}
	\SetKwInOut{Output}{Output}
	\SetKwProg{Procedure}{Procedure}{}{}
	\SetKwFunction{update}{UPDATE}
	
	\Input{Temporal graph $\SG = (\SV, \SE, \ST)$}
	\Output{Edge weights of the graphlet transition graph of $\SG$}
	
	
	Initialize all edge weights to zero \\
	Initialize $\SE$ to an empty set \\
	\For{\upshape\textbf{each }edge $e_i=u\rightarrow v$ in arrival order}{
        \For{\upshape\textbf{each} $w_1$ $\in$ neighbors$(u)$ {$\setminus$ $\{v\}$} }{
            \update{$u, v, w_1$}
        }	
        \For{\upshape\textbf{each} {$w_2$ $\in$ neighbors$(v)$ $\setminus\{$neighbors$(u)$ $\cup$ $u\}$}}{
            \update{$u, v, w_2$}
        }
        add $u\rightarrow v$ to $\SE$
	}
	
	\textbf{return} the edge weights
	
    \Procedure{\update{$u, v, w$}}{ 
        \If{$u$, $v$, and $w$ form a graphlet instance}{
            prev $\leftarrow$ graphlet of the instance $(u, v, w)$ without $u\rightarrow v$ \\
            next $\leftarrow$ graphlet of the instance $(u, v, w)$ with $u\rightarrow v$ \\
            $i$ $\leftarrow$  index of the graphlet transition from prev to next \\
            increase the weight of the edge $i$ (i.e., $w_{i}$) by $1$
        }
    }
\end{algorithm}

\smallsection{Characteristic Profile (CP):} 
We characterize the evolution of local structure in a graph $\SG$ using the significance of edge weights in its GTG $G=(V,E,W)$. 
In order to measure the significance, we follow the steps in \cite{milo2004superfamilies} for measuring the significance of each graphlet itself.
To this end, we construct the graphlet transition graph $\tilde{G}$ of a randomized graph $\SGR$.
Then, we measure the significance $SP_i$ of each edge weight $w_i$ in $G$ as follows:
\begin{equation}\label{eq1}
    SP_i := \frac{w_i - \tilde{w}_i}{w_i+\tilde{w}_i + \epsilon},
\end{equation}
where $\tilde{w}_i$ is the corresponding edge weight in $\tilde{G}$, and $\epsilon$ is a constant, which we fix to $4$. For $\tilde{w}_i$, we generate 50 instances of randomized graphs and we use the average edge weights in them. Lastly, we normalize each significance as follows:
\begin{equation}\label{eq2}
    CP_{i} := {SP_{i}}/{\sqrt{\Sigma_{i=1}^{|E|} SP_{i}^2}}.
\end{equation} 
We characterize the evolution of local structures in $\SG$ using the vector of the normalized significances (i.e., $[CP_1,\cdots,CP_{|E|}]$), which we call  \textit{characteristic profile (CP)}.



\begin{figure}[t]
    \vspace{-3mm}
     \centering
     \begin{subfigure}{0.155\textwidth}
         \includegraphics[width=\textwidth]{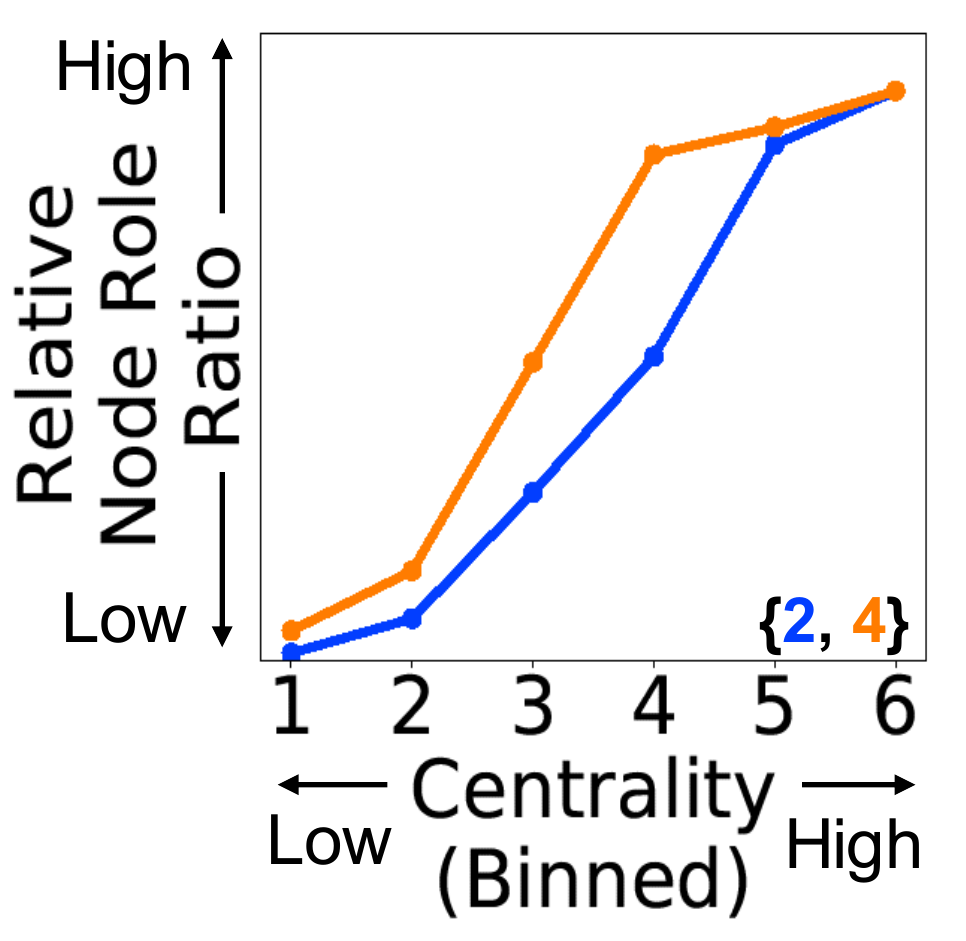}
         \caption{$d_\theta = 2$}
     \end{subfigure}
     \begin{subfigure}{0.155\textwidth}
        \includegraphics[width=\textwidth]{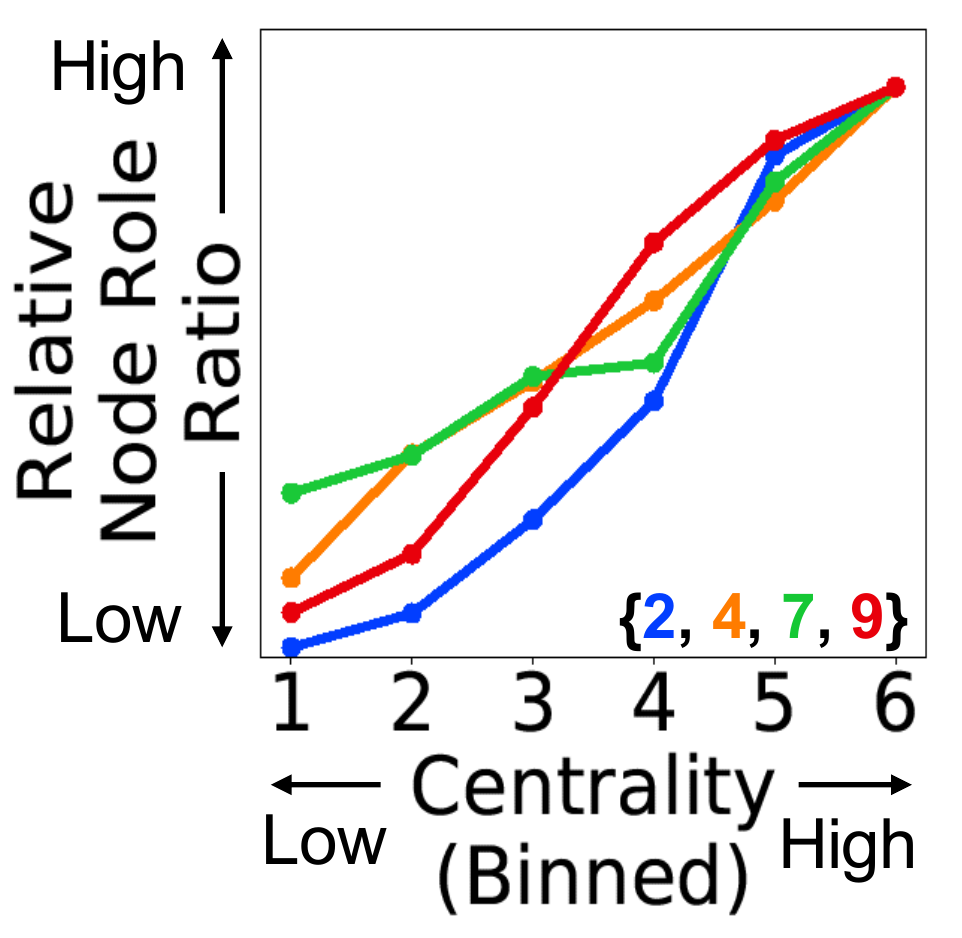}
        \caption{$d_\theta = 4$}
     \end{subfigure}
     \begin{subfigure}{0.155\textwidth}
        \includegraphics[width=\textwidth]{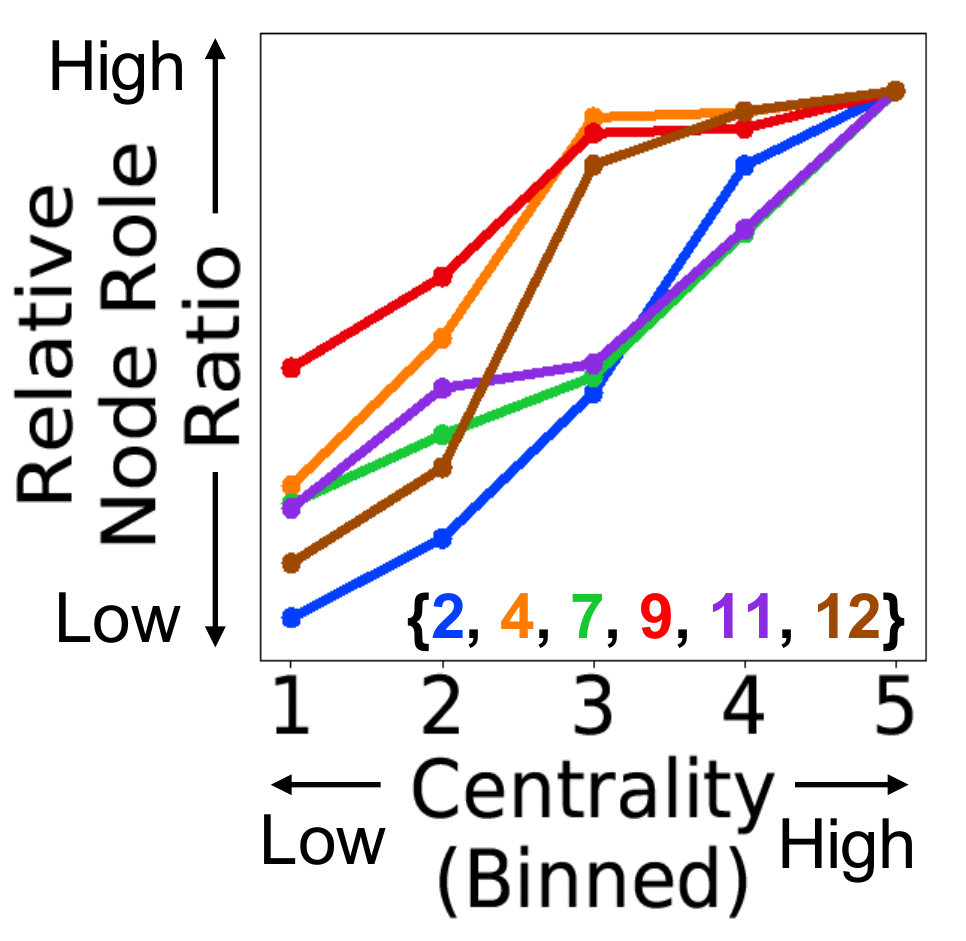}
        \caption{$d_\theta = 8$}
     \end{subfigure}
     \caption{\label{fig:trend_signal} 
     \label{fig:signal_degree}
     Example signals from the local structures of nodes regarding their future importance. The ratios of some node roles (e.g., node roles 2 and 4) at nodes monotonically increase with respect to the future in-degrees of the nodes. The ratios are rescaled so that their maximum values are the same.}
\end{figure}

\smallsection{Comparison between CPs:}
We plot the CPs of the considered real-world graphs in Table~\ref{tab:gtg}, and high levels of similarity are observed within domains.
We numerically measure the similarity between CPs using the Pearson correlation coefficients, and the results are shown in Figure~\ref{fig:domain_correlation}(a).
The correlation coefficients are particularly high between graphs from the same domain, and specifically the domains can be classified with $97.2\%$ accuracy if we use the best threshold of the correlation coefficient ($0.58$).
The results demonstrate that CPs accurately characterize the evolution of local structures.
Our observations are summarized in Observation~\ref{obs:transition:domain}.

	
	
	
	
	

\vspace{0.5mm}
\noindent\fbox{%
        \parbox{0.98\columnwidth}{%
        \vspace{-2mm}
        \begin{observation} \label{obs:transition:domain}
            The evolution patterns of local structures are similar
            in real-world graphs from the same domains.
        \end{observation}
        \vspace{-2mm}
        }%
    }







\smallsection{Comparison with Other Methods:} We evaluate three other graph characterization methods, as we evaluate ours in the right above paragraph.
In Figure~\ref{fig:domain_correlation}(b), we provide the correlation coefficients between the CPs obtained from the count of the instances of each graphlet~\cite{milo2004superfamilies}. 
Note that the email/message graphs (blue) and the online Q/A graphs (green) are not distinguished clearly. Numerically, with the best threshold of correlation coefficient ($0.95$), the classification accuracy is $83.3\%$.

We also compute the similarity between the considered real-world graphs using Graphlet-orbit Transition (GoT) \cite{aparicio2018graphlet} and Orbit Temporal Agreement (OTA) \cite{aparicio2018graphlet}, which are also based on transitions between graphlets (see Section~\ref{section:relwork} for details). 
Our way of characterization has the following major advantages over them:
\begin{itemize}
    \item \textbf{(1) Speed:} Empirically, GoT and OTA are up to $10\times$ slower than our method, as shown in Appendix~\ref{sec:appendix:compare_got_ota}.
The time complexity of them is proportional to the sum of the counts of graphlet instances in all used snapshots, while the time complexity of Algorithm~\ref{alg:count_graphlet} is proportional only the to the count of graphlet instances in the last snapshot (Theorem~\ref{thm:time:transition}).
    \item \textbf{(2) Space Efficiency:} GoT and OTA run out of memory in the two largest graphs (Patent and Stackoverflow), as shown in Appendix~\ref{sec:appendix:compare_got_ota}, while our method does not. They need to store all graphlet instances in each considered snapshot for comparison with those in the next snapshot, while Algorithm~\ref{alg:count_graphlet} maintains only the latest snapshot without having to store graphlet instances.
    \item \textbf{(3) Characterization Accuracy:} The best classification accuracies computed using the considered real-world graphs (except for Patent and Stackoverflow for which GoT and OTA run out of memory) are $81.0\%$ (GoT) and $85.7\%$ (OTA), which is lower than our classification accuracy ($97.2\%$). Detailed results are given in Appendix~\ref{sec:appendix:compare_got_ota}. Note that GoT and OTA approximate the counts of transitions between graphlets based on a small number of snapshots, while Algorithm~\ref{alg:count_graphlet} exactly counts the transitions.
\end{itemize}

In summary, \textbf{our way of characterizing temporal graphs using GTGs distinguishes the domains of temporal graphs most accurately with the accuracy of $97.2\%$.} The accuracies of the other methods are $83.3\%$, $81.0\%$, and $85.7\%$.

\section{Node Level Analysis} \label{section:node}
In this section, we study how local structures around nodes are related to their future importance.
Then, we enhance the predictability of future node centrality using the relations.

\subsection{Patterns}
We characterize the local structures of nodes using node roles and examine their relation to the nodes' future centrality.

\begin{table}[t]
\vspace{-3mm}
\caption{\label{tab:signal_degree}
The absolute value of the Spearman's rank correlation coefficients between node role ratios and future centralities (averaged over all node roles and all datasets for each centrality measure) and each value of the threshold $d_\theta$.
As the number of node neighbors increases (i.e., $d_\theta$ increases), the local-structural signals about future centralities become stronger (i.e., the absolute values increase).
}
\resizebox{\columnwidth}{!}{
    \begin{tabular}{|c|c|c|c|c|c|}
        \hline
        $d_\theta$ & Degree & Betweenness & Closeness & PageRank & Edge Betweenness \\
        \hline
        2        & 0.640  & 0.697       & 0.682     & 0.663    & 0.546            \\
        4        & 0.721  & 0.723       & 0.712     & 0.704    & 0.558            \\
        8        & 0.816  & 0.793       & 0.759     & 0.701    & 0.599            \\
        \hline
    \end{tabular}}
\end{table}

\smallsection{Local Structures of Nodes:}
Given a temporal graph $\SG$, we characterize the local structure of each node $v$ in their early stage by measuring the ratio of each node role at $v$ in the snapshot at time $t$ when the in-degree of $v$ first reaches a threshold $d_\theta$.
That is, each node $v$ is represented as a $30$-dimensional vector whose $i$-th is $\SNTi(v)/(\sum_{j=1}^{30}\SNTj(v))$  (see Section~\ref{sec:prelim:concept} for $\SNTi(v)$).


\smallsection{Future Importance of Nodes:}
Given a temporal graph $\SG$, as future importance of each node, we measure its in-degree, node betweenness centrality~\cite{freeman1977set}, closeness centrality~\cite{bavelas1950communication}, and PageRank~\cite{page1999PageRank} in the last snapshot of $\SG$.
Based on each centrality measure, we divide the nodes in $\SG$ into six groups (Group 1: top 50-100\%, Group 2: top 30-50\%, Group 3: top 10-30\%, Group 4: top 5-10\%, Group 5: top 1-5\%, and Group 6: top 0-1\%).

\begin{figure*}[t]
    \vspace{-3mm}
    \centering
		\includegraphics[width=0.4\textwidth]{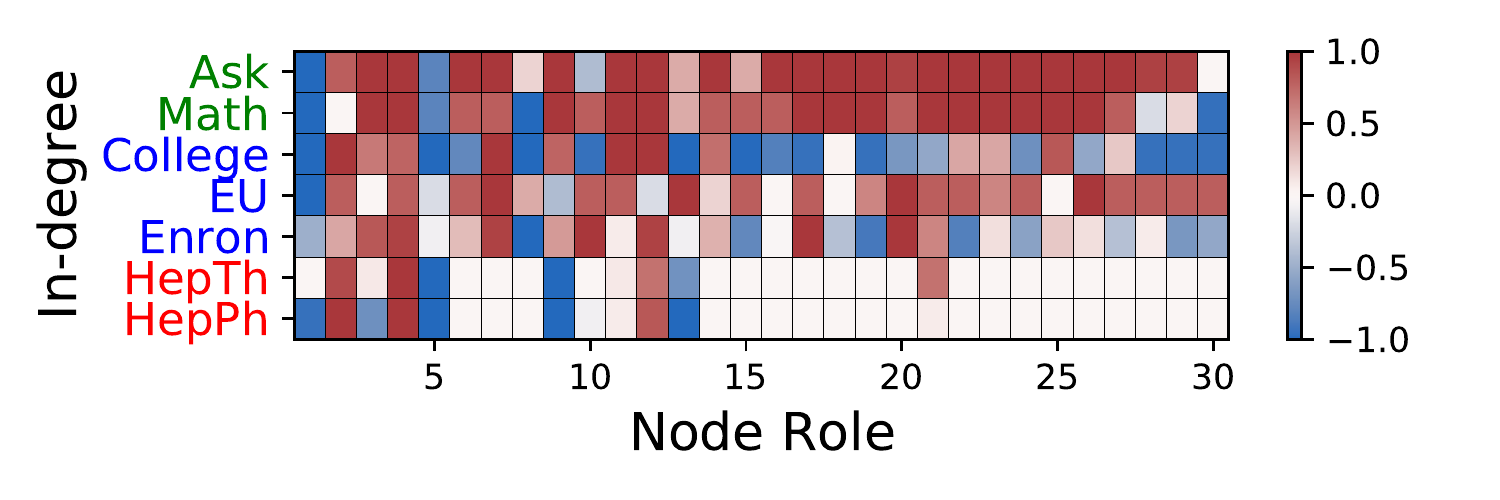}
		\includegraphics[width=0.4\textwidth]{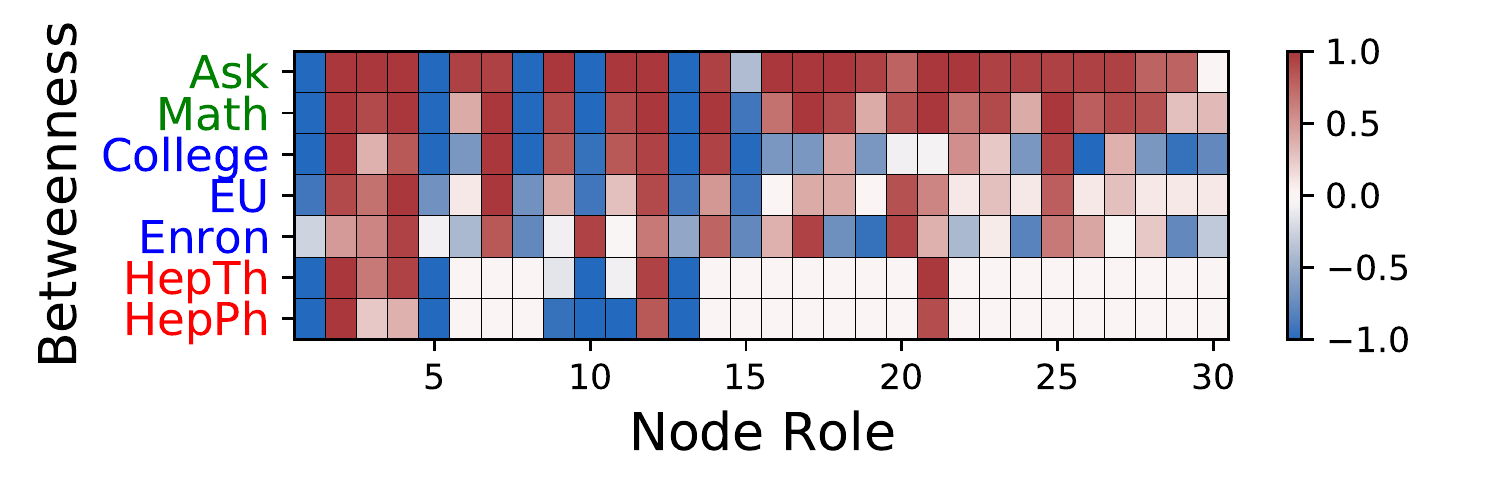} \\
            \vspace{-2mm}
		\includegraphics[width=0.4\textwidth]{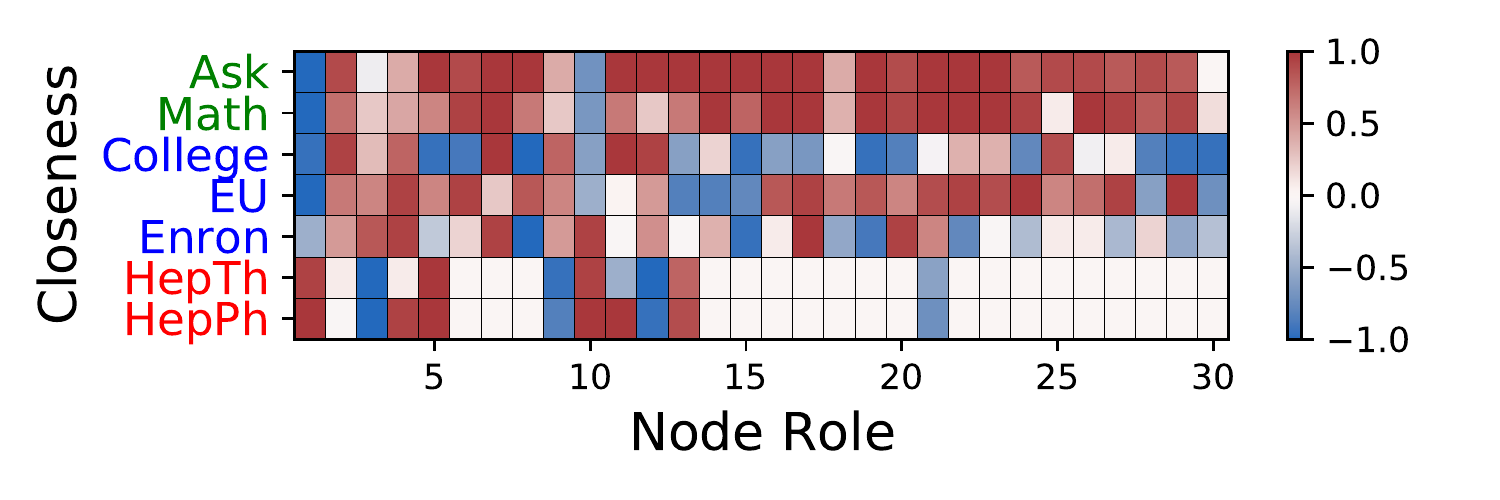}
            \includegraphics[width=0.4\textwidth]{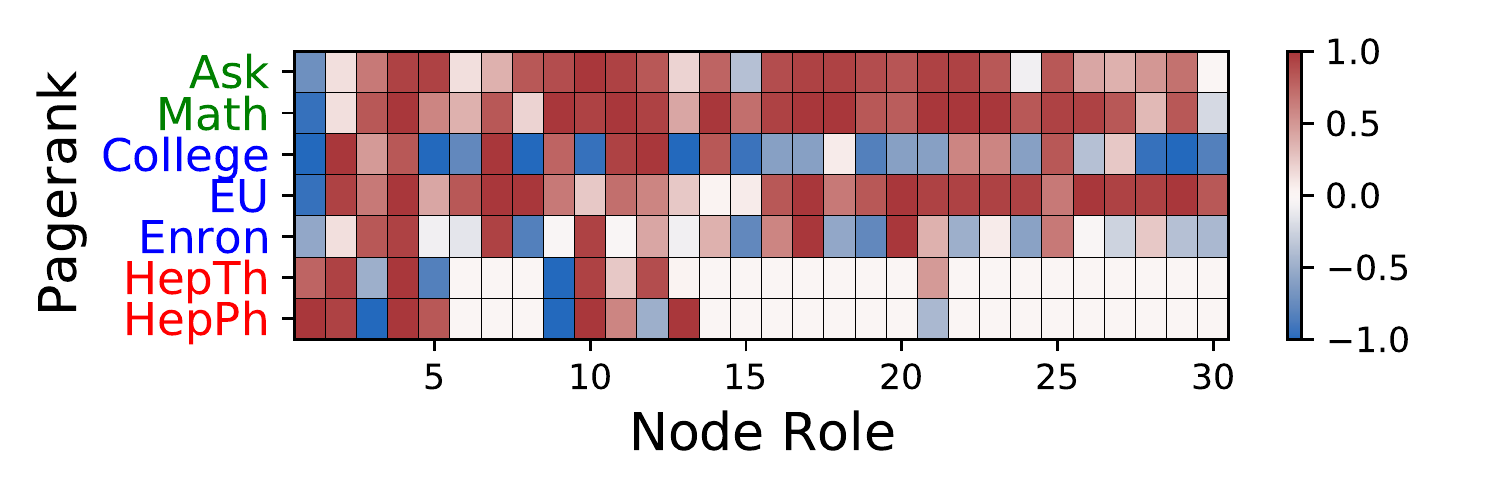} \\
            \vspace{-3mm}
	\caption{\label{fig:node_signals} 
     The Spearman's rank correlation coefficient  between node role ratios (when nodes have in-degree four, i.e., $d_\theta$ = 4)
     and future node centralities.
     The darker a cell is, the larger the absolute value of the corresponding coefficient is. Note that the absolute values of most coefficients are significantly greater than $0$.
     }
\end{figure*}

\smallsection{Finding Signals:}
For each group, we average the ratio vectors of the nodes in the group. 
Figure~\ref{fig:signal_degree} shows some averaged ratios when in-degree is used as the centrality measure. Note that the ratios of node roles 2 and 4 monotonically grow as future centrality increases, regardless of $d_\theta$ values. That is, the ratios of node roles 2 and 4 give a consistent signal regarding the nodes' future in-degree.

In Figure~\ref{fig:node_signals}, we report the Spearman's rank correlation coefficient~\cite{zwillinger1999crc} between each averaged ratio and the future centralities of nodes (specifically, the above group numbers between 1 and 6).
We also report in Table~\ref{tab:signal_degree} the absolute value of the coefficients (averaged over all node roles and all datasets) for each centrality measure and each value of the threshold $d_\theta$. Note that the average values are significantly greater than $0$ and specifically around $0.7$; and they
increase as $d_\theta$ increases, as summarized in Observation~\ref{obs:node:signal}.
%

\noindent\fbox{%
        \parbox{0.98\columnwidth}{
        \vspace{-2mm}
        \begin{observation}\label{obs:node:signal}
        In real-world graphs,
           the local structures of nodes in their early stage provide a signal regarding their future importance. The signals become stronger as nodes have more neighbors.
        \end{observation}
        \vspace{-2mm}
        }
     }
\vspace{0.5mm}

\subsection{Prediction} \label{subsection:node_prediction}

Based on the observations above, we predict the future centrality of nodes using the counts of their roles at them in their early stage. 


\smallsection{Problem Formulation:} 
We formulate the prediction problem as a classification problem, as described in Problem~\ref{problem:node}.

\vspace{0.5mm}

\noindent\fbox{%
        \parbox{0.98\columnwidth}{%
        \vspace{-2mm}
        \begin{problem}[Node Centrality Prediction] \label{problem:node}
        \noindent\begin{itemize}
            \item \textbf{Given:} the snapshot $\SGTv$ of the input graph when the in-degree of a node $v$ first reaches $d_\theta$,
            \item \textbf{Predict:} whether the centrality of the node $v$ belongs to the top $20\%$ in the last snapshot of $\SG$.
        \end{itemize}
        \end{problem}
        \vspace{-2mm}
        }%
    }
\vspace{0.5mm}

\noindent As the centrality measure, we use in-degree, betweenness centrality, closeness centrality, and PageRank.
As $d_\theta$, we use $2$, $4$, or $8$.

\smallsection{Input Features:} 
For each node $v$, we consider the snapshot $\SGT$ of the input graph $\SG$ when the in-degree of $v$ first reaches $d_\theta$. That is, $t=\tv$ and $\SGT=\SGTv$.
Then, we extract the following sets of input features for $v$: 
\bit
    \item \textbf{Local-NR:} The count of each node role at $v$ in $\SGT$. That is, $[\SNTone(v),$ $\SNTtwo(v), \cdots, \SNTthirty(v)]$ (see Section~\ref{sec:prelim:concept} for $\SNTi(v)$).
    \item \textbf{Local-NPP \cite{yang2014predicting}:} In $\SGT$, we compute (1) the count of triangles at $v$, (2) the count of wedges centered at $v$, (3) the count of wedges ended at $v$.
    \item \textbf{Global-Basic:} Counts of nodes and edges in the snapshot.
    \item \textbf{Global-NR:} We compute the $30$-dimensional vector whose $i$-th entry $\SNTi(v)/(\sum_{j=1}^{30}\SNTj(v))$
    is the ratio of each node role at $v$ in $\SGT$.
    Then, we standardize (i.e., compute the $z$-score of) the role ratio vector using the mean and standard deviation from the role ratio vectors (in $\SGT$) of all nodes with degree $d_\theta$ in $\SGT$. The features in Local-NR are also included.
    \item \textbf{Global-NPP~\cite{yang2014predicting}:} In $\SGT$, we compute (1) the number of edges not incident to $v$ and (2) the number of non-adjacent node pairs where one is a neighbor of $v$ and the other is neither a neighbor of $v$ nor $v$ itself.
    The features in Local-NPP are also included.
    \item \textbf{ALL:}  All of \textbf{Global-NR}, \textbf{Global-NPP}, and \textbf{Global-Basic}.
\eit

Note that we categorize the above sets into global and local depending on whether global information in $\SGT$ (i.e., the number of all nodes in $\SGT$) is used or only local information at $v$ is used.

\smallsection{Prediction Method:}  
As the classifier, we use the \textit{random forest} model from the Scikit-learn library. The model has 30 decision trees with a maximum depth of 10.

\smallsection{Evaluation Method:} 
We use $80\%$ of the nodes for training and the remaining $20\%$ for testing. 
We evaluate the predictive performance in terms of F1-score, accuracy, and \textit{Area Under the ROC curve} (AUROC). A higher value indicates better prediction performance.


\smallsection{Result:}
Table~\ref{tab:pred-node-overall} shows the predictive performance from each set of input features when $d_\theta=2$, and Table~\ref{tab:pred-node-overall-degree} shows how the performance depends on the in-degree threshold $d_\theta$.
In the tables, we report the mean of each prediction performance over $10$ runs in the 7 datasets in Section~\ref{section:datasets} except for the two largest ones (i.e., Patent and Stackoverflow).
From the results, we draw the following observations.

\vspace{0.5mm}

\noindent\fbox{%
        \parbox{0.98\columnwidth}{%
        \vspace{-2mm}
        \begin{observation} \label{obs:node:pred_nr_npp}
            Among local features, 
            the counts of node roles at each node (\textbf{Local-NR}) are more informative than (\textbf{Local-NPP}) for future importance prediction.
        \end{observation}
        \begin{observation}\label{obs:node:pred_all}
            The considered sets of features are complementary to each other. Using them all (\textbf{ALL}) leads to the best predictive performance in most cases.
        \end{observation}
        \begin{observation}\label{obs:node:degree}
            As nodes have more neighbors, their future importance can be predicted more accurately.
        \end{observation}
        \vspace{-2mm}
        }%
    }


\begin{table}
	\centering
	\caption{\label{tab:pred-node-overall} 
    F1-score, accuracy, and AUROC on the task of predicting future node importance when $d_\theta=2$ averaged over the 7 considered real-world graphs.
	Among local features, using \textbf{Local-NR} yields better performance than using \textbf{Local-NPP} in all settings. 
	Using  \textbf{ALL} leads to the best performance in most cases, indicating that the considered sets of features are complementary to each other. Detailed results on each dataset can be found in Appendix~\ref{sec:appendix:node:details}.
	}
	\resizebox{\columnwidth}{!}{
		\begin{tabular}{|c|c|c|c|c|c|c|}
		\hline
			Target & \multicolumn{3}{c|}{Degree} & \multicolumn{3}{c|} {Betweenness}  \\ 
			\hline
			Measure & F1-score & Accuracy & AUROC & F1-score & Accuracy & AUROC \\
			\hline
			Local-NR     & 0.39 & 0.69 & 0.68 & 0.59 & 0.83 & 0.82 \\
			Local-NPP    & 0.38 & 0.68 & 0.64 & 0.58 & 0.81 & 0.79 \\
			\hline
			Global-NR    & \bf{0.57} & \bf{0.74} & \bf{0.78} & 0.64 & 0.84 & 0.85 \\
			Global-NPP   & \bf{0.57} & 0.73      & 0.77      & 0.64 & 0.84 & 0.85 \\
			Global-Basic & 0.50      & 0.72      & 0.73      & 0.24 & 0.73 & 0.67 \\
			\hline
			ALL          & \bf{0.57}      & \bf{0.74} & \bf{0.78} & \bf{0.65} & \bf{0.85} & \bf{0.86}\\
			\hline
			\hline
			Target & \multicolumn{3}{c|}{Closeness} & \multicolumn{3}{c|} {PageRank}  \\ 
			\hline
			Measure & F1-score & Accuracy & AUROC & F1-score & Accuracy & AUROC \\
			\hline
			Local-NR     & 0.51 & 0.76 & 0.78 & 0.42 & 0.73 & 0.73  \\
			Local-NPP    & 0.43 & 0.70 & 0.69 & 0.37 & 0.69 & 0.67 \\
			\hline
			Global-NR    & 0.68 & 0.82 & 0.87 & 0.54 & \bf{0.75} & \bf{0.79} \\
			Global-NPP   & 0.66 & 0.80 & 0.85 & 0.54 & 0.74 & 0.78 \\
			Global-Basic & 0.59 & 0.75 & 0.79 & 0.47 & 0.71 & 0.74 \\
			\hline
			ALL          & \bf{0.69} & \bf{0.83} & \bf{0.88} & \bf{0.56} & \bf{0.75} & \bf{0.79} \\
			\hline
		\end{tabular}
	}
\end{table}

\smallsection{Feature Importance:} Additionally, we measure the importance of each feature in the set \textbf{ALL} using \textit{Gini-importance}~\cite{loh2011classification}, and we report the top five important features in Table~\ref{tab:feature_importance} in Appendix~\ref{sec:appendix:feature_importance}.

\vspace{0.5mm}

\noindent\fbox{%
        \parbox{0.98\columnwidth}{%
        \vspace{-2mm}
        \begin{observation}
            Strong predictors vary depending on centrality measures. For example, for betweenness centrality, the counts of node roles as bridges (i.e., \textbf{Local NR-4} and \textbf{Global NR-4}) are strong.
        \end{observation}
        
        \vspace{-2mm}
        }%
    }
\begin{table}[t]
 	\vspace{-2mm}
	\centering
	\caption{\label{tab:pred-node-overall-degree} 
    Average F1-score, accuracy, and AUROC on the task of predicting future node importance depending on $d_\theta$ (i.e., in-degree of nodes when their input features are extracted). 
	The overall performance improves with respect to $d_\theta$ in most cases. That is,
	as nodes have more neighbors, their future importance can be predicted more accurately.
 Detailed results on each dataset can be found in Appendix~\ref{sec:appendix:node:details}.
	}
	\resizebox{\columnwidth}{!}{
		\begin{tabular}{|c|c|c|c|c|c|c|}
		\hline
			Target & \multicolumn{3}{c|}{Degree} & \multicolumn{3}{c|} {Betweenness} \\ 
			\hline
			Measure & F1-score & Accuracy & AUROC & F1-score & Accuracy & AUROC \\
			\hline
			ALL $(d_\theta=2)$ & 0.59 & 0.74 & 0.78 & 0.65 & 0.85 & 0.86 \\
			ALL $(d_\theta=4)$ & 0.69 & 0.78 & 0.79 & 0.73 & 0.83 & 0.87 \\
			ALL $(d_\theta=8)$ & 0.80 & 0.81 & 0.86 & 0.82 & 0.85 & 0.90 \\
			\hline
			\hline
			Target & \multicolumn{3}{c|}{Closeness} & \multicolumn{3}{c|} {PageRank} \\ 
			\hline
			Measure & F1-score & Accuracy & AUROC & F1-score & Accuracy & AUROC \\
			\hline
			ALL $(d_\theta=2)$ & 0.69 & 0.83 & 0.88 & 0.55 & 0.75 & 0.79  \\
			ALL $(d_\theta=4)$ & 0.78 & 0.83 & 0.89 & 0.73 & 0.77 & 0.80 \\
			ALL $(d_\theta=8)$ & 0.86 & 0.88 & 0.92 & 0.85 & 0.85 & 0.83 \\
			\hline
		\end{tabular}
	}
\end{table}

\begin{figure}[t]
    \vspace{-3mm}
     \centering
     \includegraphics[width=0.45\textwidth]{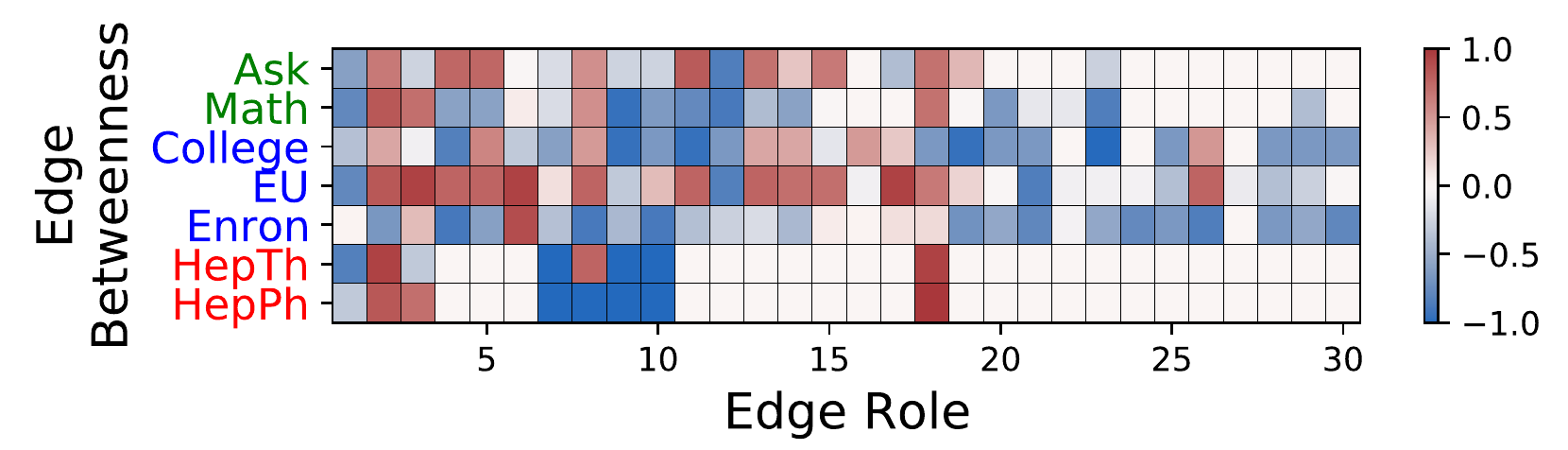} \\
     \vspace{-2mm}
     \caption{\label{fig:edge_signal}
     The Spearman's rank correlation coefficient between edge role ratios 
     (when endpoints have in-degree $4$ in total, i.e., $d_\theta$ = 4)
     and future edge centralities.
     The darker a cell is, the larger the absolute value of the corresponding coefficient is. Note that the absolute values of many coefficients are significantly greater than $0$, while they tend to be smaller than those in Figure~\ref{fig:node_signals}.
     }
     
\end{figure}

\section{Edge Level Analysis}\label{section:edge}

In this section, we investigate the signal of local structures of each edge regarding their future centrality, and based on the signal, we predict the future importance of edges.  

We generally follow the procedures in Section~\ref{section:node}, except for the following differences: (a) we examine the ratios of edge roles at each edge $u\rightarrow v$ when the sum of the in-degrees of $u$ and $v$ becomes $d_\theta$, (b) we use edge betweenness centrality~\cite{freeman1977set} as the importance measure, (c) we formulate the problem of predicting future edge importance as described in Problem~\ref{problem:edge}, (d)
we extract feature sets \textbf{Local-ER} and \textbf{Global-ER} using the (relative) counts of edge roles at edges as we extract \textbf{Local-NR} and \textbf{Global-NR}, and (e) we union \textbf{Global-ER} and \textbf{Global-Basic} for $\textbf{ALL}$.

\vspace{0.5mm}

\noindent\fbox{%
        \parbox{0.98\columnwidth}{%
        \vspace{-2mm}
        \begin{problem}[Edge Centrality Prediction] \label{problem:edge}
        \noindent\begin{itemize}
            \item \textbf{Given:} the snapshot $\SGTe$ of the input graph when the sum of the in-degrees of the endpoints of each edge first reaches $d_\theta$,
            \item \textbf{Predict:} whether the centrality of each edge belongs to the top $20\%$ in the last snapshot of $\SG$.
        \end{itemize}
        \end{problem}
        \vspace{-2mm}
        }%
    }

\vspace{0.5mm}

From Figure~\ref{fig:edge_signal}, Table~\ref{tab:signal_degree}, and Table~\ref{tab:pred-edge-overall}, we draw the following observations.


\noindent\fbox{%
        \parbox{0.98\columnwidth}{%
        \vspace{-2mm}
        \begin{observation}\label{obs:edge_signal}
        In real-world graphs, the signals from the local structures of edges in their early stage regarding their future importance are weaker, compared to the signals that from the local structures of nodes (see Figure~\ref{fig:edge_signal}).
        \end{observation}
        
        \begin{observation}\label{obs:edge_stronger}
        However, the signals become stronger as the edges are better connected, leading to better prediction performance (see Tables~\ref{tab:signal_degree} and~\ref{tab:pred-edge-overall}).
        \end{observation}
        
        \begin{observation}\label{obs:edge_pred}
           The features from edge roles (\textbf{Local-ER} and \textbf{Global-ER}) are more informative than simple global statistics (\textbf{Global-Basic}) for future importance prediction (see Table~\ref{tab:pred-edge-overall}).
        \end{observation}
        \vspace{-2mm}
        }%
     }

\begin{table}[t]
	\vspace{-3mm}
	\centering
	\caption{\label{tab:pred-edge-overall} 
    F1-score, accuracy, and AUROC on the task of predicting future edge importance averaged over the 7 considered real-world graphs.
	Using edge role-based features (\textbf{Local-ER} and \textbf{Global-ER}) yields better performance than using \textbf{Global-Basic} in most settings. 
	The overall performance improves with respect to $d_\theta$.
	That is, as edges are better connected, their future importance is predicted more accurately.
     Detailed results on each dataset can be found in Appendix~\ref{sec:appendix:edge:details}
	}
	\scalebox{0.90}{
		\begin{tabular}{|c|c|c|c|}
		\hline
			Target & \multicolumn{3}{c|}{Edge betweenness}  \\ \hline
            Measure & F1-Score & Accuracy   &  AUROC \\ \hline
            Local-ER $(d_\theta = 2)$       & 0.45 & 0.78 & 0.76 \\
            Global-ER $(d_\theta = 2)$      & 0.47 & 0.81 & 0.78 \\
            Global-Basic $(d_\theta = 2)$   & 0.42 & 0.79 & 0.73 \\
            ALL $(d_\theta = 2)$            & 0.50 & 0.80 & 0.75 \\
			\hline
			ALL $(d_\theta = 2)$            & 0.50 & 0.80 & 0.75 \\
			ALL $(d_\theta = 4)$            & 0.53 & 0.82 & 0.84 \\
			ALL $(d_\theta = 8)$            & 0.52 & 0.85 & 0.85 \\
			\hline
		\end{tabular}
	}
\end{table}


\section{Related Work}\label{section:relwork}

Previous studies on temporal graph analysis are largely categorized into (a) designing algorithms for streaming graphs~\cite{lee2020temporal, eswaran2018spotlight, liben2007link, mcgregor2014graph}, (b) discovering temporal patterns in graphs~\cite{leskovec2005graphs, akoglu2008rtm, beyer2010mechanistic, akoglu2010structure, bahulkar2016analysis}, and (c) generating graphs with realistic dynamics~\cite{barabasi1999emergence, leskovec2010kronecker, akoglu2008rtm}. This work belongs to the second category.

Studies in this category have revealed (a) universal temporal patterns, such as densification~\cite{leskovec2005graphs}, shrinking diameter~\cite{leskovec2005graphs}, and power-laws between principle eigenvalues and edge counts~\cite{akoglu2008rtm}; and (b) domain-specific patterns in hyperlink networks~\cite{broder2011graph}, metabolic networks (e.g., biochemical reactions and  protein interactions)~\cite{beyer2010mechanistic}, communication networks (e.g., phone calls and
texts)~\cite{hidalgo2008dynamics,akoglu2010structure}, and friendship networks~\cite{bahulkar2016analysis}.

In particular, for the analysis of local structures, the concept of graphlets \cite{prvzulj2007biological} (i.e., the sets of isomorphic small subgraphs with a predefined number of nodes) has been extended to temporal graphs. The extensions, which are called \textit{temporal network motifs}, have multiple variants.
Kovanen et al.~\cite{kovanen2011temporal}  defined them as sets of temporal subgraphs with a fixed number of nodes that are (a) topologically equivalent, (b) temporally equivalent (specifically, relative orders of constituent edges are identical), (c)  consecutive (specifically, constituent edges are consecutive for every node), and (d) temporally local (specifically, arrival times of consecutive edges are close enough).
Hulovaty et al.~\cite{hulovatyy2015exploring} ignores (c); and Paranjape et al.~\cite{paranjape2017motifs} 
ignores (c) and relaxes (d) by restricting only the time difference between the first edge and the last edge.
Note that all these notions focus on temporally local subgraphs, and thus they are suitable only for analyzing short-term dynamics. 

For long-term dynamics in local structures,
David et al.~\cite{aparicio2018graphlet} proposed Graph-orbit Transition (GoT) and Orbit Temporal Agreement (OTA), which characterize the dynamic of a temporal graph by approximately counting the number of transitions between node roles.
However, due to high computational overhead, only a small fraction of snapshots can be compared for estimating the counts of transitions, and as a result, their  characterization powers are significantly weaker than our characterization method using GTGs (see Section~\ref{section:graph:transition}). Recall that our method counts ``every'' transition between graphlets, and it is still significantly faster than GOT and OTA (see Section~\ref{sec:appendix:compare_got_ota} in Appendix).


For predicting the future in-degree of nodes, Yang et al.~\cite{yang2014predicting} proposed to use five features obtained from graphlets with three nodes (see Section~\ref{subsection:node_prediction} for descriptions).
As shown empirically, our proposed features tend to provide better prediction performance than these five features, and more importantly, they are complementary to each other.
Faisal and Milenkovi{\'c}~\cite{faisal2014dynamic} aimed to detect aging-related nodes, whose topological properties (e.g,. graphlet counts) change highly over time, in the gene expression process. 

On the algorithmic aspect, a great number of algorithms have been developed for the problem of counting the instances of each graphlet, which is also known as the subgraph counting problem. 
As suggested in a survey on subgraph counting \cite{ribeiro2021survey},  subgraph-counting algorithms are largely categorized into exact counting \cite{milo2002network, schiller2015stream, ortmann2017efficient, ahmed2017graphlet} and approximated counting  \cite{wernicke2005faster, aslay2018mining}.
Those in the first category are further categorized into enumeration-based approaches \cite{milo2002network, schiller2015stream}, matrix-based approaches \cite{ortmann2017efficient}, and decomposition-based approaches \cite{ahmed2017graphlet}. 
Algorithm~\ref{alg:track_graphlet} belongs to the first subcategory, and it achieves the optimal time complexity achievable by those in this subcategory, as discussed in the beginning of Section~\ref{section:graph:time}.
It is adapted from StreaM \cite{schiller2015stream}, which maintains the counts of the instances of $4$-node undirected graphlets in a fully dynamic graph stream (i.e., a stream of edge insertions and deletions). 

\section{Conclusion}\label{section:conclusion}

In this work, we examined the long-term evolution of local structures captured by graphlets at the graph, node, and edge levels.
We summarize our contribution as follows:
\bit
\item \textbf{Patterns:} We examined various patterns regarding the dynamics of local structures in temporal graphs. For example, the distributions of graphlets over time in real-world graphs differ significantly from those in random graphs, and the transitions between graphlets are surprisingly similar in graphs from the same domains.
Moreover, local structures at nodes and edges in their early stages provide strong signals regarding their future importance.


\item \textbf{Tools:} We introduced graphlet transition graphs, and we demonstrated that it is an effective tool for measuring the similarity between temporal graphs of different sizes.

\item \textbf{Predictability:} We enhanced the accuracy of predicting the future importance of nodes and edges by introducing new features based on node roles and edge roles. The features are also complementary to global graph statistics.
\eit

\noindent\textbf{Reproducibility:}  The code and the datasets are available at 
\url{https://github.com/deukryeol-yoon/graphlets-over-time}.

\begin{table*}[t]
    \vspace{-2mm}
    \caption{\label{tab:feature_importance} Results of feature importance analysis. We report the five strongest predictors and their Gini importance.  }
	\centering
	\resizebox{\textwidth}{!}{
		\begin{tabular}{ | c | >{\centering\arraybackslash}p{2.2cm} p{0.01cm} >{\centering\arraybackslash}p{0.5cm} | >{\centering\arraybackslash}p{2.2cm} p{0.01cm} >{\centering\arraybackslash}p{0.5cm} | >{\centering\arraybackslash}p{2.2cm} p{0.01cm} >{\centering\arraybackslash}p{0.5cm} | >{\centering\arraybackslash}p{2.2cm} p{0.01cm} >{\centering\arraybackslash}p{0.5cm} | >{\centering\arraybackslash}p{2.2cm} p{0.01cm} >{\centering\arraybackslash}p{0.5cm} |}
		\hline
		    Centrality & \multicolumn{3}{c|}{Rank 1} & \multicolumn{3}{c|}{Rank 2} & \multicolumn{3}{c|}{Rank 3} & \multicolumn{3}{c|}{Rank 4} & \multicolumn{3}{c|}{Rank 5} \\
		\hline
		Degree & \# of edges && 0.07 & Global-NPP 2 && 0.06 & \# of nodes && 0.06 & Global-NR 3 && 0.04 & Global-NR 10 && 0.03 \\
		Betweenness & Local-NPP 2 && 0.10 & Local-NR 4 && 0.09 & Global-NR 4 && 0.08 & Local-NR 9 && 0.06 & Global-NR 3 && 0.05 \\
		Closeness & Global-NR 5 && 0.09 & Local-NR 5 && 0.07 & \# of edges && 0.07 & Global-NPP 2 && 0.06 & \# of nodes && 0.06 \\
		PageRank & Local-NR 1 && 0.07 & \# of edges && 0.06 & Global-NPP 2 && 0.06 & \# of nodes && 0.05 & Global-NR 1 && 0.05 \\
		Edge betweenness & Global-ER 7 && 0.11 & Global-ER 2 && 0.09 & Global-ER 3 && 0.09 & \# of nodes && 0.07 & Local-ER 7 && 0.07 \\
		\hline
		\end{tabular}
	}
\end{table*}
\bibliographystyle{ACM-Reference-Format}
\bibliography{reference}

\appendix
\renewcommand{\thesection}{\Alph{section}.\arabic{section}}
\setcounter{section}{0}
\begin{figure}[t]
    \begin{subfigure}{0.24\textwidth}
        \includegraphics[width=\textwidth]{FIG/transition_graph/motif_transition_graph.pdf}
        \caption{Ours}
    \end{subfigure} \\
    \begin{subfigure}{0.22\textwidth}
        \includegraphics[width=\textwidth]{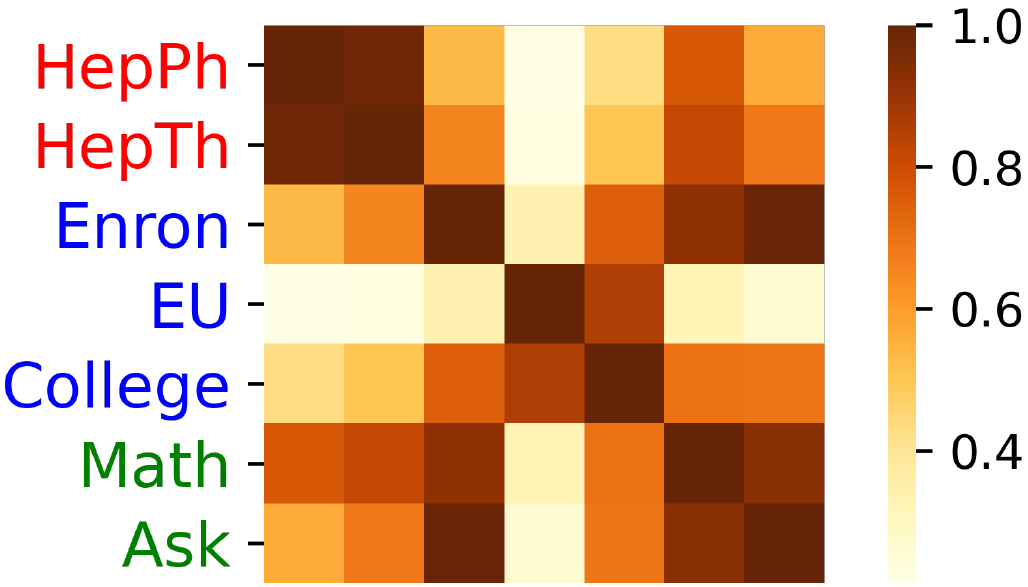}
        \caption{GoT}
    \end{subfigure}
    \begin{subfigure}{0.22\textwidth}
        \includegraphics[width=\textwidth]{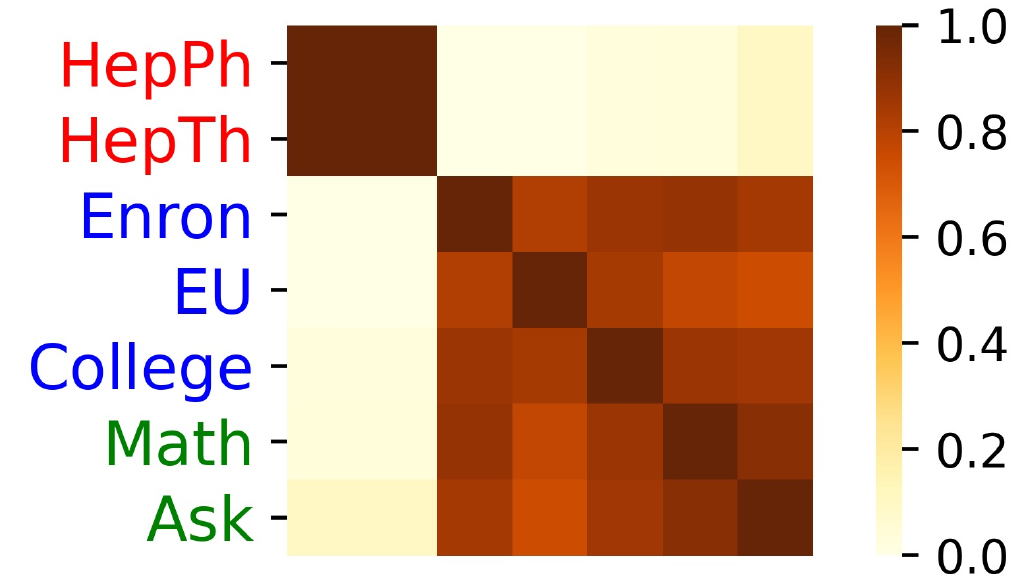}
        \caption{OTA}
    \end{subfigure}
    \caption{\label{fig:got_sim}Similarity matrices from ours, GoT, and OTA.
    The domains of graphs (distinguished by colors) are classified more accurately by ours than by GoT or OTA.}
\end{figure}

\appendix
\section{Comparison with GoT and OTA}
\label{sec:appendix:compare_got_ota}
We provide additional details regarding the comparison between our characterization method based on graphlet transition graphs (GTGs) and regarding the comparison with Graphlet-orbit Transition (GoT) and Orbit Temporal Agreement (OTA).

\smallsection{Detailed Setting:} Our experiments were conducted on a desktop with a 3.8 GHz AMD Ryzen 3900x CPU and 128GB memory. We implemented our characterization method based on GTGs in Java, and we used the official implementations for GoT and OTA provided by the authors, which were implemented in C++. In each dataset, we used $12$ snapshots with the same intervals for GoT and OTA. 

\smallsection{Output Similar Matrix:} Figure~\ref{fig:got_sim} shows the output similar matrices from our characterization method based on GTGs, GoT, and OTA.
GoT and OTA run out of memory in the two largest datasets (Patent and Stackoverflow). 
Both GoT and OTA fail to distinguish email/message graphs (blue) and online Q/A graphs (green) clearly. Numerically, with the best thresholds of similarity, the classification accuracies are $81.0\%$ (GoT) and $85.7\%$ (OTA), while the accuracy is $97.2\%$ in ours.



\begin{figure}[t]
    \includegraphics[width=0.9\columnwidth]{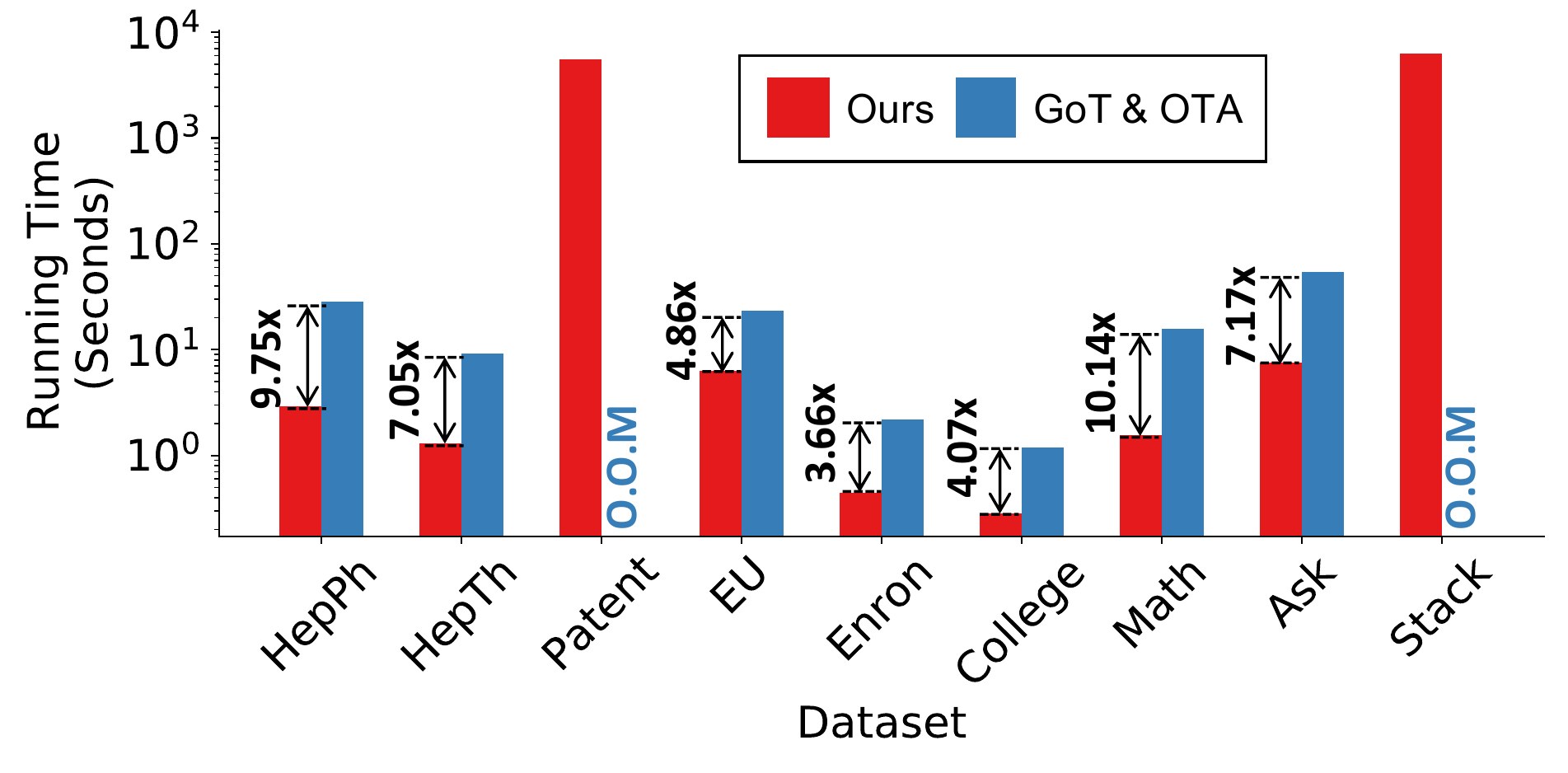}
    \vspace{-2mm}
    \caption{\label{fig:got_time} 
    Running times of ours, GoT and OTA. Ours is consistently and significantly faster than both competitors, which run out of memory in the two largest datasets.}
\end{figure}

\smallsection{Speed Comparison:} As seen in Figure~\ref{fig:got_time}, ours is faster than GoT and OTA in all the graphs. Specifically, ours is $6.68\times$ faster than the others on average. 

\section{Feature Important Analysis}
\label{sec:appendix:feature_importance}
We measure the importance of each feature in the set \textbf{ALL} (see Section~\ref{subsection:node_prediction} of the main paper) using the \textit{Gini importance}~\cite{loh2011classification}, and we report the top five important features in Table~\ref{tab:feature_importance}.

\begin{table*}[t]
	\centering
	\caption{\label{tab:node_f1}
	F1-score on the task of predicting future node importance when $d_\theta=2$.
	}
	\resizebox{\textwidth}{!}{
		\begin{tabular}{|c|c|c|c|c|c|c|c|c|c|}
		\hline
			\multirow{2}{*}{Centrality}  & \multirow{2}{*}{Feature} & \multicolumn{2}{c|}{Citation Networks} & \multicolumn{3}{c|}{Email/Message Networks} & \multicolumn{2}{c|}{Online Q/A Networks} & \multirow{2}{*}{\bf{Average}} \\\cline{3-9}
			&  & HepPh & Hepth & Email-EU & Email-Enron & Message-College & Mathoverflow & Askubuntu & \\ 
			\hline
			\multirow{6}{*}{Degree}
			& Local-NR     & 0.11$\pm$0.008       & 0.20$\pm$0.014      & 0.36$\pm$0.092      & 0.68$\pm$0.007      & 0.27$\pm$0.027      & 0.38$\pm$0.017      & \f{0.73}$\pm$0.002      & 0.39 \\
			& Local-NPP    & 0.12$\pm$0.013       & 0.19$\pm$0.016      & 0.40$\pm$0.102      & 0.60$\pm$0.011      & 0.27$\pm$0.040      & 0.36$\pm$0.016      & \s{0.72}$\pm$0.003      & 0.38 \\ \cline{2-10}
			& Global-NR    & \s{0.52}$\pm$0.013   & 0.56$\pm$0.013      & \s{0.51}$\pm$0.077  & \f{0.79}$\pm$0.004  & 0.37$\pm$0.025      & \f{0.52}$\pm$0.021  & 0.70$\pm$0.004          & \f{0.57} \\
			& Global-NPP   & 0.51$\pm$0.010       & \s{0.57}$\pm$0.018  & \s{0.51}$\pm$0.063  & 0.76$\pm$0.005      & \f{0.40}$\pm$0.030  & \f{0.52}$\pm$0.022  & 0.70$\pm$0.006          & \f{0.57} \\
			& Global-basic & 0.36$\pm$0.014       & 0.38$\pm$0.024      & 0.44$\pm$0.073      & \s{0.77}$\pm$0.008  & 0.34$\pm$0.050      & \s{0.51}$\pm$0.022  & 0.71$\pm$0.003          & \s{0.50} \\\cline{2-10}
			& ALL          & \f{0.53}$\pm$0.010   & \f{0.58}$\pm$0.013  & \f{0.52}$\pm$0.043  & \f{0.79}$\pm$0.008  & \s{0.38}$\pm$0.041  & \f{0.52}$\pm$0.026  & 0.70$\pm$0.005          & \f{0.57} \\ 
			\hline
			\multirow{6}{*}{Betweenness}
			& Local-NR     & 0.59$\pm$0.011       & 0.88$\pm$0.009      & 0.34$\pm$0.063      & 0.49$\pm$0.011      & 0.34$\pm$0.033      & \s{0.74}$\pm$0.011  & \s{0.73}$\pm$0.007  & 0.59 \\
			& Local-NPP    & 0.58$\pm$0.010       & 0.87$\pm$0.006      & 0.35$\pm$0.076      & 0.45$\pm$0.010      & 0.36$\pm$0.073      & \s{0.74}$\pm$0.011  & \s{0.73}$\pm$0.007  & 0.58 \\ \cline{2-10}
			& Global-NR    & \f{0.64}$\pm$0.007   & \f{0.90}$\pm$0.005  & 0.48$\pm$0.089      & \f{0.62}$\pm$0.014  & \s{0.38}$\pm$0.047  & \f{0.75}$\pm$0.008  & \f{0.74}$\pm$0.006  & \s{0.64} \\
			& Global-NPP   & \s{0.62}$\pm$0.010   & \s{0.89}$\pm$0.006  & \s{0.49}$\pm$0.037  & \s{0.58}$\pm$0.019  & \f{0.40}$\pm$0.034  & \f{0.75}$\pm$0.010  & \f{0.74}$\pm$0.006  & \s{0.64} \\
			& Global-basic & 0.01$\pm$0.003       & 0.27$\pm$0.028      & 0.40$\pm$0.079      & 0.36$\pm$0.017      & 0.25$\pm$0.038      & 0.32$\pm$0.013      & 0.10$\pm$0.013      & 0.24 \\ \cline{2-10}
			& ALL          & \f{0.64}$\pm$0.007   & \f{0.90}$\pm$0.007  & \f{0.53}$\pm$0.052  & \f{0.62}$\pm$0.016  & \s{0.38}$\pm$0.045  & \f{0.75}$\pm$0.010  & \f{0.74}$\pm$0.007  & \f{0.65}\\ 
			\hline
			\multirow{6}{*}{Closeness}
			& Local-NR     & 0.49$\pm$0.010       & 0.53$\pm$0.010      & 0.28$\pm$0.077      & 0.69$\pm$0.008      & 0.24$\pm$0.034      & 0.58$\pm$0.014      & \s{0.75}$\pm$0.005  & 0.51 \\
			& Local-NPP    & 0.37$\pm$0.015       & 0.51$\pm$0.014      & 0.31$\pm$0.067      & 0.46$\pm$0.010      & 0.25$\pm$0.038      & 0.43$\pm$0.017      & 0.66$\pm$0.006      & 0.43 \\\cline{2-10}
			& Global-NR    & \s{0.84}$\pm$0.006   & \s{0.75}$\pm$0.007  & 0.47$\pm$0.046      & \f{0.83}$\pm$0.008  & \s{0.38}$\pm$0.055  & \f{0.69}$\pm$0.024  & \f{0.81}$\pm$0.002  & \s{0.68} \\
			& Global-NPP   & 0.83$\pm$0.008       & 0.74$\pm$0.007      & \s{0.52}$\pm$0.047  & \s{0.76}$\pm$0.008  & \f{0.39}$\pm$0.022  & \s{0.64}$\pm$0.019  & 0.71$\pm$0.005      & 0.66 \\
			& Global-basic & 0.82$\pm$0.004       & 0.70$\pm$0.010      & 0.44$\pm$0.074      & 0.72$\pm$0.010      & 0.33$\pm$0.010      & 0.51$\pm$0.015      & 0.60$\pm$0.004      & 0.59 \\\cline{2-10}
			& ALL          & \f{0.85}$\pm$0.008   & \f{0.76}$\pm$0.008  & \f{0.53}$\pm$0.043  & \f{0.83}$\pm$0.007  & 0.36$\pm$0.051      & \f{0.69}$\pm$0.022  & \f{0.81}$\pm$0.003  & \f{0.69} \\ 
			\hline
			\multirow{6}{*}{PageRank}
			& Local-NR     & 0.44$\pm$0.013       & 0.15$\pm$0.018      & 0.42$\pm$0.069      & 0.64$\pm$0.008      & 0.25$\pm$0.038      & 0.46$\pm$0.016      & 0.58$\pm$0.009      & 0.42 \\
			& Local-NPP    & 0.41$\pm$0.012       & 0.18$\pm$0.017      & 0.43$\pm$0.086      & 0.39$\pm$0.009      & 0.25$\pm$0.040      & 0.41$\pm$0.013      & 0.53$\pm$0.008      & 0.37 \\\cline{2-10}
			& Global-NR    & \f{0.64}$\pm$0.014   & 0.41$\pm$0.015      & \s{0.49}$\pm$0.078  & \f{0.74}$\pm$0.006  & 0.35$\pm$0.056      & \f{0.53}$\pm$0.019  & \f{0.63}$\pm$0.009  & \s{0.54} \\
			& Global-NPP   & \f{0.64}$\pm$0.012   & \s{0.43}$\pm$0.015  & \f{0.55}$\pm$0.046  & \s{0.65}$\pm$0.011  & \f{0.38}$\pm$0.047  & \s{0.52}$\pm$0.017  & \s{0.61}$\pm$0.007  & \s{0.54} \\
			& Global-basic & \s{0.63}$\pm$0.010   & 0.31$\pm$0.023      & 0.41$\pm$0.054      & \s{0.65}$\pm$0.007  & 0.28$\pm$0.035      & 0.49$\pm$0.010      & 0.54$\pm$0.008      & 0.47 \\\cline{2-10}
			& ALL          & \f{0.64}$\pm$0.009   & \f{0.44}$\pm$0.013  & \f{0.55}$\pm$0.035  & \f{0.74}$\pm$0.006  & \s{0.37}$\pm$0.030  & \f{0.53}$\pm$0.020  & \f{0.63}$\pm$0.008  & \f{0.56} \\ 
			\hline
   
		\hline
		\end{tabular}
	}
\end{table*}

\begin{table*}[t]
	\centering
	\caption{\label{tab:node_accuracy}
	Accuracy on the task of predicting future node importance when $d_\theta=2$.
	}
	\resizebox{\textwidth}{!}{
		\begin{tabular}{|c|c|c|c|c|c|c|c|c|c|}
		\hline
			\multirow{2}{*}{Centrality}  & \multirow{2}{*}{Feature} & \multicolumn{2}{c|}{Citation Networks} & \multicolumn{3}{c|}{Email/Message Networks} & \multicolumn{2}{c|}{Online Q/A Networks} & \multirow{2}{*}{\bf{Average}} \\\cline{3-9}
			&  & HepPh & Hepth & Email-EU & Email-Enron & Message-College & Mathoverflow & Askubuntu & \\ 
			\hline
			\multirow{6}{*}{Degree}
			& Local-NR     & 0.71$\pm$0.008       & 0.72$\pm$0.006      & \s{0.80}$\pm$0.024  & 0.62$\pm$0.006      & \f{0.77}$\pm$0.019  & 0.66$\pm$0.008      & 0.58$\pm$0.003      & 0.69 \\
			& Local-NPP    & 0.71$\pm$0.010       & 0.72$\pm$0.007      & 0.79$\pm$0.031      & 0.55$\pm$0.006      & 0.75$\pm$0.020      & 0.65$\pm$0.016      & 0.58$\pm$0.002      & 0.68 \\ \cline{2-10}
			& Global-NR    & \f{0.76}$\pm$0.007   & \s{0.77}$\pm$0.007  & \f{0.82}$\pm$0.024  & \f{0.77}$\pm$0.004  & \s{0.76}$\pm$0.017  & \f{0.68}$\pm$0.014  & \f{0.61}$\pm$0.004  & \f{0.74} \\
			& Global-NPP   & \s{0.75}$\pm$0.003   & \s{0.77}$\pm$0.007  & \s{0.80}$\pm$0.023  & \s{0.75}$\pm$0.004  & \f{0.77}$\pm$0.019  & \f{0.68}$\pm$0.012  & \f{0.61}$\pm$0.004  & \s{0.73} \\
			& Global-basic & 0.73$\pm$0.006       & 0.74$\pm$0.008      & 0.77$\pm$0.028      & \s{0.75}$\pm$0.006  & 0.74$\pm$0.019      & \s{0.67}$\pm$0.011  & \s{0.60}$\pm$0.003  & 0.72 \\\cline{2-10}
			& ALL          & \f{0.76}$\pm$0.008   & \f{0.78}$\pm$0.006  & \f{0.82}$\pm$0.019  & \f{0.77}$\pm$0.005  & \s{0.76}$\pm$0.016  & \f{0.68}$\pm$0.014  & \f{0.61}$\pm$0.004  & \f{0.74} \\ 
			\hline
			\multirow{6}{*}{Betweenness}
			& Local-NR     & 0.79$\pm$0.005       & \s{0.93}$\pm$0.005  & 0.78$\pm$0.017      & 0.82$\pm$0.006      & \f{0.76}$\pm$0.019  & \f{0.86}$\pm$0.005  & \f{0.90}$\pm$0.003  & 0.83 \\
			& Local-NPP    & 0.79$\pm$0.004       & 0.92$\pm$0.004      & 0.75$\pm$0.028      & 0.81$\pm$0.005      & 0.65$\pm$0.032      & \f{0.86}$\pm$0.005  & \f{0.90}$\pm$0.003  & 0.81 \\ \cline{2-10}
			& Global-NR    & \f{0.81}$\pm$0.004   & \f{0.94}$\pm$0.003  & \s{0.81}$\pm$0.017  & \f{0.84}$\pm$0.007  & \s{0.75}$\pm$0.016  & \f{0.86}$\pm$0.004  & \f{0.90}$\pm$0.002  & \s{0.84} \\
			& Global-NPP   & \s{0.80}$\pm$0.005   & \f{0.94}$\pm$0.003  & 0.79$\pm$0.021      & \s{0.83}$\pm$0.007  & \f{0.76}$\pm$0.020  & \f{0.86}$\pm$0.004  & \f{0.90}$\pm$0.003  & \s{0.84} \\
			& Global-basic & 0.71$\pm$0.005       & 0.73$\pm$0.010      & 0.75$\pm$0.033      & 0.77$\pm$0.007      & 0.70$\pm$0.019      & \s{0.67}$\pm$0.008  & \s{0.77}$\pm$0.004  & 0.73 \\ \cline{2-10}
			& ALL          & \f{0.81}$\pm$0.003   & \f{0.94}$\pm$0.004  & \f{0.82}$\pm$0.018  & \f{0.84}$\pm$0.008  & \s{0.75}$\pm$0.019  & \f{0.86}$\pm$0.004  & \f{0.90}$\pm$0.003  & \f{0.85}\\ 
			\hline
			\multirow{6}{*}{Closeness}
			& Local-NR     & 0.76$\pm$0.003       & 0.76$\pm$0.005      & 0.78$\pm$0.027      & 0.75$\pm$0.005      & \f{0.76}$\pm$0.015  & 0.72$\pm$0.006      & \s{0.78}$\pm$0.003  & 0.76 \\
			& Local-NPP    & 0.73$\pm$0.003       & 0.75$\pm$0.006      & 0.75$\pm$0.032      & 0.63$\pm$0.007      & 0.74$\pm$0.016      & 0.65$\pm$0.010      & 0.64$\pm$0.004      & 0.70 \\\cline{2-10}
			& Global-NR    & \f{0.91}$\pm$0.003   & \f{0.86}$\pm$0.004  & 0.80$\pm$0.018      & \f{0.85}$\pm$0.006  & \s{0.75}$\pm$0.019  & \f{0.77}$\pm$0.011  & \f{0.82}$\pm$0.002  & \s{0.82} \\
			& Global-NPP   & \s{0.90}$\pm$0.005   & \s{0.85}$\pm$0.002  & \s{0.81}$\pm$0.013  & \s{0.80}$\pm$0.007  & \f{0.76}$\pm$0.011  & \s{0.73}$\pm$0.010  & 0.73$\pm$0.003      & 0.80 \\
			& Global-basic & \s{0.90}$\pm$0.003   & 0.82$\pm$0.004      & 0.76$\pm$0.033      & 0.77$\pm$0.008      & 0.73$\pm$0.024      & 0.66$\pm$0.009      & 0.64$\pm$0.003      & 0.75 \\\cline{2-10}
			& ALL          & \f{0.91}$\pm$0.004   & \f{0.86}$\pm$0.004  & \f{0.82}$\pm$0.020  & \f{0.85}$\pm$0.005  & 0.75$\pm$0.017      & \f{0.77}$\pm$0.011  & \f{0.82}$\pm$0.002  & \f{0.83} \\ 
			\hline
			\multirow{6}{*}{PageRank}
			& Local-NR     & \s{0.76}$\pm$0.007   & 0.72$\pm$0.009      & 0.80$\pm$0.020      & \s{0.74}$\pm$0.006  & \s{0.75}$\pm$0.017  & 0.66$\pm$0.011      & 0.65$\pm$0.006      & 0.73 \\
			& Local-NPP    & 0.74$\pm$0.007       & 0.72$\pm$0.007      & 0.77$\pm$0.023      & 0.65$\pm$0.018      & 0.73$\pm$0.023      & 0.63$\pm$0.005      & 0.62$\pm$0.006      & 0.69 \\\cline{2-10}
			& Global-NR    & \f{0.81}$\pm$0.005   & \f{0.75}$\pm$0.007  & \s{0.81}$\pm$0.022  & \f{0.80}$\pm$0.005  & \s{0.75}$\pm$0.015  & \f{0.68}$\pm$0.010  & \f{0.67}$\pm$0.006  & \f{0.75} \\
			& Global-NPP   & \f{0.81}$\pm$0.005   & \s{0.74}$\pm$0.007  & \f{0.82}$\pm$0.023  & \s{0.74}$\pm$0.011  & \f{0.76}$\pm$0.019  & \s{0.67}$\pm$0.008  & \s{0.66}$\pm$0.006  & \s{0.74} \\
			& Global-basic & \f{0.81}$\pm$0.004   & 0.73$\pm$0.008      & 0.75$\pm$0.021      & \s{0.74}$\pm$0.006  & 0.70$\pm$0.025      & 0.65$\pm$0.007      & 0.59$\pm$0.006      & 0.71 \\\cline{2-10}
			& ALL          & \f{0.81}$\pm$0.005   & \f{0.75}$\pm$0.004  & \s{0.81}$\pm$0.017  & \f{0.80}$\pm$0.006  & \s{0.75}$\pm$0.014  & \f{0.68}$\pm$0.010  & \f{0.67}$\pm$0.005  & \f{0.75} \\ 
			\hline
   
		\hline
		\end{tabular}
	}
\end{table*}

\section{Detailed Results of Future Node Importance Prediction}
\label{sec:appendix:node:details}

In Tables~\ref{tab:node_f1}-\ref{tab:node_auroc_degree},
we provide the average predictive performances and standard deviations over $10$ runs on the task of predicting future node centrality in each real-world graph in terms of several evaluation metrics. 
The detailed experimental settings can be found in Section~\ref{subsection:node_prediction} of the main paper. 

\section{Detailed Results of Future Edge Importance Prediction}
\label{sec:appendix:edge:details}
In Tables~\ref{tab:edge_f1_score}-\ref{tab:edge_auroc},
we provide the average predictive performances and standard deviations over $10$ runs on the task of predicting future node centrality in each real-world graph in terms of several evaluation metrics. 
The detailed experimental settings can be found in Section~\ref{section:edge} of the main paper.

\begin{table*}[t]
\vspace{-2mm}
	\centering
	\caption{\label{tab:node_auroc}
	AUROC on the task of predicting future node importance when $d_\theta=2$.
	}
	\resizebox{\textwidth}{!}{
		\begin{tabular}{|c|c|c|c|c|c|c|c|c|c|}
		\hline
			\multirow{2}{*}{Centrality}  & \multirow{2}{*}{Feature} & \multicolumn{2}{c|}{Citation Networks} & \multicolumn{3}{c|}{Email/Message Networks} & \multicolumn{2}{c|}{Online Q/A Networks} & \multirow{2}{*}{\bf{Average}} \\\cline{3-9}
			&  & HepPh & Hepth & Email-EU & Email-Enron & Message-College & Mathoverflow & Askubuntu & \\ 
			\hline
			\multirow{6}{*}{Degree}
			& Local-NR     & 0.69$\pm$0.006       & 0.70$\pm$0.006      & 0.80$\pm$0.037      & 0.67$\pm$0.007      & 0.69$\pm$0.020      & 0.65$\pm$0.012      & 0.58$\pm$0.006      & 0.68 \\
			& Local-NPP    & 0.64$\pm$0.005       & 0.67$\pm$0.007      & 0.75$\pm$0.032      & 0.56$\pm$0.007      & 0.64$\pm$0.025      & 0.63$\pm$0.012      & 0.56$\pm$0.004      & 0.64 \\ \cline{2-10}
			& Global-NR    & \f{0.81}$\pm$0.004   & \s{0.83}$\pm$0.005  & \f{0.85}$\pm$0.037  & \f{0.86}$\pm$0.005  & \s{0.73}$\pm$0.026  & \f{0.71}$\pm$0.013  & \f{0.65}$\pm$0.006  & \f{0.78} \\
			& Global-NPP   & \s{0.80}$\pm$0.002   & \s{0.83}$\pm$0.005  & \s{0.83}$\pm$0.032  & 0.83$\pm$0.005      & \f{0.74}$\pm$0.026  & \f{0.71}$\pm$0.012  & \f{0.65}$\pm$0.005  & \s{0.77} \\
			& Global-basic & 0.74$\pm$0.004       & 0.74$\pm$0.007      & 0.82$\pm$0.035      & \s{0.84}$\pm$0.006  & 0.68$\pm$0.037      & \s{0.68}$\pm$0.009  & \s{0.62}$\pm$0.005  & 0.73 \\\cline{2-10}
			& ALL          & \f{0.81}$\pm$0.004   & \f{0.84}$\pm$0.006  & \f{0.85}$\pm$0.036  & \f{0.86}$\pm$0.005  & \s{0.73}$\pm$0.025  & \f{0.71}$\pm$0.014  & \f{0.65}$\pm$0.006  & \f{0.78} \\ 
			\hline
			\multirow{6}{*}{Betweenness}
			& Local-NR     & 0.84$\pm$0.006       & \s{0.98}$\pm$0.002  & 0.79$\pm$0.033      & 0.80$\pm$0.007      & 0.70$\pm$0.029      & \s{0.83}$\pm$0.010  & \s{0.82}$\pm$0.005  & 0.82 \\
			& Local-NPP    & 0.83$\pm$0.006       & \s{0.98}$\pm$0.003  & 0.73$\pm$0.031      & 0.68$\pm$0.008      & 0.65$\pm$0.041      & 0.82$\pm$0.010      & 0.81$\pm$0.006      & 0.79 \\ \cline{2-10}
			& Global-NR    & \f{0.87}$\pm$0.005   & \f{0.99}$\pm$0.002  & \s{0.81}$\pm$0.030  & \f{0.87}$\pm$0.007  & \s{0.71}$\pm$0.032  & \f{0.86}$\pm$0.008  & \f{0.86}$\pm$0.005  & \s{0.85} \\
			& Global-NPP   & \s{0.85}$\pm$0.007   & \s{0.98}$\pm$0.002  & \s{0.81}$\pm$0.032  & \s{0.84}$\pm$0.008  & \f{0.72}$\pm$0.033  & \f{0.86}$\pm$0.008  & \f{0.86}$\pm$0.005  & \s{0.85} \\
			& Global-basic & 0.62$\pm$0.008       & 0.74$\pm$0.013      & 0.76$\pm$0.042      & 0.74$\pm$0.005      & 0.63$\pm$0.038      & 0.63$\pm$0.014      & 0.60$\pm$0.007      & 0.67 \\ \cline{2-10}
			& ALL          & \f{0.87}$\pm$0.006   & \f{0.99}$\pm$0.001  & \f{0.83}$\pm$0.024  & \f{0.87}$\pm$0.007  & \s{0.71}$\pm$0.033  & \f{0.86}$\pm$0.009  & \f{0.86}$\pm$0.005  & \f{0.86}\\ 
			\hline
			\multirow{6}{*}{Closeness}
			& Local-NR     & 0.79$\pm$0.005       & 0.79$\pm$0.006      & 0.76$\pm$0.047      & 0.84$\pm$0.006      & 0.68$\pm$0.026      & 0.76$\pm$0.009      & 0.84$\pm$0.002      & 0.78 \\
			& Local-NPP    & 0.74$\pm$0.005       & 0.78$\pm$0.008      & 0.72$\pm$0.042      & 0.65$\pm$0.009      & 0.61$\pm$0.022      & 0.63$\pm$0.007      & 0.70$\pm$0.005      & 0.69 \\\cline{2-10}
			& Global-NR    & \f{0.97}$\pm$0.002   & \s{0.93}$\pm$0.005  & \s{0.82}$\pm$0.035  & \f{0.93}$\pm$0.003  & \s{0.73}$\pm$0.024  & \f{0.83}$\pm$0.011  & \s{0.89}$\pm$0.002  & \s{0.87} \\
			& Global-NPP   & \s{0.96}$\pm$0.003   & 0.92$\pm$0.004      & \f{0.83}$\pm$0.028  & \s{0.88}$\pm$0.004  & \s{0.73}$\pm$0.026  & \s{0.79}$\pm$0.011  & 0.81$\pm$0.003      & 0.85 \\
			& Global-basic & 0.95$\pm$0.003       & 0.89$\pm$0.005      & 0.79$\pm$0.044      & 0.85$\pm$0.007      & 0.69$\pm$0.035      & 0.68$\pm$0.012      & 0.70$\pm$0.003      & 0.79 \\\cline{2-10}
			& ALL          & \f{0.97}$\pm$0.002   & \f{0.94}$\pm$0.005  & \s{0.82}$\pm$0.027  & \f{0.93}$\pm$0.004  & \f{0.75}$\pm$0.028  & \f{0.83}$\pm$0.010  & \f{0.90}$\pm$0.001  & \f{0.88} \\ 
			\hline
			\multirow{6}{*}{PageRank}
			& Local-NR     & 0.77$\pm$0.006       & 0.65$\pm$0.010      & 0.80$\pm$0.031      & 0.81$\pm$0.004      & \s{0.67}$\pm$0.011  & 0.69$\pm$0.015      & 0.70$\pm$0.006      & 0.73 \\
			& Local-NPP    & 0.75$\pm$0.004       & 0.63$\pm$0.008      & 0.77$\pm$0.031      & 0.63$\pm$0.010      & 0.62$\pm$0.024      & 0.64$\pm$0.009      & 0.67$\pm$0.006      & 0.67 \\\cline{2-10}
			& Global-NR    & \f{0.87}$\pm$0.005   & \s{0.78}$\pm$0.006  & \s{0.84}$\pm$0.035  & \f{0.88}$\pm$0.003  & \f{0.72}$\pm$0.029  & \f{0.71}$\pm$0.012  & \f{0.73}$\pm$0.005  & \f{0.79} \\
			& Global-NPP   & \s{0.86}$\pm$0.005   & 0.77$\pm$0.007      & \f{0.85}$\pm$0.028  & \s{0.82}$\pm$0.008  & \f{0.72}$\pm$0.018  & \s{0.70}$\pm$0.011  & \s{0.71}$\pm$0.006  & \s{0.78} \\
			& Global-basic & \s{0.86}$\pm$0.005   & 0.73$\pm$0.010      & 0.81$\pm$0.033      & 0.81$\pm$0.008      & 0.66$\pm$0.026      & 0.66$\pm$0.011      & 0.62$\pm$0.004      & 0.74 \\\cline{2-10}
			& ALL          & \f{0.87}$\pm$0.005   & \f{0.79}$\pm$0.007  & \f{0.85}$\pm$0.035  & \f{0.88}$\pm$0.004  & \f{0.72}$\pm$0.031  & \f{0.71}$\pm$0.014  & \f{0.73}$\pm$0.005  & \f{0.79} \\ 
			\hline
   
		\hline
		\end{tabular}
	}
\end{table*}

\begin{table*}
    \vspace{-1mm}
	\centering
	\caption{\label{tab:node_f1_degree}
	F1-score on the task of predicting future node importance depending on $d_\theta$ (i.e., in-degree of nodes when their input features are extracted).}
        \resizebox{\textwidth}{!}{
		\begin{tabular}{|c|c|c|c|c|c|c|c|c|c|}
		\hline
			\multirow{2}{*}{Centrality}  & \multirow{2}{*}{Feature} & \multicolumn{2}{c|}{Citation Networks} & \multicolumn{3}{c|}{Email/Message Networks} & \multicolumn{2}{c|}{Online Q/A Networks} & \multirow{2}{*}{\bf{Average}} \\\cline{3-9}
			& & HepPh & Hepth & Email-EU & Email-Enron & Message-College & Mathoverflow & Askubuntu &  \\ 
			\hline
			\multirow{3}{*}{Degree}
			& ALL $(d_\theta=2)$ & 0.53$\pm$0.010     & 0.58$\pm$0.013       & \s{0.52}$\pm$0.043    & \f{0.19$\pm$0.005}& 0.38$\pm$0.041           & \s{0.52}$\pm$0.026      &  \f{0.70$\pm$0.005}  & 0.59 \\
			& ALL $(d_\theta=4)$ & \s{0.67}$\pm$0.007 & \s{0.78}$\pm$0.008   & \s{0.62}$\pm$0.062    & 1.00*             & \s{0.49}$\pm$0.045       & \f{0.89}$\pm$0.006      &  1.00*               & \s{0.69}\\
			& ALL $(d_\theta=8)$ & \f{0.82}$\pm$0.006 & \f{0.93}$\pm$0.006   & \f{0.74}$\pm$0.036    & 1.00*             & \f{0.71}$\pm$0.022       & 1.00*                   &  1.00*               & \f{0.80}\\
            \hline
            
            \multirow{3}{*}{Betweenness}
            & ALL $(d_\theta=2)$ & 0.64$\pm$0.007     & 0.90$\pm$0.007       & \s{0.53}$\pm$0.052    & 0.62$\pm$0.016     & 0.38$\pm$0.045       & 0.75$\pm$0.010      & 0.74$\pm$0.007          & 0.65 \\
			& ALL $(d_\theta=4)$ & \s{0.72}$\pm$0.012 & \s{0.94}$\pm$0.007   & 0.52$\pm$0.066        & \s{0.75}$\pm$0.006 & \s{0.50}$\pm$0.042   & \s{0.84}$\pm$0.009  & \s{0.84}$\pm$0.008      & \s{0.73} \\
			& ALL $(d_\theta=8)$ & \f{0.77}$\pm$0.008 & \f{0.97}$\pm$0.005   & \f{0.65}$\pm$0.045    & \f{0.86}$\pm$0.009 & \f{0.69}$\pm$0.057   & \f{0.91}$\pm$0.009  & \f{0.89}$\pm$0.007      & \f{0.82} \\
            \hline
			
            \multirow{3}{*}{Closeness}
            & ALL $(d_\theta=2)$ & 0.85$\pm$0.008       & 0.76$\pm$0.008        & 0.53$\pm$0.043        & 0.83$\pm$0.007     & 0.36$\pm$0.051      & 0.69$\pm$0.022      & 0.81$\pm$0.003      & 0.69 \\
			& ALL $(d_\theta=4)$ & \s{0.87}$\pm$0.008   & \s{0.85}$\pm$0.010    & \s{0.55}$\pm$0.072    &\s{0.91}$\pm$0.006  & \s{0.54}$\pm$0.032  & \s{0.85}$\pm$0.013  & \s{0.91}$\pm$0.004  & \s{0.78} \\
			& ALL $(d_\theta=8)$ & \f{0.88}$\pm$0.007   & \f{0.90}$\pm$0.009    & \f{0.65}$\pm$0.061    & \f{0.97}$\pm$0.004 & \f{0.74}$\pm$0.045  & \f{0.95}$\pm$0.007  & \f{0.98}$\pm$0.003  & \f{0.86} \\
            \hline
		
            \multirow{3}{*}{PageRank}
            & ALL $(d_\theta=2)$ & 0.64$\pm$0.009      & 0.44$\pm$0.013      & 0.52$\pm$0.035        & 0.74$\pm$0.006     & 0.37$\pm$0.030       & 0.53$\pm$0.020      & 0.63$\pm$0.006  & 0.55       \\
			& ALL $(d_\theta=4)$ & \s{0.74}$\pm$0.008  & \s{0.71}$\pm$0.010  & \s{0.62}$\pm$0.037    & \s{0.87}$\pm$0.006 & \s{0.48}$\pm$0.040   & \s{0.79}$\pm$0.008  & \s{0.89}$\pm$0.003      & \s{0.73}   \\
			& ALL $(d_\theta=8)$ & \f{0.83}$\pm$0.007  & \f{0.85}$\pm$0.012  & \f{0.68}$\pm$0.049    & \f{0.95}$\pm$0.003 & \f{0.72}$\pm$0.033   & \f{0.95}$\pm$0.006  & \f{0.98}$\pm$0.003  & \f{0.85}   \\
            \hline
            \multicolumn{10}{l}{* All nodes satisfying the condition on $d_\theta$ have the same class, belonging to top $20\%$ in terms of the considered centrality measure.} \\
  
		\end{tabular}
	}
\end{table*}

\begin{table*}[t]
    \vspace{-1mm}
	\centering
	\caption{\label{tab:node_accuracy_degree}
	Accuracy on the task of predicting future node importance depending on $d_\theta$.}
	\resizebox{\textwidth}{!}{
		\begin{tabular}{|c|c|c|c|c|c|c|c|c|c|}
		\hline
			\multirow{2}{*}{Centrality}  & \multirow{2}{*}{Feature} & \multicolumn{2}{c|}{Citation Networks} & \multicolumn{3}{c|}{Email/Message Networks} & \multicolumn{2}{c|}{Online Q/A Networks} & \multirow{2}{*}{\bf{Average}} \\\cline{3-9}
			& & HepPh & Hepth & Email-EU & Email-Enron & Message-College & Mathoverflow & Askubuntu &  \\ 
			\hline
			\multirow{3}{*}{Degree}
			& ALL $(d_\theta=2)$ & \s{0.76}$\pm$0.008 & 0.78$\pm$0.006       & 0.82$\pm$0.019        & \f{0.77$\pm$0.005}& \s{0.76}$\pm$0.016       & \s{0.68}$\pm$0.014  &  \f{0.61$\pm$0.004}  & 0.74 \\
			& ALL $(d_\theta=4)$ & \s{0.76}$\pm$0.006 & \s{0.80}$\pm$0.009   & \s{0.83}$\pm$0.028    & 1.00*             & 0.70$\pm$0.046           & \f{0.81}$\pm$0.009  &  1.00*               & \s{0.78}\\
			& ALL $(d_\theta=8)$ & \f{0.79}$\pm$0.006 & \f{0.88}$\pm$0.010   & \f{0.86}$\pm$0.021    & 1.00*             & \f{0.72}$\pm$0.025       & 1.00*               &  1.00*               & \f{0.81}\\
            \hline
            
            \multirow{3}{*}{Betweenness}
            & ALL $(d_\theta=2)$ & 0.81$\pm$0.003 & 0.94$\pm$0.004       & \f{0.82}$\pm$0.018       & \s{0.84}$\pm$0.008   & \f{0.75}$\pm$0.019   & \s{0.86}$\pm$0.004  & \f{0.90}$\pm$0.003  & \f{0.85} \\
			& ALL $(d_\theta=4)$ & 0.81$\pm$0.008 & \s{0.96}$\pm$0.004   & \s{0.80}$\pm$0.023       & 0.82$\pm$0.004       & \s{0.70}$\pm$0.023   & \s{0.86}$\pm$0.009  & \s{0.89}$\pm$0.006  & \s{0.83} \\
			& ALL $(d_\theta=8)$ & 0.81$\pm$0.004 & \f{0.98}$\pm$0.004   & \f{0.82}$\pm$0.019       & \f{0.85}$\pm$0.009   & \s{0.70}$\pm$0.049   & \f{0.88}$\pm$0.010  & 0.88$\pm$0.006      & \f{0.85} \\
            \hline
			
            \multirow{3}{*}{Closeness}
            & ALL $(d_\theta=2)$ & 0.91$\pm$0.004   & 0.86$\pm$0.004        & \f{0.82}$\pm$0.020        & 0.85$\pm$0.005 & \f{0.75}$\pm$0.017      & 0.77$\pm$0.011      & 0.82$\pm$0.002      & \s{0.83} \\
			& ALL $(d_\theta=4)$ & 0.91$\pm$0.005   & \s{0.88}$\pm$0.007    & \s{0.80}$\pm$0.022    & \s{0.89}$\pm$0.007 & 0.70$\pm$0.021  & \s{0.80}$\pm$0.015  & \s{0.86}$\pm$0.006  & \s{0.83} \\
			& ALL $(d_\theta=8)$ & 0.91$\pm$0.006   & \f{0.89}$\pm$0.009    & \f{0.82}$\pm$0.020    & \f{0.94}$\pm$0.006 & \s{0.73}$\pm$0.046  & \f{0.91}$\pm$0.013  & \f{0.95}$\pm$0.006  & \f{0.88} \\
            \hline
		
            \multirow{3}{*}{PageRank}
            & ALL $(d_\theta=2)$ & \s{0.81}$\pm$0.005  & \s{0.75}$\pm$0.004  & 0.81$\pm$0.017        & 0.80$\pm$0.006     & \f{0.75}$\pm$0.014 & \s{0.68}$\pm$0.010  & 0.67$\pm$0.006     & 0.75       \\
			& ALL $(d_\theta=4)$ & \s{0.81}$\pm$0.006  & \s{0.75}$\pm$0.007  & \f{0.83}$\pm$0.017    & \s{0.83}$\pm$0.007 & 0.68$\pm$0.025   & \s{0.68}$\pm$0.011  & \s{0.81}$\pm$0.003 & \s{0.77}   \\
			& ALL $(d_\theta=8)$ & \f{0.83}$\pm$0.005  & \f{0.81}$\pm$0.012  & \s{0.82}$\pm$0.023    & \f{0.92}$\pm$0.004 & \s{0.72}$\pm$0.025   & \f{0.91}$\pm$0.011  & \f{0.96}$\pm$0.003 & \f{0.85}   \\
            \hline
            \multicolumn{10}{l}{* All nodes satisfying the condition on $d_\theta$ have the same class, belonging to top $20\%$ in terms of the considered centrality measure.} \\
  
		\end{tabular}
	}
\end{table*}

\begin{table*}[t]
	\centering
	\caption{\label{tab:node_auroc_degree}
	AUROC on the task of predicting future node importance depending on $d_\theta$.}
	\resizebox{\textwidth}{!}{
		\begin{tabular}{|c|c|c|c|c|c|c|c|c|c|}
		\hline
			\multirow{2}{*}{Centrality}  & \multirow{2}{*}{Feature} & \multicolumn{2}{c|}{Citation Networks} & \multicolumn{3}{c|}{Email/Message Networks} & \multicolumn{2}{c|}{Online Q/A Networks} & \multirow{2}{*}{\bf{Average}} \\\cline{3-9}
			& & HepPh & Hepth & Email-EU & Email-Enron & Message-College & Mathoverflow & Askubuntu &  \\ 
			\hline
			\multirow{3}{*}{Degree}
			& ALL $(d_\theta=2)$ & 0.81$\pm$0.004     & 0.84$\pm$0.005       & \s{0.85}$\pm$0.035    & \f{0.86$\pm$0.005}& \s{0.73}$\pm$0.025       & \f{0.71}$\pm$0.014      &  \f{0.65$\pm$0.005}  & 0.78 \\
			& ALL $(d_\theta=4)$ & \s{0.83}$\pm$0.005 & \s{0.87}$\pm$0.006   & \s{0.85}$\pm$0.036    & 1.00*             & 0.72$\pm$0.027           & \s{0.68}$\pm$0.018      &  1.00*               & \s{0.79}\\
			& ALL $(d_\theta=8)$ & \f{0.87}$\pm$0.007 & \f{0.90}$\pm$0.013   & \f{0.88}$\pm$0.027    & 1.00*             & \f{0.78}$\pm$0.031       & 1.00*                   &  1.00*               & \f{0.86}\\
            \hline
            
            \multirow{3}{*}{Betweenness}
            & ALL $(d_\theta=2)$ & 0.87$\pm$0.005     & \s{0.99}$\pm$0.001   & \s{0.83}$\pm$0.024    & 0.87$\pm$0.007     & 0.71$\pm$0.033       & 0.86$\pm$0.009      & 0.86$\pm$0.005          & 0.86 \\
			& ALL $(d_\theta=4)$ & \s{0.89}$\pm$0.006 & \s{0.99}$\pm$0.001   & 0.81$\pm$0.040        & \s{0.89}$\pm$0.004 & \s{0.73}$\pm$0.026   & \s{0.90}$\pm$0.007  & \s{0.91}$\pm$0.004      & \s{0.87} \\
			& ALL $(d_\theta=8)$ & \f{0.90}$\pm$0.003 & \f{1.00}$\pm$0.001   & \f{0.84}$\pm$0.026    & \f{0.93}$\pm$0.006 & \f{0.77}$\pm$0.044   & \f{0.94}$\pm$0.009  & \f{0.94}$\pm$0.006      & \f{0.90} \\
            \hline
			
            \multirow{3}{*}{Closeness}
            & ALL $(d_\theta=2)$ & 0.97$\pm$0.002    & 0.94$\pm$0.005        & \s{0.84}$\pm$0.033    & 0.93$\pm$0.004     & 0.73$\pm$0.028      & 0.83$\pm$0.010      & 0.90$\pm$0.002      & 0.88 \\
			& ALL $(d_\theta=4)$ & 0.97$\pm$0.002    & \s{0.95}$\pm$0.004    & 0.82$\pm$0.027        & \s{0.95}$\pm$0.004 & \s{0.75}$\pm$0.030  & \s{0.88}$\pm$0.012  & \s{0.93}$\pm$0.004  & \s{0.89} \\
			& ALL $(d_\theta=8)$ & 0.97$\pm$0.003    & \f{0.96}$\pm$0.006    & \f{0.88}$\pm$0.024    & \f{0.98}$\pm$0.004 & \f{0.79}$\pm$0.043  & \f{0.92}$\pm$0.016  & \f{0.95}$\pm$0.013  & \f{0.92} \\
            \hline
		
            \multirow{3}{*}{PageRank}
            & ALL $(d_\theta=2)$ & 0.87$\pm$0.005      & 0.79$\pm$0.008      & \s{0.85}$\pm$0.035    & 0.88$\pm$0.004     & \s{0.72}$\pm$0.031   & \s{0.71}$\pm$0.014   & \f{0.73}$\pm$0.005  & 0.79       \\
			& ALL $(d_\theta=4)$ & \s{0.89}$\pm$0.006  & \s{0.83}$\pm$0.008  & \f{0.87}$\pm$0.018    & \s{0.90}$\pm$0.006 & 0.71$\pm$0.028       & 0.69$\pm$0.009       & 0.70$\pm$0.007      & \s{0.80}   \\
			& ALL $(d_\theta=8)$ & \f{0.91}$\pm$0.005  & \f{0.87}$\pm$0.013  & \f{0.87}$\pm$0.034    & \f{0.95}$\pm$0.009 & \f{0.79}$\pm$0.018   & \f{0.73}$\pm$0.040   & \s{0.71}$\pm$0.049  & \f{0.83}   \\
            \hline
            \multicolumn{10}{l}{* All nodes satisfying the condition on $d_\theta$ have the same class, belonging to top $20\%$ in terms of the considered centrality measure.} \\
  
		\end{tabular}
	}
\end{table*}

\begin{table*}[t]
	\centering
	\caption{\label{tab:edge_f1_score}	F1-score on the task of predicting future edge importance.
	}
	\resizebox{\textwidth}{!}{
		\begin{tabular}{|c|c|c|c|c|c|c|c|c|c|}
		\hline
			\multirow{2}{*}{Centrality}  & \multirow{2}{*}{Feature} & \multicolumn{2}{c|}{Citation Networks} & \multicolumn{3}{c|}{Email/Message Networks} & \multicolumn{2}{c|}{Online Q/A Networks} & \multirow{2}{*}{\bf{Average}}\\\cline{3-9}
			& & HepPh & Hepth &  Email-EU & Email-Enron & Message-College & Mathoverflow & Askubuntu & \\ 
	    	\hline
			\multirow{8}{*}{Edge Betweenness}
			& Local-ER $(d_\theta=2)$    & 0.68 $\pm$ 0.004       & 0.59 $\pm$ 0.011      & 0.14 $\pm$ 0.094      & 0.74 $\pm$ 0.013      & \f{0.41} $\pm$ 0.042  & \s{0.21} $\pm$ 0.038  & \f{0.40} $\pm$ 0.013  & 0.45 \\
			& Global-ER $(d_\theta=2)$   & \s{0.69} $\pm$ 0.015   & \s{0.63} $\pm$ 0.041  & 0.18 $\pm$ 0.122      & \s{0.79} $\pm$ 0.051  & 0.38 $\pm$ 0.060      & \f{0.23} $\pm$ 0.058  & \s{0.39} $\pm$ 0.022  & \f{0.47} \\
			& Global-Basic $(d_\theta=2)$& \s{0.69} $\pm$ 0.013   & 0.51 $\pm$ 0.168      & \s{0.22} $\pm$ 0.132  & 0.75 $\pm$ 0.072      & 0.37 $\pm$ 0.064      & 0.15 $\pm$ 0.116      & 0.26 $\pm$ 0.180  & 0.42 \\
			& ALL $(d_\theta=2)$         & \f{0.71} $\pm$ 0.005   & \f{0.68} $\pm$ 0.009  & \f{0.25} $\pm$ 0.186  & \f{0.84} $\pm$ 0.005  & \s{0.40} $\pm$ 0.062  & \f{0.23} $\pm$ 0.060  & 0.36 $\pm$ 0.018  & \s{0.50} \\
			\cline{2-10}
			& ALL $(d_\theta=2)$     & \f{0.71} $\pm$ 0.005          & 0.68 $\pm$ 0.009      & \s{0.25} $\pm$ 0.186  & \f{0.84} $\pm$ 0.005     & \s{0.40} $\pm$ 0.062      & 0.23 $\pm$ 0.060      & 0.36 $\pm$ 0.018  & 0.50 \\
			& ALL $(d_\theta=4)$     & \f{0.71} $\pm$ 0.007      & \s{0.72} $\pm$ 0.009  & \f{0.33} $\pm$ 0.104      & \s{0.77} $\pm$ 0.006         & \f{0.43} $\pm$ 0.086  & \s{0.29} $\pm$ 0.071  & \s{0.46} $\pm$ 0.014  & \f{0.53} \\
			& ALL $(d_\theta=8)$     & \s{0.69} $\pm$ 0.004      & \f{0.75} $\pm$ 0.009  & 0.17 $\pm$ 0.079  & 0.72 $\pm$ 0.011     & 0.39 $\pm$ 0.052  & \f{0.31} $\pm$ 0.055  & \f{0.53} $\pm$ 0.023  & \s{0.52} \\
            
		\hline
		\end{tabular}
	}
\end{table*}

\begin{table*}
	\centering
	\caption{\label{tab:edge_accuracy}
	Accuracy on the task of predicting future edge importance.
	}
	
	\resizebox{\textwidth}{!}{
		\begin{tabular}{|c|c|c|c|c|c|c|c|c|c|}
		\hline
			\multirow{2}{*}{Centrality}  & \multirow{2}{*}{Feature} & \multicolumn{2}{c|}{Citation Networks} & \multicolumn{3}{c|}{Email/Message Networks} & \multicolumn{2}{c|}{Online Q/A Networks} & \multirow{2}{*}{\bf{Average}}\\\cline{3-9}
			& & HepPh & Hepth &  Email-EU & Email-Enron & Message-College & Mathoverflow & Askubuntu & \\ 
	    	\hline
			\multirow{8}{*}{Edge Betweenness}
			& Local-ER $(d_\theta=2)$    & 0.66 $\pm$ 0.003       & 0.72 $\pm$ 0.005      & \s{0.87} $\pm$ 0.041  & 0.75 $\pm$ 0.012      & \f{0.70} $\pm$ 0.030  & 0.91 $\pm$ 0.008  & 0.91 $\pm$ 0.003  & 0.78 \\
			& Global-ER $(d_\theta=2)$   & \f{0.68} $\pm$ 0.020   & \f{0.75} $\pm$ 0.023  & \f{0.88} $\pm$ 0.041  & \f{0.81} $\pm$ 0.058  & \f{0.70} $\pm$ 0.035  & 0.91 $\pm$ 0.009  & 0.91 $\pm$ 0.003  & \f{0.81} \\
			& Global-Basic $(d_\theta=2)$& 0.66 $\pm$ 0.036       & 0.72 $\pm$ 0.046      & \s{0.87} $\pm$ 0.040  & \s{0.79} $\pm$ 0.052  & 0.67 $\pm$ 0.053      & 0.91 $\pm$ 0.009  & 0.91 $\pm$ 0.008  & 0.79 \\
			& ALL $(d_\theta=2)$         & \s{0.67} $\pm$ 0.037   & \s{0.73} $\pm$ 0.047  & \s{0.87} $\pm$ 0.042  & \f{0.81} $\pm$ 0.055  & \s{0.68} $\pm$ 0.052  & 0.91 $\pm$ 0.009  & 0.91 $\pm$ 0.007  & \s{0.80} \\
			\cline{2-10}
			& ALL $(d_\theta=2)$     & 0.67 $\pm$ 0.037          & 0.73 $\pm$ 0.005      & \s{0.87} $\pm$ 0.042  & \s{0.81} $\pm$ 0.055     & 0.68 $\pm$ 0.052      & 0.91 $\pm$ 0.009      & \f{0.91} $\pm$ 0.007      & 0.80 \\
			& ALL $(d_\theta=4)$     & \s{0.73} $\pm$ 0.006      & \s{0.79} $\pm$ 0.005  & 0.86 $\pm$ 0.050      & 0.78 $\pm$ 0.024         & \s{0.78} $\pm$ 0.024  & \s{0.93} $\pm$ 0.007  & \s{0.90} $\pm$ 0.004  & \s{0.82} \\
			& ALL $(d_\theta=8)$     & \f{0.75} $\pm$ 0.003      & \f{0.81} $\pm$ 0.005  & \f{0.88} $\pm$ 0.046  & \f{0.82} $\pm$ 0.017     & \f{0.82} $\pm$ 0.017  & \f{0.94} $\pm$ 0.005  & \f{0.91} $\pm$ 0.002  & \f{0.85} \\
            
		\hline
		\end{tabular}
	}
\end{table*}

\begin{table*}
	\centering
	\caption{\label{tab:edge_auroc}	AUROC on the task of predicting future edge importance.
	}
	\resizebox{\textwidth}{!}{
		\begin{tabular}{|c|c|c|c|c|c|c|c|c|c|}
		\hline
			\multirow{2}{*}{Centrality}  & \multirow{2}{*}{Feature} & \multicolumn{2}{c|}{Citation Networks} & \multicolumn{3}{c|}{Email/Message Networks} & \multicolumn{2}{c|}{Online Q/A Networks} & \multirow{2}{*}{\bf{Average}}\\\cline{3-9}
			& & HepPh & Hepth &  Email-EU & Email-Enron & Message-College & Mathoverflow & Askubuntu & \\ 
	    	\hline
			\multirow{8}{*}{Edge Betweenness}
			& Local-ER $(d_\theta=2)$    & 0.71 $\pm$ 0.003       & 0.77 $\pm$ 0.007      & \f{0.64} $\pm$ 0.080  & 0.82 $\pm$ 0.009      & \f{0.68} $\pm$ 0.035  & \f{0.85} $\pm$ 0.013  & \f{0.86} $\pm$ 0.006  & \s{0.76} \\
			& Global-ER $(d_\theta=2)$   & \f{0.74} $\pm$ 0.027   & \f{0.80} $\pm$ 0.033  & \s{0.63} $\pm$ 0.092  & \s{0.88} $\pm$ 0.058  & \f{0.68} $\pm$ 0.032  & \f{0.85} $\pm$ 0.015  & \f{0.86} $\pm$ 0.007  & \f{0.78} \\
			& Global-Basic $(d_\theta=2)$& 0.71 $\pm$ 0.053       & 0.77 $\pm$ 0.055      & \f{0.64} $\pm$ 0.093  & 0.87 $\pm$ 0.049      & 0.65 $\pm$ 0.058      & 0.73 $\pm$ 0.164      & 0.76 $\pm$ 0.150      & 0.73 \\
			& ALL $(d_\theta=2)$         & \s{0.72} $\pm$ 0.054   & \s{0.79} $\pm$ 0.057  & \f{0.64} $\pm$ 0.099  & \f{0.89} $\pm$ 0.051  & \s{0.66} $\pm$ 0.056  & \s{0.76} $\pm$ 0.150  & \s{0.78} $\pm$ 0.139  & 0.75 \\
			\cline{2-10}
			& ALL $(d_\theta=2)$     & 0.72 $\pm$ 0.054          & 0.79 $\pm$ 0.057      & 0.64 $\pm$ 0.099      & 0.89 $\pm$ 0.051  & 0.66 $\pm$ 0.056      & 0.76 $\pm$ 0.150      & 0.78 $\pm$ 0.139      & 0.75 \\
			& ALL $(d_\theta=4)$     & \s{0.80} $\pm$ 0.004      & \s{0.87} $\pm$ 0.004  & \f{0.75} $\pm$ 0.050  & \f{0.94} $\pm$ 0.002  & \s{0.74} $\pm$ 0.026  & \s{0.89} $\pm$ 0.016  & \s{0.89} $\pm$ 0.009  & \s{0.84} \\
			& ALL $(d_\theta=8)$     & \f{0.82} $\pm$ 0.005      & \f{0.89} $\pm$ 0.005  & \s{0.70} $\pm$ 0.046  & \s{0.93} $\pm$ 0.002      & \f{0.79} $\pm$ 0.021  & \f{0.90} $\pm$ 0.011  & \f{0.90} $\pm$ 0.008  & \f{0.85} \\
            
		\hline
		\end{tabular}
	}
\end{table*}

\end{document}